\documentclass[twocolumn]{IEEEtran}
\usepackage{amsmath}
\usepackage{latexsym}
\usepackage{amssymb}
\usepackage{bm}
  \usepackage[pdftex]{graphicx}
\newtheorem{thm}    {Theorem}%[section]
\newtheorem{lem}     {Lemma}%[section]
\newtheorem{corollary}  {Corollary}%[section]
\newtheorem{proposition}        {Proposition}%[section]
\newtheorem{define}     {Definition}%[section]
\def\CON{\mathop{\rm CON}}
\def\BIN{\mathop{\rm BIN}}
\def\ACC{\mathop{\rm ACC}}
\def\REJ{\mathop{\rm REJ}}

\def\Supp{\mathop{\rm Supp}}
\def\RR{\mathbb R}
\def\bW{\bm{W}}
\def\bw{\bm{w}}
\def\bV{\bm{V}}

\def\argmax{\mathop{\rm argmax}}
\def\uni{\mathop{\rm uni}}

\def\argmin{\mathop{\rm argmin}}

\def\im{\mathop{\rm Im}}

\def\cX{{\cal X}}

\def\rE{{\rm E}}
\def\Pr{{\rm Pr}}

\newcommand{\qed}{\hfill \IEEEQED}

\newcommand{\sU}{\mathsf{U}}
\newcommand{\sK}{\mathsf{K}}
\newcommand{\sk}{\mathsf{k}}
\newcommand{\sm}{\mathsf{m}}
\newcommand{\bF}{\mathbb{F}}
\newcommand{\sM}{\mathsf{M}}
\newcommand{\sL}{\mathsf{L}}

\def\QED{\mbox{\rule[0pt]{1.5ex}{1.5ex}}}

\def\endproof{\hspace*{\fill}~\QED\par\endtrivlist\unskip}
 \newenvironment{proofof}[1]{\vspace*{5mm} \par \noindent
         \quad{\it Proof of #1:\hspace{2mm}}}{\qed
}

\def\Label#1{\label{#1}\ [\ #1\ ]\ }
\def\Label{\label}

\begin{document}
\title{Secure list decoding and its application to bit-string commitment
%\thanks{
%The material in this paper was presented in part at the IEEE International Symposium on Information Theory (ISIT2019),
%Paris, France, July 7 -- 12, 2019 \cite{MH17-1}.}
}
\author{
Masahito~Hayashi~\IEEEmembership{Fellow,~IEEE}\thanks{Masahito Hayashi is with 
Shenzhen Institute for Quantum Science and Engineering, Southern University of Science and Technology,
Shenzhen, 518055, China,
Guangdong Provincial Key Laboratory of Quantum Science and Engineering,
Southern University of Science and Technology, Shenzhen 518055, China,
and
Graduate School of Mathematics, Nagoya University, Nagoya, 464-8602, Japan.
(e-mail:hayashi@sustech.edu.cn, masahito@math.nagoya-u.ac.jp)}}
\date{}
\maketitle
\begin{abstract}
We propose a new concept of secure list decoding, which is related to bit-string commitment.
While the conventional list decoding requires that the list contains the transmitted message,
secure list decoding requires the following additional security conditions
to work as a modification of bit-string commitment.
The first additional security condition is
the receiver's uncertainty for the transmitted message, which is stronger than
the impossibility of the correct decoding,
even though the transmitted message is contained in the list.
The other additional security condition is the impossibility for the sender to estimate another element of the decoded list except for the transmitted message.
The first condition is evaluated by the equivocation rate.
The asymptotic property is evaluated by three parameters,
the rates of the message and list sizes, and the equivocation rate.
We derive the capacity region of this problem.
We show that the combination of hash function and secure list decoding
yields the conventional bit-string commitment.
Our results hold even when the input and output systems are general probability spaces including continuous systems.
When the input system is a general probability space,
we formulate the abilities of the honest sender and the dishonest sender in a different way.
\end{abstract}

\section{Introduction}\Label{S1}
\IEEEPARstart{R}{elaxing} the condition of the decoding process,
Elias \cite{Elias} and Wozencraft \cite{Wo} independently
introduced list decoding as the method to allow more than one element
as candidates of the message sent by the encoder at the decoder.
When one of these elements coincides with the true message, the decoding is regarded as successful.
The paper \cite{Guruswami}
discussed its algorithmic aspect.
In this formulation, Nishimura \cite{Nish} obtained 
the channel capacity by showing its strong converse part\footnote{The strong converse part is the argument that the average error goes to $1$ if the code has a transmission rate over the capacity.}.
That is, he showed that the transmission rate is less than the conventional capacity plus 
the rate of the list size, i.e., the number of list elements.
Then, the reliable transmission rate does not increase
even when list decoding is allowed
if the list size does not increase exponentially.
%The achievability of this bound has been proved only when the number of list is not exponentially increasing.
In the non-exponential case, these results were generalized by Ahlswede \cite{Ah}.
Further, the paper \cite{Haya} showed that 
the upper bound of capacity by Nishimura can be attained even if the list size increases exponentially.
When the number of lists is $\sL$,
the capacity can be achieved by choosing the same codeword for $\sL$ distinct messages.

However, the merit of the increase in the list size
was not discussed sufficiently.
To get a merit of list coding, we need a code construction that is essentially different from conventional coding.
Since the above capacity-achieving code construction does not have an essential difference from the conventional coding,
%to find an essential difference from the conventional coding,
we need to rule out the above type of construction of list coding.
That is, to extract a merit of list decoding, we need additional parameters to characterize
the difference from the conventional code construction, which can be expected to rule out such a trivial construction.

To seek a merit of list decoding, we focus on bit commitment, which is a fundamental task in information security.
It is known that bit commitment can be realized when a noisy channel is available \cite{CC}.
Winter et al \cite{BC1,BC2} studied bit-string commitment, the bit string version of bit commitment
when an unlimited bidirectional noiseless channel is available between Alice and Bob, and
a discrete memoryless noisy channel $W :{\cal X}\to {\cal Y}$ 
from Alice to Bob, which may be used $n$ times.
They derived the asymptotically optimal rate as $n$ goes to infinity, which is called the commitment capacity.
Since their result is based on Shannon theory,
the tightness of their result shows the strong advantage of Shannon theoretic approach to 
information theoretic security.
This result was extended to the formulation with multiplex coding \cite{BC3}.
However, their optimal method has the following problems;
\begin{description}
\item[(P1)]
When the number of use of the channel is limited,
it is impossible to send a message with a larger rate than the commitment capacity.

\item[(P2)]
Their protocol assumes that the output system ${\cal Y}$ is a finite set
because they employ the method of type.
However, when a noisy channel is realized by wireless communication, like
an additive white Gaussian noise (AWGN) channel,
the output system ${\cal Y}$ is a continuous set.
\end{description}

%Since the problem (P1) cannot be resolved under the same problem setting as Winter et al \cite{},
To resolve the problem (P1), 
it is natural to relax the condition for bit-string commitment.
Winter et al \cite{BC1,BC2} imposed strong security for the concealing condition.
However, studies in information theory, in particular, papers for wire-tap channel,
often employs equivocation rate instead of strong security.
In this paper, to relax the condition of bit-string commitment by using equivocation rate,
we consider the following simple protocol by employing list decoding,
%To explain this protocol, we consider the following bit-string commitment protocol,
where Alice wants to send her message $M \in \{1, \ldots,\sM\}$ to Bob.
\begin{description}
\item[(i)] (Commit Phase) Alice sends her message $M$ to Bob via a noisy channel. 
Bob outputs $\sL$ messages as the list. 
The list is required to contain the message $M$. 

\item[(ii)] (Reveal Phase) 
Alice sends her message $M$ to Bob via a noiseless channel.
If $M$ is contained in Bob's decoded list, Bob accepts it.
Otherwise, Bob rejects it.
%Here, the remaining $L-1$ products will be resolved into their component elements and used for the resources for the next products.
\end{description}
In order that the protocol with phase (i) and (ii) works for bit-string commitment,
the following requirements need to be satisfied.
\begin{description}
\item[(a)] 
The message $M$ needs to be one of $\sL$ messages $M_1, \ldots, M_{\sL} $ output by Bob. 
\item[(b)] 
Bob cannot identify the message $M$ at the phase (i). 
\item[(c)] 
Alice cannot find another element among 
$\sL$ messages $M_1, \ldots, M_{\sL} $ output by Bob.
\end{description}
The requirement (a) is the condition for the requirement for the conventional list decoding
while the requirements (b) and (c) 
correspond to the concealing condition and the binding condition, respectively and
have not been considered in the conventional list decoding.

%to find an essential difference from the conventional coding,
In this paper, we propose a new concept of secure list decoding
%by adding the following two additional conditions to the list decoding
by adding the requirements (b) and (c).
One typical condition for (b) is the conventional equivocation rate based on
the conditional entropy.
In this paper, we also consider the equivocation rate based on
the conditional R\'{e}nyi entropy similar to the paper \cite{YH,YT}\footnote{
While the conference paper \cite{MH17-1} discussed a similar modification of list decoding,
it did not consider the equivocation rate.
In this sense, the content of this paper is different from that of \cite{MH17-1}.}.
Hence, our code can be evaluated by three parameters.
The first one is the rate of the message size, the second one is the rate of list size,
and the third one is the equivocation rate.
Using three parameters, we define the capacity region. 
In addition, our method works with a general output system including a continuous output system, which resolves the problem (P2) while an extension to such a general case was mentioned as an open problem in \cite{BC2}. 

In the second step, we extend our result to the case with a general input system including a continuous input system.
We need to be careful in formulating the problem setting in this case.
If Alice is allowed to access infinitely many input elements in a continuous input system,
the conditional entropy rate $H(X|Y)$ might be infinity.
Further, it is not realistic for Alice to access infinitely many input elements.
because a realistic modulator converts messages to finite constellation points in 
a continuous input system in wireless communication.
Therefore, we need to separately formulate honest Alice and dishonest Alice as follows.
The honest Alice is assumed to access only a fixed finite subset of a general input system.
But, the dishonest Alice is assumed to access all elements of the general input system.
Under this problem setting, we derived the capacity region.

In the third step, 
we propose a conversion method to make a protocol for bit-string commitment with strong security as 
the concealing condition (b)
by converting a secure list decoding code.
In this converted protocol, 
the security parameter for the concealing condition (b) is evaluated by variational distance in the same way as 
Winter et al \cite{BC1,BC2}.
In particular, this converted protocol has strong security even with continuous input and output systems,
where the honest Alice and the dishonest Alice has different accessibility to the continuous input system.
In this converted protocol, 
the rate of message size of the bit-string commitment is the same as the
equivocation rate based on the conditional entropy of the original secure list decoding code,
which shows the merit of the equivocation rate of a secure list decoding code.
In fact, the bit-string commitment with the continuous case was treated as an open problem in the preceding studies \cite{BC2}.
In addition, due to the above second step, our protocol for bit-string commitment
works even when the accessible alphabet by the honest Alice is different from 
the accessible alphabet by the dishonest Alice.

This paper is structured as follows. 
Section \ref{SCom} reviews the existing results for bit-string commitment. 
Section \ref{S51T} explains how we mathematically handle a general probability space as input and output systems including continuous systems.
Section \ref{S2} gives the formulation of secure list decoding.
Section \ref{S51II} introduces information quantities used in our main results.
%Section \ref{S51} prepares several information quantities.
Section \ref{S51} states our results for secure list decoding with a discrete input system.
Section \ref{S7} explains our formulation of secure list decoding with general input system
and states our results under this setting.
Section \ref{SecBS} presents the application of secure list decoding to the 
bit-string commitment with strong security.
Section \ref{S-C} shows the converse part, and
Section \ref{S-K} proves the direct part.

\section{Review of existing results for bit-string commitment}\Label{SCom}
Before stating our result,
we review existing results for bit-string commitment \cite{BC1,BC2}.
Throughout this paper, the base of the logarithm is chosen to be $2$.
Also, we employ the standard notation for
probability theory, in which, upper case letters denote
random variables and the corresponding lower case letters
denote their realizations.
Bit-string commitment has two security parameters, 
the concealing parameter $\delta_{\CON}>0$ and the binding parameter $\delta_{\BIN}>0$.
We denote the message revealed by Alice in Reveal Phase by $\hat{M}$.
Let $Z_1$ be all information that Bob obtains during Commit Phase,
and $Z_2$ be all information that Bob obtains during Reveal Phase except for $\hat{M}$. 
Here, $Z_1$ contains the information generated by Bob.
After Reveal Phase, Bob makes his decision, $\ACC$ (accept) or $\REJ$ (rejection).
For this decision, Bob has a function 
$\beta(Z_1,Z_2,\hat{M})$ that takes the value $\ACC$ or $\REJ$.
When Alice intends to send message $M$ in ${\cal M}$ to Bob,
the concealing and binding conditions are given as follows.
\begin{description}
\item[(CON)]
~Concealing condition with $\delta_{\CON}>0$.
When Alice is honest,
the inequality
\begin{align}
\frac{1}{2}\|P_{Z_1|M=m}-P_{Z_1|M=m'}\|_1\le \delta_{\CON}
\end{align}
holds for $m\neq m' \in {\cal M}$.
\item[(BIN)]
~Binding condition with $\delta_{\BIN}>0$.
We assume that the message $M$ is subject to the uniform distribution on ${\cal M}$.
When Alice and Bob are honest, %and $M=\hat{M}$,
\begin{align}
\Pr ( \beta(Z_1,Z_2,{M}) = \ACC)
\ge 1-\delta_{\BIN}.\Label{BIN1}
\end{align}
When Bob is honest,
the inequality 
\begin{align}
&\Pr ( \beta(Z_1,z_2,m) = \ACC,
\beta(Z_1,z_2',m') = \ACC)\nonumber \\
\le & \delta_{\BIN}\Label{BIN2}
\end{align}
holds for $m\neq m' \in {\cal M}$ and $z_2,z_2'$.
\end{description}

When the protocol with (i) and (ii) is used for bit-string commitment,
the conditions (a) and (c) guarantee \eqref{BIN1} and \eqref{BIN2} of (BIN), respectively,
and the condition (b) guarantees (CON).

Now, we denote a noisy channel $\bW$ from a finite set ${\cal X}$ to a finite set ${\cal Y}$
by using a set $\{W_x\}_{x \in {\cal X}}$ of distributions on ${\cal Y}$.
Winter et al \cite{BC1,BC2} considered the situation that 
Alice and Bob use the channel $\bW$ at $n$ times and the noiseless channel can be used freely.
Winter et al \cite{BC1,BC2} defined the commitment capacity as the maximum rate when the code satisfies 
Concealing condition with $\delta_{\CON,n}$ and 
Binding condition with $\delta_{\BIN,n}$ under the condition that the parameters $\delta_{\CON,n}$ and 
$\delta_{\BIN,n}$ approach to zero as $n$ goes to infinity.
They derived the commitment capacity under 
%To state their result, they introduced 
the following conditions for the channel $\bW$;
\begin{description}
\item[(W1)]
${\cal X}$ and ${\cal Y}$ are finite sets.
\item[(W2)]
For any $x \in {\cal X}$, the relation
\begin{align}
\min_{x \in {\cal X}}
\min_{P \in {\cal P}({\cal X}\setminus \{x\})}
D\bigg( \sum_{x' \in {\cal X}\setminus \{x\}}P(x')W_{x'} \bigg\|W_x\bigg) 
%\min_{P \in {\cal P}({\cal X}\setminus \{x\})}
%d_v( W_x, \sum_{x' \in{\cal X}\setminus \{x\} } P(x')W_{x'}) 
>0 
\end{align}
holds, where $D(P\|Q)$ is the Kullback-Leibler divergence between two distributions $P$ and $Q$.
%variational distance.
%no distribution $P\in {\cal P}({\cal X}\setminus \{x\})$ 
%satisfies $W_x= \sum_{x' \in{\cal X}\setminus \{x\} } P(x')W_{x'}$.
This condition is called the non-redundant condition.
\end{description}

To state their result, we introduce a notation;
Given a joint distribution $P_{X,Y}$ on a discrete set ${\cal X}\times {\cal Y}$,
we denote the conditional distribution $P_{X|Y=y}$ under the condition that $Y=y$.
Then, the conditional entropy $H(X|Y)$ is given as
\begin{align}
H(X|Y)_{P_{X,Y}}&:=\sum_{y \in {\cal Y}}P_Y(y) H(P_{X|Y=y}),\\
H( P_{X|Y=y}) &:=
-\sum_{x \in {\cal X}} P_{X|Y=y}(x) \log P_{X|Y=y}(x).
\end{align}
When the joint distribution $P_{X,Y}$ is given as
$P_{X,Y}(x,y)=W_x(y)P(x)$ by using a distribution $P \in {\cal P}({\cal X})$,
we denote 
the conditional entropy $H(X|Y)_{P_{X,Y}}$ by $H(X|Y)_P$.
They showed the following proposition;
\begin{proposition}[\protect{\cite[Theorem 2]{BC1}, \cite{BC2}}]\Label{Pro1}
When the channel $\bW$ satisfies Conditions (W1) and (W2),
the commitment capacity is given as 
\begin{align}
\sup_{P \in {\cal P}({\cal X})} H(X|Y)_P.
\Label{WIN}
\end{align}\hfill $\square$
\end{proposition}

Many noisy channels are physically realized by wireless communication,
and such channels have continuous output system ${\cal Y}$.
%Hence, Condition (W4) is crucial in practical application.
Indeed, if we apply discretization to a continuous output system ${\cal Y}$, 
we obtain a discrete output system ${\cal Y}'$.
When we apply their result to the channel with the discrete output system ${\cal Y}'$,
the obtained protocol satisfies Condition (BIN) even when Bob uses the continuous output system ${\cal Y}$.
However, the obtained protocol does not satisfy Condition (CON) in general
under the continuous output system ${\cal Y}$.

In fact, Condition (W2) can be removed and 
Proposition \ref{Pro1} can be generalized as follows.
Therefore, Condition (W2) can be considered as an assumption for simplifying our analysis.
\begin{proposition}\Label{Pro2}
Assume that the channel $\bW$ satisfies Condition (W1).
We define ${\cal X}_0\subset {\cal X}$ as
\begin{align}
{\cal X}_0:= \argmin_{{\cal X}' \subset {\cal X}}
\Big\{ |{\cal X}'|~\Big|
CH \{ W_x\}_{x \in {\cal X}'}=CH \{ W_x\}_{x \in {\cal X}}
\Big\},
\end{align}
where $CH {\cal S}$ expresses the convex hull of a set ${\cal S}$.
Then, the commitment capacity is given as 
\begin{align}
\sup_{P \in {\cal P}({\cal X}_0)} H(X|Y)_P.
\Label{WIN5}
\end{align}\hfill $\square$
\end{proposition}

Proposition \ref{Pro2} follows from 
Proposition \ref{Pro1} in the following way. 
Due to Condition (W1), 
the channel $\bW$ with input alphabet ${\cal X}_0$
satisfies Condition (W2) as well as (W1).
Hence, the commitment capacity is lower bounded by \eqref{WIN5}.
Since any operation with 
the channel $\bW$ with input alphabet ${\cal X}$
can be simulated with ${\cal X}_0$.
Therefore, the commitment capacity is upper bounded by \eqref{WIN5}.
Thus, we obtain Proposition \ref{Pro2}.

\section{Various types of conditional entropies with general probability space}\Label{S51T}
We focus on an input alphabet ${\cal X}$ with finite cardinality, and 
denote the set of probability distributions on ${\cal X}$ by ${\cal P}({\cal X})$.
But, an output alphabet ${\cal Y}$ may have infinite cardinality and is a general measurable set.
In this paper, the output alphabet ${\cal Y}$ is treated as a general probability space with a measure $\mu(dy)$
because this description covers the probability space of finite elements and the set of real values.
Hence, when the alphabet ${\cal Y}$ is a discrete set including a finite set,
the measure $\mu(dy)$ is chosen to be the counting measure.
%That is, in the discrete case, the pair of $\int_{{\cal Y}}$ and $\mu(dy)$ can be replaced by 
%$\sum_{y \in {\cal Y}}$, which produces the conventional notation in the discrete case.
When the alphabet ${\cal Y}$ is a vector space over the real numbers $\RR$,
the measure $\mu(dy)$ is chosen to be the Lebesgue measure.
Throughout this paper,
we will use an upper case letter and corresponding lower case letter to stand for a probability measure and its density function.
When we treat a probability distribution $P$ on the alphabet ${\cal Y}$,
it is restricted to a distribution absolutely continuous with respect to $\mu(dy)$. 
In the following, we use the lower case $p(y)$ to express the Radon-Nikodym derivative of $P$ with respect to the measure $\mu(dy)$, i.e., the probability density function of $P$ 
so that $P(dy)=p(y) \mu(dy)$.
%That is, we identify the probability measure with the Radon-Nikodym derivative with respect to the measure $\mu(dy)$.
This kind of channel description covers many useful channels.
For example, phase-shift keying (PSK) scheme of additive white Gaussian noise (AWGN) channels satisfies this condition.
In addition, 
the capacity of AWGN channel with the energy constraint
can be approximately achieved
when the input alphabet for encoding is restricted to a finite subset of the set of real numbers.

For a distribution $P$ on ${\cal Y}$ and a general measure $Q$ on ${\cal Y}$,
we define the Kullback–Leibler (KL) divergence
$D(P\|Q):= \mathbb{E}_{P}[\log \frac{p(Y)}{q(Y)}]$
and R\'{e}nyi divergence of order $\alpha (\neq 1) >0$
$D_\alpha(P\|Q):= \frac{1}{\alpha-1}\log 
\mathbb{E}_{P}[(\frac{p(Y)}{q(Y)})^{\alpha-1}]$.

When ${\cal M}$ is a finite set and ${\cal Y}$ is a general probability space,
the conditional entropy is defined as
\begin{align}
H(M|Y):= \int_{{\cal Y}} H(P_{M|Y=y}) p(y)\mu (dy).
\end{align}
This quantity can be written as
\begin{align}
&H(M|Y)= -D(P_{MY}\|I_{M} \times P_Y) \nonumber \\
=&\max_{Q \in {\cal P}({\cal Y})} -D(P_{MY}\|I_{M} \times Q) ,
\end{align}
where $I_{M}$ is defined as $I_{M}(m)=1$.
We focus on the following type of R\'{e}nyi conditional entropy
$H_\alpha(M|Y)$ as \cite{Arimoto,Iwamoto,epsilon}
%and $H_\alpha^{\uparrow}(M|Y)$ as
\begin{align}
%H_\alpha(M|Y)&:= -D_\alpha(P_{MY}\|I_{M} \times P_Y) \\
H_\alpha(M|Y)&:=
\max_{Q \in {\cal P}({\cal Y})} -D_\alpha(P_{MY}\|I_{M} \times Q) .
\end{align}
$H_\alpha(M|Y)$ is monotonically
decreasing for $\alpha$ \cite[Lemma 7]{epsilon}.
Hence, we have
$H(M|Y)\ge H_\alpha(M|Y)$ for $\alpha>1$. %,H_\alpha^{\uparrow}(M|Y)$. 
It is known that the maximum is attained by 
$q_\alpha(y):=
\frac{(\sum_m p_{MY}(m,y)^{\alpha} )^{1/\alpha}}{
\int_{{\cal Y}} (\sum_m p_{MY}(m,y)^{\alpha} )^{1/\alpha} \mu(dy) }$
\cite[Lemma 4]{epsilon}.
Hence, when two pairs of variables $(M_1,Y_1)$ and $(M_2,Y_2)$ are independent, 
we have the additivity;
\begin{align}
H_\alpha(M_1M_2|Y_1Y_2)= H_\alpha(M_1|Y_1)+H_\alpha(M_2|Y_2).\Label{AOR}
\end{align}

\section{Problem setting}\Label{S2}
\subsection{Our problem setting without explicit description of coding structure}\Label{S2-1}
To realize the requirements (a), (b), and (c) mentioned in Section \ref{S1}, 
we formulate the mathematical conditions for the protocol for 
a given channel $\bW$ from the discrete system ${\cal X}$ to the other system ${\cal Y}$
with integers $\sL < \sM$ and security parameters $\epsilon_A,\delta_C,\delta_D$.
In the asymptotic regime, i.e., the case when 
the channel $\bW$ is used $n$ times and $n$ goes to infinity,
the integers $\sL$ and $ \sM$ go to infinity, which realizes 
the situation that the security parameters $\epsilon_A,\delta_C,$ and $\delta_D$
approach to zero.
Hence, when $\sL$ and $\sM$ is fixed,
%no code $(\phi,\Psi)$ works with 
the security parameters cannot be chosen to be arbitrarily small.
In the following, we describe the condition in an intuitive form in the first step. Later, we 
transform it into a coding-theoretic form
because the coding-theoretic form matches the theoretical discussion including the proofs of our main results.

Alice sends her message $M \in {\cal M}:= \{1, \ldots, \sM\}$ via a noisy channel 
with an encoder $\phi$, which is a map from ${\cal M} $ to ${\cal X}$.
Bob outputs the $\sL$ messages $M_1, \ldots M_{\sL}$.
The decoder is given as the following $\Psi$;
For $y\in {\cal Y}$, we choose a subset $\Psi(y) \subset {\cal M}$ with $|\Psi(y)|= \sL$.

Then, we impose the following conditions
for an encoder $\phi$ and a decoder $\Psi$.
\begin{description}
\item[(A)] 
Verifiable condition with $\epsilon_A>0$. Any element $m \in {\cal M}$ satisfies
\begin{align}
    \Pr [m \notin \Psi(Y) | X = \phi(m)] \le \epsilon_A.
    \end{align}
\item[(B)]
Equivocation version of concealing condition with $r>0$.
The inequality
\begin{align}
H(M|Y) \ge r \Label{HI1}
\end{align}
holds.
\item[(C)]
Binding condition for honest Alice with $\delta_C>0$.
Any distinct pair $m'\neq m$ satisfies
\begin{align}
\Pr [m' \in \Psi(Y) | X=\phi(m)] \le \delta_C . 
\end{align}
\end{description}

Now, we discuss how the code $(\phi,\Psi)$ can be used for the task explained in Section \ref{S1}.
Assume that Alice sends her message $M$ to Bob 
by using the encoder $\phi$ via noisy channel $\bW$
and Bob gets the list $M_1, \ldots, M_{\sL}$ by applying the decoder $\Psi$ at Step (i).
At Step (ii), Alice sends her message $M$ to Bob via a noiseless channel.
Verifiable condition (A) guarantees that her message $M$ belongs to Bob's list. 
Hence, the requirement (a) is satisfied.
Equivocation version of concealing condition (B) forbids Bob to identify Alice's message at Step (i), hence it guarantees the requirement (b).
In the asymptotic setting, this condition is weaker than Concealing condition (CON)
when $\delta_{\CON}$ goes to zero and $r$ is smaller than $\log \sM$.
Hence, this relaxation enables us to exceed the rate \eqref{WIN} derived by \cite{BC1,BC2}.
This type of relaxation is often used in wire-tap channel \cite{CK79}.

In fact, if $m$ is Alice's message and there exists another element $m' (\neq m) \in {\cal M}$ %and an element  $x_0 $ that might be different from $\phi(m)$ 
such that
$\Pr [m \in \Psi(Y) | X=\phi(m)] $ and 
$\Pr [m' \in \Psi(Y) | X=\phi(m)] $
%$\delta_{C,m}(x_0,\Psi)$ 
are close to $1$,
Alice can make the following cheating as follows;
She sends $m'$ instead of $m$ at the phase (ii).
Since Condition (C) forbids Alice such cheating, 
it guarantees the requirement (c). Hence, it can be considered as the binding condition for honest  Alice.
Further, Bob is allowed to decode less than $\sL$ messages.
That is, $\sL$ is the maximum number that Bob can list as the candidates of the original message.
However, Condition (C) assumes honest Alice who uses the correct encoder $\phi$.
Dishonest Alice can send an element  $x_0 $ different from $\phi(m)$
such that $\Pr [m \in \Psi(Y) | X=x_0] $ and 
$\Pr [m' \in \Psi(Y) | X=x_0] $ are close to $1$.
%her message by using a different encoder.
To cover such a case, we impose the following condition instead of Condition (C).

\begin{description}
\item[(D)] Binding condition for dishonest Alice with $\delta_D>0$.
For $x\in {\cal X}$, we define the quantity $\delta(x,\Psi)$ as the second largest value among
$\{ \Pr [m \in \Psi(Y) | X=x] \}_{m=1}^{\sM}$.
Then, any $x\in {\cal X}$ satisfies
\begin{align}
\delta(x,\Psi) \le \delta_D .  \Label{HB3}
\end{align}
\end{description}

In fact, Condition (D) implies that
\begin{align}
\Pr [m' ,m \in \Psi(Y) | X=x] \le \delta_D. 
\Label{XPA}
\end{align}
Eq. \eqref{XPA} can be shown by contradiction due to the following relation;
\begin{align}
&\Pr [m' ,m \in \Psi(Y) | X=x]  \nonumber \\
\le &\min (\Pr [m \in \Psi(Y) | X=x],\Pr [m' \in \Psi(Y) | X=x])
 \nonumber \\
 \le & \delta(x,\Psi).
\end{align}

The difference between Conditions (C) and (D) are summarized as follows.
Condition (C) expresses the possibility 
that Alice makes cheating in the reveal phase while she behaves honestly in the commit phase.
Condition (D) expresses the possibility 
that Alice makes cheating in the reveal phase 
when she behaves dishonestly even in the commit phase.
Hence, it can be considered as the binding condition for
dishonest Alice.
Therefore, while 
the case with honest Alice and honest Bob is summarized in Fig. \ref{F-honest},
the case with dishonest Alice and honest Bob is summarized in Fig. \ref{F-dishonest}.

We consider another possibility for requirement (b)
by replacing the conditional entropy by the conditional R\'{e}nyi entropy of order $\alpha>1$.
\begin{description}
\item[(B$\alpha$)]
R\'{e}nyi equivocation type of concealing condition
%non-decodable condition 
of order $\alpha>1$ with $r$.
The inequality
\begin{align}
H_\alpha(M|Y) \ge r\Label{HI1-1}
\end{align}
holds.
\end{description}

Now, we observe how to characterize the code constructed to achieve the capacity in the paper \cite{Haya}.
For this characterization, we consider the following code when $\sM' \sL=\sM$.
We divide the $\sM$ messages into $\sM'$ groups whose group is composed of $\sL$ messages.
First, we prepare a code $({\phi}',{\psi}')$ to transmit the message with size $\sM'$ with a decoding error probability $\epsilon_A'$,
where ${\phi}'$ is the encoder and ${\psi}'$ is the decoder.
When the message $M$ belongs to the $i$-th group,
Alice sends ${\phi}'(i)$.
Using the decoder ${\psi}'$, Bob recovers $i'$.
Then, Bob outputs $\sL$ elements that belongs to the $i'$-th group.
In this code, the parameter $H(M|Y)$ is given as $\log \sL$. Hence, it satisfies condition (B) with a good parameter.
However, the parameters $\delta_C$ and $\delta_D$ become at least $1-\epsilon_A'$.
Hence, this protocol essentially does not satisfy Biding condition (C) nor (D).
In this way, our security parameter rules out the above trivial code construction.

\begin{figure}[t]
%\centering
%\includegraphics[scale=0.4]{MHRepeater.png}
\begin{center}
  \includegraphics[width=0.98\linewidth]{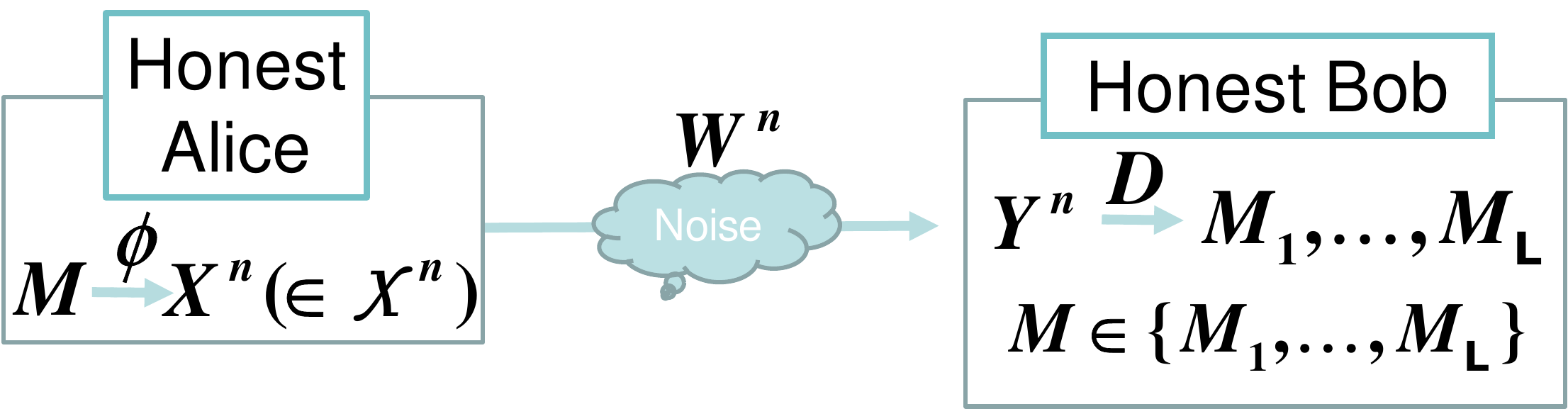}
  \end{center}
\caption{Case with honest Alice and honest Bob.
The set of Bob's decoded messages contains Alice's message $M$.
Alice cannot infer other decoded messages.}
\Label{F-honest}
\end{figure}   

\begin{figure}[t]
%\centering
%\includegraphics[scale=0.4]{MHRepeater.png}
\begin{center}
  \includegraphics[width=0.98\linewidth]{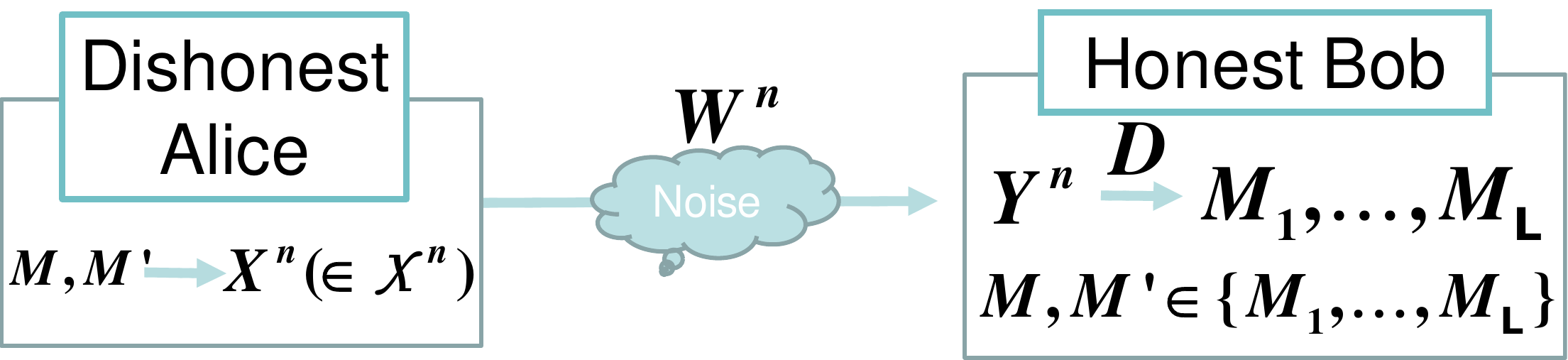}
  \end{center}
\caption{Case with dishonest Alice and honest Bob.
Dishonest Alice chooses $X^n \in {\cal X}^n$ such that
she infers at least two elements in the set of Bob's decoded messages.
Condition (D) guarantees the non-existence of such an element $X^n \in {\cal X}^n$.}
\Label{F-dishonest}
\end{figure}   

\subsection{Our setting with coding-theoretic description}\Label{S2-1B}
To rewrite the above conditions in a coding-theoretic way, 
we introduce several notations.
For $x\in {\cal X}$ and a distribution on ${\cal X}$, we define the distribution $W_x$ and $W_{P}$ on ${\cal Y}$
as $W_x(y):=W(y|x)$ and $W_P(y):= \sum_{x\in {\cal X}}P(x)W(y|x)$.
Alice sends her message $M \in {\cal M}:= \{1, \ldots, \sM\}$ via noisy channel 
$\bW$ with a code $\phi$, which is a map from ${\cal M} $ to ${\cal X}$.
Bob' decoder is described as disjoint subsets 
$D=\{{\cal D}_{m_1, \ldots, m_{\sL}}\}_{\{m_1, \ldots, m_{\sL}\} \subset {\cal M}}$
such that $\cup_{\{m_1, \ldots, m_{\sL}\} \subset {\cal M}} {\cal D}_{m_1, \ldots, m_{\sL}}={\cal Y}$.
That is, %in these two descriptions, 
we have the relation $
{\cal D}_{m_1, \ldots, m_{\sL}}=\{y| \{m_1, \ldots, m_{\sL}\}= \Psi(y)\}$.
In the following, we denote our decoder by $D$ instead of $\Psi$.

In particular, 
when a decoder 
has only one outcome as an element of ${\cal M}$ 
it is called a single-element decoder.
It is given as disjoint subsets 
$\tilde{{\cal D}}=\{ \tilde{{\cal D}}_{m}\}_{m \in {\cal M}}$
such that $\cup_{m \in {\cal M} } \tilde{{\cal D}}_{m}={\cal Y}$.
Here, remember that Winter et al \cite{BC1,BC2} assumes 
the uniform distribution on ${\cal M}$ for the message $M$ in Binding condition.

\begin{thm}
When the message $M$ is subject to the uniform distribution on ${\cal M}$ 
in a similar way to Winter et al \cite{BC1,BC2},
the conditions (A) -- (D) for an encoder $\phi$ 
and a decoder $D=\{{\cal D}_{m_1, \ldots, m_{\sL}}\}_{\{m_1, \ldots, m_{\sL}\} \subset {\cal M}}$
are rewritten in a coding-theoretic way as follows.
\begin{description}
\item[(A)] 
Verifiable condition.
\begin{align}
\epsilon_{A}(\phi,D)
:=&\max_{m \in {\cal M}}\epsilon_{A,m}(\phi(m),D) 
\le  \epsilon_A \\
\epsilon_{A,m}(x,D) 
:=&
1-\sum_{
m_1, \ldots, m_{\sL}
}
W_{x}({\cal D}_{m_1, \ldots, m_{\sL}}), %\nonumber \\
\end{align}
where the above sum is taken under the condition
$m \in \{m_1, \ldots, m_{\sL}\}$.

\item[(B)]
Equivocation version of concealing condition with $r>0$.
\begin{align}
E(\phi):=& \log \sM -\min_{Q \in {\cal P}({\cal Y})}
\sum_{m=1}^{\sM}\frac{1}{\sM} D(W_{\phi(m)} \|  Q)\nonumber \\
\ge & r.\Label{HI2}
\end{align}

\item[(B$\alpha$)]
R\'{e}nyi equivocation type of concealing condition of order $\alpha>1$ with $r$.
\begin{align}
&E_\alpha(\phi)\nonumber \\
:=&\log \sM \nonumber \\
&- \min_{Q \in {\cal P}({\cal Y})}
\frac{1}{\alpha-1}\log 
\sum_{m=1}^{\sM}\frac{1}{\sM} 2^{(\alpha-1)D_\alpha(W_{\phi(m)} \|  Q)}
\nonumber \\
\ge &r.\Label{HI3}
\end{align}

\item[(C)]
Binding condition for honest Alice.
%There is no map $\psi$ from ${\cal M}$ to ${\cal M}$ such that
%$\psi(m)\neq m$ for $m \in {\cal M}$
%and
\begin{align}
&\delta_{C}(\phi,D)
:=
\max_{m \in {\cal M}}
\delta_{C,m}(\phi(m),D) \le  \delta_C 
 %\nonumber \\
\\
&\delta_{C,m}(x,D)  \nonumber \\
&:=\max_{m' (\neq m) \in {\cal M}}
\sum_{ m_1, \ldots, m_{\sL}}
W_{x}({\cal D}_{m_1, \ldots, m_{\sL}}) ,
\Label{EE8}
\end{align}
where the above sum is taken under the condition 
$m' \in \{m_1, \ldots, m_{\sL}\}$.
\item[(D)] Binding condition for dishonest Alice.
For $x\in {\cal X}$, we define the quantity
$\delta_{D,x}(D)$ as the second largest value among 
$\{(1-\epsilon_{A,m}(x,D))\}_{m=1}^{\sM}$.
Then, the relation
\begin{align}
\delta_{D}(D)&:=
\max_{x \in {\cal X}} \delta_{D,x}(D)
\le  \delta_D \Label{HB3E}
\end{align}
holds.
\end{description}
\hfill $\square$\end{thm}

\begin{proof}
For any $m\in {\cal M}$ and $y \in {\cal Y}$,
the condition $m \in \Psi(y)$ is equivalent to the condition $y \in \cup _{ m_1, \ldots, m_{\sL}:\{m_1, \ldots, m_{\sL}\} \ni m'} {\cal D}_{m_1, \ldots, m_{\sL}}$.
Since 
\begin{align*}
&\sum_{ m_1, \ldots, m_{\sL}:\{m_1, \ldots, m_{\sL}\} \ni m }
W_{x}({\cal D}_{m_1, \ldots, m_{\sL}}) \\
=& W_{x}\Big(\bigcup _{ m_1, \ldots, m_{\sL}:\{m_1, \ldots, m_{\sL}\} \ni m} {\cal D}_{m_1, \ldots, m_{\sL}}
\Big),
\end{align*}
we obtain the equivalence between the conditions (A) and (C) given in Section \ref{S2-1} and those given here. 
In a similar way, the condition \eqref{HB3} is equivalent to the condition \eqref{HB3E}, which implies the desired equivalence with respect to the condition (D).
Since $M$ is subject to the uniform distribution, \eqref{HI1} and \eqref{HI1-1} are
equivalent to \eqref{HI2} and \eqref{HI3}. In fact, 
since $\min_{Q \in {\cal P}({\cal Y})}
\sum_{m=1}^{\sM}\frac{1}{\sM} D(W_{\phi(m)} \|  Q)$\par
\noindent$= \sum_{m=1}^{\sM}\frac{1}{\sM} D(W_{\phi(m)} \| 
\sum_{m=1}^{\sM}\frac{1}{\sM}W_{\phi(m)})
=I(M;Y)$,
$E(\phi)$ is calculated as $H(M)-I(M;Y)=H(M|Y)$,
and $E_\alpha(\phi)$ is calculated as
\begin{align}
&2^{-(\alpha-1)E_\alpha(\phi)}\nonumber \\
= &
 \min_{Q \in {\cal P}({\cal Y})}
\Big( \frac{1}{\sM}\Big)^{\alpha-1} \sum_{m=1}^{\sM}\frac{1}{\sM} 2^{(\alpha-1)D_\alpha(W_{\phi(m)} \|  Q)}\nonumber \\
= &
 \min_{Q \in {\cal P}({\cal Y})}
 \sum_{m=1}^{\sM}\frac{1}{\sM} 
 \int_{{\cal Y}} 
\Big( \frac{\frac{w_{\phi(m)}(y)}{\sM}}{q(y)}\Big)^{\alpha-1}
 w_{\phi(m)}(y) \mu(dy) \nonumber \\
 =&
 \min_{Q \in {\cal P}({\cal Y})}
2^{(\alpha-1) D_\alpha (P_{MY}\|I_M \times Q )}  \nonumber \\
 =& 
 2^{-(\alpha-1) \max_{Q \in {\cal P}({\cal Y})} (-D_\alpha (P_{MY}\|I_M \times Q ))} \nonumber \\
 =&
2^{-(\alpha-1) H_\alpha(M|Y) } .
\end{align}
Hence,
%As the condition $m=\psi(y)$ is equivalent to the condition $ y \in \tilde{{\cal D}}_m$,
we obtain the desired equivalence for the conditions (B) and (B$\alpha$).
\end{proof}

In the following, 
when a code $(\phi,D) $ satisfies conditions (A), (B) and (D), 
it is called an $(\epsilon_A,r,\delta_D)$ code. Also, 
for a code $(\phi,D) $,
we denote $\sM$ and $\sL$
by $|(\phi,D)|_1$ and $|(\phi,D)|_2$.
Also, we allow a stochastic encoder, in which $\phi(m)$ is a distribution $P_m$ on ${\cal X}$.
In this case, for a function $f$ from ${\cal X}$ to $\mathbb{R}$,
$f(\phi(m))$ expresses $\sum_{x}f(x)P_m(x)$.

%\section{General channel}
%\Label{S3}

\section{Information quantities and regions with general probability space}\Label{S51II}
\subsection{Information quantities}
Section \ref{S51T} introduced 
various types of conditional entropies with general probability space.
This section introduces other types of information quantities with general probability space.
In general, a channel from ${\cal X}$ to ${\cal Y}$ is described as a collection $\bW$ of conditional
probability measures $W_{x}$ on ${\cal Y}$ for all inputs $x \in {\cal X}$.
Then, we impose the above assumption to $W_{x}$ for any $x \in {\cal X}$.
Hence, we have $W_x(dy)=w_x(y)\mu(dy)$.
We denote the conditional probability density function by $\bw=(w_x)_{x\in {\cal X}}$.
When a distribution on ${\cal X}$ is given by a probability distribution $P\in {\cal P}({\cal X})$,
and a conditional distribution on a set ${\cal Y}$ with the condition on ${\cal X}$ is given by $\bV$,
we define the joint distribution $\bW \times P$
on ${\cal X} \times {\cal Y}$ by $\bW \times P(B,x):=W(B|x)P(x)$, 
and the distribution $\bW \cdot P$ on ${\cal Y}$ by $\bW \cdot P(B):=\sum_x W(B|x)P(x)$
for a measurable set $B \subset {\cal Y}$.
Also, we define  the notations $\bw \times P$ and $\bw \cdot P$ as
$\bw \times P(y,x)\mu(dy):=\bW \times P(dy,x)=w_x(y)P(x)\mu(dy)$ and
$\bw \cdot P(y)\mu(dy):=\bW \cdot P(dy)=\sum_{x\in {\cal X}}w_x(y)P(x)\mu(dy)$.
We also employ the notations $W_P:= \bW \cdot P$ and $w_P:= \bw \cdot P$.

As explained in Section \ref{S51},
we denote the expectation and the variance under the distribution $P \in {\cal P}({\cal Y})$ 
by $\mathbb{E}_P[~]$ and $\mathbb{V}_P[~]$, respectively.
When $P$ is the distribution $W_x \in {\cal P}({\cal Y})$ with $x \in {\cal X}$,
we simplify them as $\mathbb{E}_x[~]$ and $\mathbb{V}_x[~]$, respectively.
This notation is also applied to the $n$-fold extended setting on ${\cal Y}^n$.
In contrast, when we consider the expectation on the discrete set ${\cal X}$ or ${\cal X}^n$,
$\rE_{T}$ expresses the expectation with respect to the random variable $T$ that takes values in 
the set ${\cal X}$ or the set ${\cal X}^n$.

In our analysis, 
for $P \in {\cal P}({\cal X})$, we address the following quantities;
\begin{align}
%C(\bW)&:=\max_{P \in {\cal P}({\cal X})}I(X;Y)_P,\Label{CU} \\
&I(X;Y)_P\nonumber \\
:=&
D(\bW \times P\| W_P \times P)=\sum_{x\in {\cal X}}P(x) D(W_x\|W_P),\\
&I_\alpha (X;Y)_P \nonumber \\
:=&
\min_{Q \in {\cal P}({\cal Y})} D_\alpha(\bW \times P\| Q \times P) \nonumber \\
=&
\min_{Q \in {\cal P}({\cal Y})}
\frac{1}{\alpha-1}
\log \int_{{\cal Y}}
\sum_{x \in {\cal X}} P(x) w_x(y)^{\alpha} q(y)^{-\alpha+1} \mu(dy) \nonumber \\
\stackrel{(a)}{=}&
\frac{\alpha}{\alpha-1}\log 
\int_{{\cal Y}} \Big(\sum_{x \in {\cal X}} P(x) w_x(y)^{\alpha}\Big)^{\frac{1}{\alpha}}\mu(dy),
\Label{LL192E}\\
&H(X)_P\nonumber \\
 :=&-\sum_{x\in {\cal X}}P(x) \log P(x),
%C_{1+s}(W)&:=\max_{P} I_{1+s}(P,W) \\
%I_{1+s}(P,W)&:=\log  \sum_{y \in {\cal Y}}(\sum_{x}P(x)W_x(y)^{1+s})^{\frac{1}{1+s}},
\end{align}
%where the base of logarithm is $2$.
where $(a)$ follows from the equality condition of H\"{o}lder inequality \cite{Sibson}. 
Since 
in this paper, the conditional distribution on $Y$ conditioned with $X$ is fixed to the channel $\bW$,
it is sufficient to fix a joint distribution $P \in {\cal P}({\cal X})$ in the above notation.
In addition, our analysis needs mathematical analysis with a Markov chain $U-X-Y$ with a variable on a finite set ${\cal U}$.
Hence, we generalize the above notation as follows.
\begin{align}
&I(X;Y|U)_P:=\sum_{u \in {\cal U}}P_U(u) D(\bW \times P\| W_P \times P_{X|U=u}),\\
&H(X|U)_P :=-\sum_{u \in {\cal U}}
\sum_{x\in {\cal X}}P(x,u) \log \frac{P(x,u)}{P_U(u)},
\end{align}
and
\begin{align}
&I_\alpha (X;Y|U)_P \nonumber \\
:=&
\sum_{u \in {\cal U}}P_U(u)
\min_{Q \in {\cal P}({\cal Y})} D_\alpha(\bW \times P\| Q \times P_{X|U=u}) .
\end{align}

\subsection{Regions}
Then, we define the following regions.
\begin{align}
&{\cal C}\nonumber \\
:=&
% \cup_{P\in {\cal P}({\cal U}\times {\cal X})} \{ (R_1,R_2) | 0 < R_1-R_2< I(X;Y|U)_P < R_1 < 
%H(X|U)_P ,~0< R_1,~0< R_2\} \nonumber \\
%=&
\bigcup_{P\in {\cal P}({\cal U}\times {\cal X})} 
\left\{
\!\left(\!
\begin{array}{c}
R_1 \\R_2 \\R_3
\end{array}
\! \right)\!
 \left| 
\begin{array}{l}
0 < R_1-R_2< I(X;Y|U)_P, \!\!\!\! \\
R_3 \le R_1- I(X;Y|U)_P, \\
R_1 < H(X|U)_P ,\\
0< R_1,R_2,R_3
\end{array}
\right. \right\}  \\
&{\cal C}^s\nonumber \\
:=&
% \cup_{P\in {\cal P}({\cal U}\times {\cal X})} \{ (R_1,R_2) | 0 < R_1-R_2< I(X;Y|U)_P < R_1 < 
%H(X|U)_P ,~0< R_1,~0< R_2\} \nonumber \\
%=&
\bigcup_{P\in {\cal P}({\cal U}\times {\cal X})} 
\left\{
\!\left(\!
\begin{array}{c}
R_1 \\R_2 \\R_3
\end{array}
\! \right)\!
 \left| 
\begin{array}{l}
0 < R_1-R_2< I(X;Y|U)_P, \!\!\!\!\\
 R_3 \le H(X|YU)_P, \\
R_1 < H(X|U)_P ,\\
0< R_1,R_2,R_3
\end{array}
\right. \right\}  \\
&{\cal C}_\alpha\nonumber \\
:=&
% \cup_{P\in {\cal P}({\cal U}\times {\cal X})} \{ (R_1,R_2) | 0 < R_1-R_2< I(X;Y|U)_P < R_1 < 
%H(X|U)_P ,~0< R_1,~0< R_2\} \nonumber \\
%=&
\bigcup_{P\in {\cal P}({\cal U}\times {\cal X})} 
\left\{
\!\left(\!
\begin{array}{c}
R_1 \\R_2 \\R_3
\end{array}
\! \right)\!
 \left| 
\begin{array}{l}
0 < R_1-R_2< I(X;Y|U)_P,\!\!\!\! \\
R_3 < R_1- I_\alpha(X;Y|U)_P, \\
R_1 < H(X|U)_P ,\\
0< R_1,R_2,R_3
\end{array}
\right. \right\} .
\end{align}
In the above definitions, there is no restriction for the cardinality of ${\cal U}$.
Due to the relations
\begin{align}
\begin{aligned}
H(X|U)_P& =\sum_{u \in {\cal U}}P_U(u) H(X)_{P_{X|U=u}},\\
H(X|YU)_P&=\sum_{u \in {\cal U}}P_U(u) H(X|Y)_{P_{X|U=u}}, 
\end{aligned}
\Label{BNI}
\end{align}
and
$I(X;Y|U)_P=H(X|U)_P-H(X|YU)_P$, 
%\sum_{u \in {\cal U}}P_U(u)I(X;Y)_{P_{X|U=u}}$,
Caratheodory lemma guarantees that
the cardinality of ${\cal U}$ can be restricted to $3$
in the definitions of  ${\cal C}$  
and ${\cal C}^s$.
In addition, the condition $R_3 < R_1- I_\alpha(X;Y|U)_P$ 
in the definition of  ${\cal C}_\alpha$
is rewritten as
\begin{align}
2^{(\alpha-1)I_\alpha(X;Y|U)_P}< 2^{(\alpha-1)(R_1-R_3)} .
\end{align}
Since the relation $2^{(\alpha-1)I_\alpha(X;Y|U)_P}=
\sum_{u \in {\cal U}}P_U(u) 
2^{(\alpha-1)I_\alpha(X;Y)_{P|X|U=u}}$ holds,
Caratheodory lemma guarantees that
the cardinality of ${\cal U}$ can be restricted to $4$
in the definition of  ${\cal C}_\alpha$.  

To see the relation between two regions ${\cal C}$ and ${\cal C}^s$, 
we focus on the inequality
\begin{align}
R_1- I(X;Y|U)_P < &H(X|U)_P - I(X;Y|U)_P 
\nonumber \\
=& H(X|YU)_P
\end{align}
in the region ${\cal C}$.
Hence, 
the condition $R_3 \le R_1- I(X;Y|U)_P$ is stronger than
the condition $R_3 \le H(X|YU)_P$, which implies the relation;
\begin{align}
{\cal C} \subset {\cal C}^s.\Label{COQ}
\end{align}

When we focus only on $R_1$ and $R_3$ instead of $(R_1,R_2,R_3)$,
we have simpler characterizations.
We define the regions;
\begin{align}
{\cal C}^{1,3} &:= \{ (R_1,R_3) |\exists R_2 \hbox{ such that }  
(R_1,R_2,R_3) \in {\cal C}\} \\
{\cal C}^{s,1,3} &:= \{ (R_1,R_3) |\exists R_2 \hbox{ such that }  
(R_1,R_2,R_3) \in {\cal C}^s\} \\
{\cal C}_{\alpha}^{1,3} &:= \{ (R_1,R_3) |\exists R_2 \hbox{ such that }  
(R_1,R_2,R_3) \in {\cal C}_\alpha\} .
\end{align}
Then, we have the following lemma.

\begin{lem}\Label{LL2}
We have
\begin{align}
\overline{{\cal C}^{1,3}} &= 
\left\{
 (R_1,R_3) \left| 0 \le R_1 \le \log d ,~
 0 \le R_3 \le \gamma_1(R_1) 
\right. \right\}\Label{XM1} \\
\overline{{\cal C}_{\alpha}^{1,3}} &=
\left\{
 (R_1,R_3) \left| 0 \le R_1 \le \log d ,~
 0 \le R_3 \le \gamma_{\alpha}(R_1) 
\right. \right\} , \Label{XM3}
\end{align}
and
\begin{align}
&\overline{{\cal C}^{s,1,3}}\nonumber \\
=&
\left\{
 (R_1,R_3) \left| 0 \le R_1 \le \log d ,~
 0 \le R_3 \le  \max_{R \le R_1} \gamma_1(R) 
\right. \right\}\Label{XM2} ,
\end{align}
where $d:= |{\cal X}|$ and 
\begin{align}
&\gamma_1(R_1) \nonumber \\
:=&\max_{P \in {\cal P}({\cal U}\times {\cal X})}
\{ H(X|YU)_P| H(X|U)_P=R_1 \},
 \\
& \gamma_\alpha(R_1) \nonumber \\
:=&
\max_{P \in {\cal P}({\cal U}\times {\cal X})}
\{ R_1- I_\alpha(X;Y|U)_P
| H(X|U)_P=R_1 \}.
\end{align}
When $|{\cal X}|$ is infinite, 
the condition $\le \log d$ is removed in the above equations.
\hfill $\square$\end{lem}
Lemma \ref{LL2} is shown in Appendix \ref{A-LL2}.
For the analysis on the above regions, 
we define the functions;
\begin{align}
\gamma_{1,o}(R_1) &:=\max_{P \in {\cal P}( {\cal X})}
\{ H(X|Y)_P| H(X)_P=R_1 \}
 \\
\gamma_{\alpha,o}(R_1) &:=
\max_{P \in {\cal P}( {\cal X})}
\{ R_1- I_\alpha(X;Y)_P
| H(X)_P=R_1 \}.
\end{align}
Then, we have the following lemma.
\begin{lem}\Label{LL3}
When $\gamma_{1,o}$ is a concave function, 
we have $\gamma_{1}(R_1)=\gamma_{1,o}(R_1)$.
When $\gamma_{\alpha,o}$ is a concave function, 
we have $\gamma_{\alpha}(R_1)=\gamma_{\alpha,o}(R_1)$.
\end{lem}
Lemma \ref{LL3} is shown in Appendix \ref{A-LL3}.
Using these two lemmas, we numerically calculate
the regions
$\overline{{\cal C}^{1,3}}$,
$\overline{{\cal C}^{s,1,3}}$, and
$\overline{{\cal C}_{\alpha}^{1,3}}$ as Fig. \ref{F-graph}.

\begin{figure}[t]
%\centering
%\includegraphics[scale=0.4]{MHRepeater.png}
\begin{center}
  \includegraphics[width=0.98\linewidth]{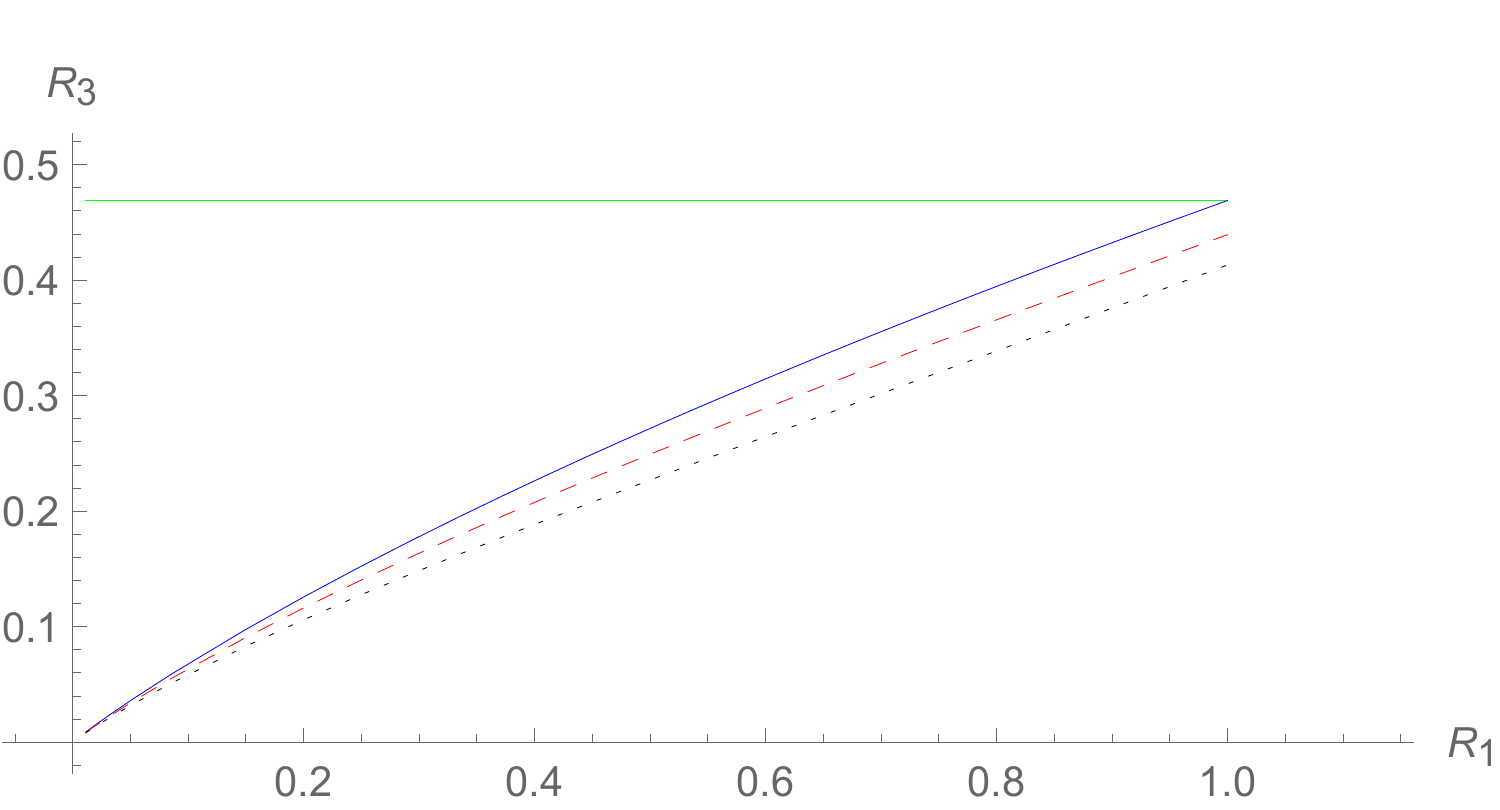}
  \end{center}
\caption{Numerical plots for 
$\overline{{\cal C}^{1,3}}$, 
$\overline{{\cal C}^{s,1,3}}$ and $\overline{{\cal C}_{\alpha}^{1,3}}$
under the binary symmetric channel with cross over probability $0.1$.
Green normal horizontal line expresses the upper bound of $\overline{{\cal C}^{s,1,3}}$.
Blue normal line expresses the upper bound of $\overline{{\cal C}^{1,3}}$.
Red dashed line expresses the upper bound of $\overline{{\cal C}_{1.1}^{1,3}}$.
Black dotted line expresses the upper bound of $\overline{{\cal C}_{1.2}^{1,3}}$.
Other bounds of $\overline{{\cal C}_{1,3}}$ and $\overline{{\cal C}_{\alpha,1,3}}$ are $R_1=1$ and $R_3=0$.
We numerically checked that 
 $\gamma_{1,o}$, $\gamma_{1.1,o}$, and $\gamma_{1.2,o}$
satisfy the condition in Lemma \ref{LL3}.}
\Label{F-graph}
\end{figure}

We also define the quantities;
\begin{align}
C&:=\sup_{(R_1,R_2,R_3)\in {\cal C}}  R_3,\quad
C^s:=
\sup_{(R_1,R_2,R_3)\in {\cal C}^s}  R_3,\\
C_\alpha&:=
\sup_{(R_1,R_2,R_3)\in {\cal C}_\alpha}  R_3.
\end{align}
Then, using \eqref{XM2} and \eqref{XM1}, we have the following lemma.
\begin{lem}\Label{LL1}
\begin{align}
C&=C^s=\max_{P \in {\cal P}({\cal X})} H(X|Y)_P, \Label{BKO}\\
C_\alpha&=\max_{P \in {\cal P}({\cal X})}H(X)_P - I_\alpha(X;Y)_P.
\end{align}
\hfill $\square$\end{lem}

\section{Results for secure list decoding with discrete input}\Label{S51}
\subsection{Statements of results}
%\subsection{Capacity regions}\Label{S-Y}
To give the capacity region, 
we consider $n$-fold discrete memoryless extension $\bW^n$ of the channel $\bW$.
A sequence of codes $\{(\phi_n,D_n)\}$ is called strongly secure 
when 
$\epsilon_A(\phi_n,D_n) $ and
$\delta_{D}(D_n) $ approach to zero.
A sequence of codes $\{(\phi_n,D_n)\}$ is called weakly secure 
when 
$\epsilon_A(\phi_n,D_n) $ and
$\delta_C(\phi_n,D_n) $ approach to zero.
A rate triple $(R_1,R_2,R_3)$ is strongly deterministically (stochastically) achievable when
there exists a strongly secure sequence of deterministic (stochastic) codes $\{(\phi_n,D_n)\}$
such that 
$\frac{1}{n}\log |(\phi_n,D_n)|_1$ approaches to $ R_1$,  
$\frac{1}{n}\log |(\phi_n,D_n)|_2$ approaches to $R_2$\footnote{The definitions of 
$|(\phi_n,D_n)|_1$ and $|(\phi_n,D_n)|_2$ are given in the end of Section \ref{S2-1B}.}, and
$\lim_{n \to \infty}\frac{1}{n}E(\phi_n) \ge R_3$.
A rate triple $(R_1,R_2,R_3)$ is $\alpha$-strongly deterministically (stochastically) achievable when
there exists a strongly secure sequence of deterministic (stochastic) codes $\{(\phi_n,D_n)\}$
such that 
$\frac{1}{n}\log |(\phi_n,D_n)|_1$ approaches to $ R_1$,  
$\frac{1}{n}\log |(\phi_n,D_n)|_2$ approaches to $R_2$, and
$\lim_{n \to \infty}\frac{1}{n}E_\alpha(\phi_n) \ge R_3$.
A rate triplet $(R_1,R_2,R_3)$ is ($\alpha$-)weakly deterministically (stochastically) achievable when
there exists a weakly secure sequence of deterministic (stochastic) codes $\{(\phi_n,D_n)\}$
such that 
$\frac{1}{n}\log |(\phi_n,D_n)|_1$ approaches to $R_1$,
$\frac{1}{n}\log |(\phi_n,D_n)|_2$ approaches to $R_2$, and
$\lim_{n \to \infty}\frac{1}{n}E(\phi_n) \ge R_3$
($\lim_{n \to \infty}\frac{1}{n}E_\alpha(\phi_n) \ge R_3$).
%In fact, when $R_1$ is larger than channel capacity in the asymptotic setting,
%Non-decodable condition (B) holds due to the strong converse theorem of the channel coding.
Then, we denote the set of strongly deterministically (stochastically) achievable rate triple 
$(R_1,R_2,R_3)$ by ${\cal R}_{(s,d)}^L$ (${\cal R}_{(s,s)}^L$).
In the same way, we denote the set of weakly deterministically (stochastically) achievable rate triple 
$(R_1,R_2,R_3)$ by ${\cal R}_{(w,d)}^L$ (${\cal R}_{(w,s)}^L$).
The $\alpha$-version with $\alpha>1$ is denoted by 
${\cal R}_{(s,d)}^{L,\alpha}$, ${\cal R}_{(s,s)}^\alpha$,
${\cal R}_{(w,d)}^{L,\alpha}$, and ${\cal R}_{(w,s)}^\alpha$, respectively.

%To characterize these rate region, we prepare the following lemma.
As outer bounds of ${\cal R}_{(w,d)}^L$, ${\cal R}_{(s,s)}^L$, and ${\cal R}_{(s,d)}^L$, we have the following theorem.
\begin{thm}\Label{Converse}
We have the following characterization.
\begin{align}
{\cal R}_{(w,d)}^L  \subset \overline{{\cal C}}
,\quad % ,\Label{Con1} \\
{\cal R}_{(s,s)}^L \subset \overline{{\cal C}^s} 
,\quad % ,\Label{Con1} \\
{\cal R}_{(s,d)}^L \subset \overline{{\cal C}} ,\Label{Con2} 
\end{align}
where $ \overline{{\cal C}}$ expresses the closure of the set ${\cal C}$.
\hfill $\square$\end{thm}

For their inner bounds, we have the following theorem.
\begin{thm}\Label{TH3}
Assume the condition (W2).
(i) A rate triplet $(R_1,R_2,R_3)$ is strongly deterministically achievable
when 
there exists a distribution $P \in {\cal P}(\cX)$ such that 
\begin{align}
0 < R_1-R_2 &< I(X;Y)_P ,\Label{NHOA}\\
R_1 & < H(X)_P, \Label{NHOB}\\
R_3& \le R_1-I(X;Y)_P \Label{NHO}.
%I(X;Y)_P\eta (I(X;Y)_P)-F(I(X;Y)_P|P)\\
%R_2 &> R_1 -I(X;Y)_P.
\end{align}
(ii) A rate triplet $(R_1,R_2,R_3)$ is $\alpha$-strongly deterministically achievable
when 
there exists a distribution $P \in {\cal P}(\cX)$ such that 
\begin{align}
0 < R_1-R_2 &< I(X;Y)_P, \Label{NHO2A}\\
 R_1 & < H(X)_P , \Label{NHO2B}\\
 R_3 & \le R_1- I_\alpha(X;Y)_P 
 \Label{NHO2}.
%I(X;Y)_P\eta (I(X;Y)_P)-F(I(X;Y)_P|P)\\
%R_2 &> R_1 -I(X;Y)_P.
\end{align}
\hfill $\square$\end{thm}

In fact, 
the condition $R_1-R_2< I(X;Y)_P$ corresponds to Verifiable condition (A),
the condition $I(X;Y)_P \le R_1-R_3$ ($I_\alpha(X;Y)_P \le R_1-R_3 $) does to 
(R\'{e}nyi) equivocation type of concealing condition (B),
and
the condition $R_1 < H(X)_P $ does to the binding condition for dishonest Alice (D).
Theorems \ref{Converse} and \ref{TH3} are shown in Sections \ref{S-C} and \ref{S-K}, respectively.
We have the following corollaries from Theorems \ref{Converse} and \ref{TH3}. %, whose detailed derivations are given in Section \ref{S-D}.

\begin{corollary}\Label{Cor46}
When Condition (W2) holds,
we have the following relation for $G\in \{(s,d),(w,d)\}$;
\begin{align}
\overline{{\cal R}_{G}^L}=\overline{{\cal C}}
\Label{TDAB}
\end{align}
and
\begin{align}
{\cal C}_{\alpha} \subset{\cal R}_{G}^{L,\alpha}.
\Label{TDAB5}
\end{align}
\hfill $\square$
\end{corollary}

Hence, even when our binding condition is relaxed to Condition (C), 
when our code is limited to deterministic codes,
we have the same region as the case with Condition (D).
%The above corollary means that 

\begin{proof}
It is sufficient to show the direct part.
For this aim, we notice that the following relation for $\alpha >\alpha'>1$;
\begin{align}
{{\cal R}_{G}^{L,\alpha}}
\subset {{\cal R}_{G}^{L,\alpha'}},
\quad
\overline{\cup_{\alpha>1} {\cal C}_\alpha}=
\overline{{\cal C}}.
\end{align}
Hence, it is sufficient to show 
that there exists a strongly secure sequence of deterministic codes
with the rate triplet $(R_1,R_2,R_3)$ to satisfy 
\begin{align}
0 < R_1-R_2&< I(X;Y|U)_P, \Label{NPAA}\\
 R_1 &< H(X|U)_P ,\Label{NPAB}\\
 R_3 &\le R_1- I_\alpha(X;Y|U)_P \Label{NPA}
\end{align}
for a given $P\in {\cal P}({\cal X} \times {\cal U})$.
%For a given $P\in {\cal P}({\cal X} \times {\cal U})$, 
There exist
distributions $P_1, \ldots, P_{\sU} \in {\cal P}({\cal X})$ such that
${\cal U}=\{1, \ldots, \sU \} $ and $P_u(x)= \frac{P(x,u)}{P_U(u)}$ for $u \in {\cal U}$,
where $P_U(u)=\sum_{x'\in {\cal X}} P(x',u) $.
Then, we have
$\sum_{u \in {\cal U}}P_U (u)I(X;Y)_{P_u}=I(X;Y|U)_{P}$,
$\sum_{u \in {\cal U}}P_U (u)H(X)_{P_u}=H(X|U)_{P}$, and
$\sum_{u \in {\cal U}}P_U (u)I_\alpha(X;Y)_{P_u}=I_\alpha(X;Y|U)_{P}$.

For simplicity, in the following, we consider the case with $\sU=2$.
We choose a sequence $\{(\phi_{n,1},D_{n,1})\}$ ($\{(\phi_{n,2},D_{n,2})\}$) 
of strongly secure deterministic codes that achieve
the rates to satisfy \eqref{NHO2A}, \eqref{NHO2B}, and \eqref{NHO2} with $P=P_1 (P_2)$.
We denote $P_U(1)$ by $\lambda$.
Then, we define the concatenation $\{(\phi_{n},D_{n})\}$ 
as follows.
We assume that $\phi_{ \lfloor \lambda n \rfloor,1}$($\phi_{n-\lfloor \lambda n \rfloor,2}$) 
is a map from ${\cal M}_1$(${\cal M}_2$) to 
${\cal X}^{\lfloor \lambda n \rfloor}$
(${\cal X}^{n-\lfloor \lambda n \rfloor}$).
The encoder $\phi_{n} $ is given as a map 
from $(m_1,m_2)\in {\cal M}_1 \times {\cal M}_2$ to 
$(\phi_{\lfloor \lambda n \rfloor}(m_1),
\phi_{n-{\lfloor \lambda n \rfloor}}(m_2))\in {\cal X}^{n}$. 
The decoder $D_{n}$ is given as a map from 
${\cal Y}^{n}$ to ${\cal M}_1^{\sL_1} \times {{\cal M}_2}^{\sL_2}$ as
\begin{align}
&D_{n}(y_1, \ldots, y_{n})\nonumber \\
:=&(D_{\lfloor \lambda n \rfloor,1}(y_1, \ldots, y_{\lfloor \lambda n \rfloor}),
D_{n-\lfloor \lambda n \rfloor,2}
(y_{\lfloor \lambda n \rfloor+1}, \ldots, y_{n})) 
\end{align}
for $(y_1, \ldots, y_{n})\in {\cal Y}^{n}$.
We have 
$\epsilon_A(\phi_{n},D_{n}) \le \epsilon_A(\phi_{\lfloor \lambda n \rfloor,1},D_{\lfloor \lambda n \rfloor,1}) 
+ \epsilon_A(\phi_{n-\lfloor \lambda n \rfloor,2},D_{n-\lfloor \lambda n \rfloor,2}) $
because 
the code $(\phi_{n},D_{n})$ is correctly decoded when 
both codes $(\phi_{\lfloor \lambda n \rfloor,1},D_{\lfloor \lambda n \rfloor,2})$ 
and $(\phi_{n-\lfloor \lambda n \rfloor,2},D_{n-\lfloor \lambda n \rfloor,2})$ are correctly decoded.
Alice can cheat the decoder $D_{n}$
only when 
Alice cheats one of the decoders $D_{\lfloor \lambda n \rfloor,1}$ and $D_{n-\lfloor \lambda n \rfloor,2}$.
Hence,
$\delta_{D}(D_{n}) \le \min( \delta_{D}(D_{\lfloor \lambda n \rfloor,1}),
\delta_{D}(D_{n-\lfloor \lambda n \rfloor,2}))$.
Therefore, 
the concatenation $\{(\phi_{n},D_{n})\}$ is also strongly secure.

The rate tuples of the code $(\phi_{n},D_{n})$ is calculated as  
$|(\phi_{n},D_{n})|_i= |(\phi_{\lfloor \lambda n \rfloor,1},D_{\lfloor \lambda n \rfloor,1})|_i
+|(\phi_{n-\lfloor \lambda n \rfloor,2},D_{n-\lfloor \lambda n \rfloor,2})|_i$ for $i=1,2$.
Also, using the additivity property \eqref{AOR}, 
we have $E_\alpha(\phi_n)=E_\alpha(\phi_{\lfloor \lambda n \rfloor,1})
+E_\alpha(\phi_{n-\lfloor \lambda n \rfloor,2})$.
Hence, we have shown 
the existence of a strongly secure sequence of deterministic codes
with the rate triplet $(R_1,R_2,R_3)$ to satisfy the conditions
\eqref{NPAA}, \eqref{NPAB}, and \eqref{NPA}
when $\sU=2$.
For a general $\sU$, we can show the same statement by repeating 
the above procedure.
\end{proof}

\subsection{Outline of proof of direct theorem}\Label{OUT-K}
Here, we present the outline of the direct theorem (Theorem \ref{TH3}).
Since $\lim_{\alpha \to 1}I_\alpha(X;Y)_P=I (X;Y)_P$, 
the first part (i) follows from the second part (ii).
Hence, we show only the second part (ii) in Section \ref{S-K}
based on the random coding.
To realize Binding condition for dishonest Alice (D),
we need to exclude the existence of $x^n \in {\cal X}^n$ and $m\neq m' \in {\cal M}_n$
such that $1-\epsilon_{A,m}(x^n,D)$ and $1-\epsilon_{A,m'}(x^n,D)$ are far from 0.
For this aim, we 
focus on Hamming distance 
$d_H(x^n,{x^n}')$ between $x^n=(x_1^n, \ldots, x^n_n), {x^n}'=({x_1^n}', \ldots, {x^n_n}') \in {\cal X}^n$
as
\begin{align}
d_H(x^n,{x^n}'):= | \{  k|  x_k^n\neq {x_k^n}'\}|.
\end{align}
and introduce functions $\{\xi_x\}_{x \in {\cal X}}$ to satisfy the following conditions;
\begin{align}
&\mathbb{E}_x[\xi_x(Y)] =0,\Label{CS1}\\
&\zeta_1:=\min_{x\neq x' \in {\cal X}} \mathbb{E}_{x'}[-\xi_x(Y)] >0, \Label{CS2}\\
&\zeta_2 :=\max_{x, x' \in {\cal X}} \mathbb{V}_{x'}[\xi_x(Y)] < \infty.\Label{CS3}
\end{align}
For $x^n=(x_1^n, \ldots, x_n^n)\in {\cal X}^n$ and $y^n=(y_1^n, \ldots, y_n^n)\in {\cal Y}^n$, 
we define
\begin{align}
\xi_{x^n}(y^n):= \sum_{i=1}^n \xi_{x_i^n}(y_i^n).
\end{align}
Then, given an encoder $\phi_n$ mapping ${\cal M}_n$ to ${\cal X}^n$,
we impose the following condition on Bob's decoder to include the message $m$ in his decoded list;
the inequality 
\begin{align}
\xi_{\phi_n(m)}(Y^n) \ge - \epsilon_1 n \Label{CS4}
\end{align}
holds when $Y^n$ is observed.
The condition \eqref{CS4} guarantees that 
$1-\epsilon_{A,m}(x^n,D)$ is small when 
$d_H(x^n, \phi_n(m))$ is larger than a certain threshold.

As shown in Section \ref{S-K},
due to the conditions \eqref{CS1},  \eqref{CS2}, and \eqref{CS3}, 
the condition \eqref{CS4} guarantees that the quantity $\delta_D(D)$ is small.
Indeed, we have the following lemma, which is shown in Section \ref{PfL1}.
\begin{lem}\Label{LS3}
When the condition (W2) holds, there exist
functions $\{\xi_x\}_{x \in {\cal X}}$ to satisfy the conditions \eqref{CS1},  \eqref{CS2}, and \eqref{CS3}.
\hfill $\square$\end{lem}

\section{Results for secure list decoding with continuous input}\Label{S7}
In the previous section, we assume that Alice can access only 
elements of the finite set ${\cal X}$ even when Alice is malicious.
However, in the wireless communication case,
the input system is given as a continuous space $\tilde{\cal X}$.
When we transmit a message via such a channel,
usually we fix the set ${\cal X}$ of constellation points as a subset of $\tilde{\cal X}$,
and the modulator converts an element of input alphabet to a constellation point.
That is, the choice of the set ${\cal X}$ depends on the performance of the modulator.
In this situation, it is natural that dishonest Alice can send any element of the continuous space 
$\tilde{\cal X}$ while honest Alice sends only an element of ${\cal X}$.
Therefore, only the condition (D) is changed as follows
because only the condition (D) is related to dishonest Alice.

\begin{description}
\item[(D')] Binding condition for dishonest Alice.
For $x\in \tilde{\cal X}$, we define the quantity
$\delta_{D',x}(D)$ as the second largest value among 
$\{(1-\epsilon_{A,m}(x,D))\}_{m=1}^{\sM}$.
Then, the relation
\begin{align}
\delta_{D'}(D)&:=
\max_{x \in {\cal X}} \delta_{D',x}(D)
\le  \delta_C \Label{HB3EP}
\end{align}
holds.
\end{description}

Then, a sequence of codes $\{(\phi_n,D_n)\}$ is called ultimately secure when 
$\epsilon_A(\phi_n,D_n) $ and $\delta_{D'}(D_n) $ approach to zero.
A rate triple $(R_1,R_2,R_3)$ is ($\alpha$)-ultimately deterministically (stochastically) achievable when
there exists a ultimately secure sequence of deterministic (stochastic) codes $\{(\phi_n,D_n)\}$
such that 
$\frac{1}{n}\log |(\phi_n,D_n)|_1$ approaches to $ R_1$,  
$\frac{1}{n}\log |(\phi_n,D_n)|_2$ approaches to $R_2$, and
$\lim_{n \to \infty}\frac{1}{n}E(\phi_n) \ge R_3$
($\lim_{n \to \infty}\frac{1}{n}E_\alpha(\phi_n) \ge R_3$).
We denote the set of ultimately deterministically (stochastically) achievable rate triple 
$(R_1,R_2,R_3)$ by ${\cal R}_{(u,d)}^L$ (${\cal R}_{(u,s)}^L$).
The $\alpha$-version with $\alpha>1$ is denoted by 
${\cal R}_{(u,d)}^{L,\alpha}$, ${\cal R}_{(u,s)}^{L,\alpha}$, respectively.

The same converse result as Theorem \ref{Converse} holds 
for ${\cal R}_{(u,d)}^L$ and ${\cal R}_{(u,s)}^L$
because a sequence of ultimately secure codes is strongly secure.
Hence, the aim of this section is to recover the same result as Theorem \ref{TH3} for 
ultimately secure codes under a certain condition for our channel.
The key point of this problem setting is 
to exclude the existence of $x^n\in \tilde{\cal X}^n$ and $m\neq m' \in {\cal M}_n$
such that $1-\epsilon_{A,m}(x^n,D)$ and $1-\epsilon_{A,m'}(x^n,D)$ are far from 0.
For this aim, we need to assume a distance $d$ on the space $\tilde{\cal X}$.
%For two element $d(x,x')$ for 
%focus on Hamming distance 
%$d(x^n,{x^n}')$ between $x^n=(x_1^n, \ldots, x^n_n), {x^n}'=({x_1^n}', \ldots, {x^n_n}') \in {\cal X}^n$
Then, we may consider functions $\{\xi_x\}_{x \in {\cal X}}$ to satisfy the following conditions
in addition to \eqref{CS1};
\begin{align}
&\hat{\zeta}_1(r):=\inf_{x \in {\cal X}, x' \in \tilde{\cal X}: d(x,x')\ge r} 
\mathbb{E}_{x'}[-\xi_x(Y)] >0, \Label{CS2B}\\
&\hat{\zeta}_2 :=\sup_{x \in {\cal X}, x' \in \tilde{\cal X}} \mathbb{V}_{x'}[\xi_x(Y)] < \infty\Label{CS3B}
\end{align}
for $r>0$.
It is not difficult to prove the same result as Theorem \ref{TH3}
when the above functions $\{\xi_x\}_{x \in {\cal X}}$ exist.
However, it is not so easy to prove the existence of the above functions
under natural models including AWGN channel.
%In particular, it is hard to satisfy the condition \eqref{CS3B}.
Therefore, we introduce the following condition instead of \eqref{CS2B} and \eqref{CS3B}.
\begin{description}
\item[(W3)] 
There exist functions $\{\xi_x\}_{x \in \tilde{\cal X}}$ to satisfy the following conditions in addition to \eqref{CS1};
\begin{align}
&\bar{\zeta}_{1,t}(r)\nonumber \\
:=&\frac{-1}{t}\log \sup_{x \in {\cal X}, x' \in \tilde{\cal X}: d(x,x')\ge r} 
\mathbb{E}_{x'}[2^{t(\xi_{x}(Y)-\xi_{x'}(Y))}] \nonumber \\
>&0, \Label{CS2C}\\
&\bar{\zeta}_2 :=\sup_{x \in \tilde{\cal X}} \mathbb{V}_{x}[\xi_x(Y)] < \infty\Label{CS3C}
\end{align}
for $r>0$ and $t\in (0,1/2)$.
Indeed, as discussed in Step 1 of our proof of Lemma \ref{LL12B},
when functions $\{\xi_x\}_{x \in \tilde{\cal X}}$ satisfy the above conditions 
and the difference between two vectors ${x^n}'$ and $x^n$ satisfy a certain condition,
we can distinguish a vector ${x^n}'$ from $x^n$ by using $\xi_{x_1}+\cdots +\xi_{x_n}$.
\end{description}
Notice that $\bar{\zeta}_{1,t}(r)$ is monotonically increasing for $r$.

That is, we have the following theorem.
\begin{thm}\Label{TH4}
Assume the conditions (W2) and (W3).
(i) A rate triplet $(R_1,R_2,R_3)$ is ultimately deterministically achievable
when 
there exists a distribution $P \in {\cal P}(\cX)$ such that 
\begin{align}
0 < R_1-R_2< I(X;Y)_P \le R_1-R_3\le  R_1 < H(X)_P \Label{NHOLL}.
%I(X;Y)_P\eta (I(X;Y)_P)-F(I(X;Y)_P|P)\\
%R_2 &> R_1 -I(X;Y)_P.
\end{align}
(ii) A rate triplet $(R_1,R_2,R_3)$ is $\alpha$-ultimately deterministically achievable
when 
there exists a distribution $P \in {\cal P}(\cX)$ such that 
\begin{align}
0 <& R_1-R_2< I(X;Y)_P \le I_\alpha(X;Y)_P \nonumber \\
\le &R_1-R_3\le  R_1 < H(X)_P \Label{NHO2BT}.
%I(X;Y)_P\eta (I(X;Y)_P)-F(I(X;Y)_P|P)\\
%R_2 &> R_1 -I(X;Y)_P.
\end{align}
\hfill $\square$\end{thm}

Since $ {\cal R}_{(u,d)}^L \subset {\cal R}_{(s,d)}^L$
and $ {\cal R}_{(u,s)}^L \subset {\cal R}_{(s,s)}^L$,
the combination of Theorems \ref{Converse} and \ref{TH4} yields the following corollary
in the same way as Corollary \ref{Cor46}.

\begin{corollary}\Label{CorT}
When Conditions (W2) and (W3) hold,
we have the following relations % for $K\in \{(u,s),(u,d)\}$.
\begin{align}
\overline{{\cal R}_{(u,d)}^L}=\overline{{\cal C}}, \quad
\overline{{\cal R}_{(u,s)}^L}\subset\overline{{\cal C}^s}
\Label{TDAN}
\end{align}
and
\begin{align}
{\cal C}_\alpha\subset{\cal R}_{(u,d)}^{L,\alpha} \subset{\cal R}_{(u,s)}^{L,\alpha}.
\Label{TDAB5LL}
\end{align}
\hfill $\square$
\end{corollary}

As an example, we consider an additive noise channel when 
$\tilde{\cal X}=\mathbb{R}^d$,
which equips the standard Euclidean distance $d$.
The output system ${\cal Y}$ is also given as $\mathbb{R}^d$.
We fix a distribution $P_N$ for the additive noise $N$ on 
$\tilde{\cal X}$.
%where its probability density function $p_N$ is continuous on $\mathbb{R}^d$. 
Then, we define the additive noise channel $\{W[P_N]_x\}_{x \in \tilde{\cal X}}$ as $w_x(y):= p_N(y-x)$.
We assume the following conditions;
\begin{align}
&\infty> \mathbb{E}_0[ -\log w_0(Y)]>-\infty \Label{FA1}\\
& \mathbb{V}_0[ -\log w_0(Y)]<\infty.\Label{FA2}
\end{align}
Then, we have the following lemma.
\begin{lem}\Label{LE4}
When the additive noise channel $\{W[P_N]_x\}_{x \in \tilde{\cal X}}$ satisfies \eqref{FA1} and \eqref{FA2}, 
and when $\xi_x$ is chosen as $\xi_x(y):=\log w_x (y)
-\mathbb{E}_0[ \log w_0(Y)]$,
the condition (W3) holds.
\hfill $\square$\end{lem}

\begin{proof}
Since the range of $t$ in the condition \eqref{CS2C}
is $(0,1/2)$,
we assume thwe assume that the real number $t$ belongs to $(0,1/2)$ in this proof.
The conditions \eqref{CS1} and \eqref{CS3C} follow from \eqref{FA1} and \eqref{FA2}, respectively.
\begin{align}
-\frac{1}{t}\log \mathbb{E}_{x'}[2^{t(\xi_{x}(Y)-\xi_{x'}(Y))}] 
=D_{1-t}(W_{x'}\|W_{x}) .
\end{align}
For an small real number $\epsilon<1/3$, we choose $r_0>0$ such that
\begin{align}
W_0(\{ y\in {\cal Y}|  d(y,0)< r_0 \}) \le \epsilon. \Label{BH1}
\end{align}
We define the function $f$ from ${\cal Y}$ to $\{0,1\}$ such that
$f^{-1}(\{0\})=\{ y\in {\cal Y}|  d(y,0)< r_0 \}$.
When $x_0 $ satisfies $ d(x_0,0)>2r_0$, we have
\begin{align}
W_{x_0}\circ f^{-1}(\{0\}) \le  \epsilon.\Label{BH2}
\end{align}
Since $W_{x_0}\circ f^{-1}(\{1\}),W_{0}\circ f^{-1}(\{0\})\le 1 $,
\eqref{BH1} and \eqref{BH2} imply that
\begin{align}
&2^{ -t D_{1-t}(W_{x_0}\circ f^{-1}\|W_{0}\circ f^{-1}) }
\nonumber \\
=&W_{x_0}\circ f^{-1}(\{0\})^{1-t}W_{0}\circ f^{-1}(\{0\})^t\nonumber \\
&+W_{x_0}\circ f^{-1}(\{1\})^{1-t}W_{0}\circ f^{-1}(\{1\})^t \nonumber\\
\le& \epsilon^{t}+\epsilon^{1-t}.
\end{align}
Thus,
\begin{align}
D_{1-t}(W_{x_0}\circ f^{-1}\|W_{0}\circ f^{-1}) 
\ge - \frac{1}{t} \log (\epsilon^{t}+\epsilon^{1-t}).
\end{align}
When $d(x,x')>2r_0$, we have
\begin{align}
& -\frac{1}{t}\log \mathbb{E}_{x'}[2^{t(\xi_{x}(Y)-\xi_{x'}(Y))}] 
\nonumber \\
=&D_{1-t}(W_{x'}\|W_{x}) 
= D_{1-t}(W_{x'-x}\|W_{0})\nonumber  \\
\ge & D_{1-t}(W_{x'-x}\circ f^{-1}\|W_{0}\circ f^{-1}) \nonumber \\
\ge &- \frac{1}{t} \log (\epsilon^{t}+\epsilon^{1-t}) >0.
\end{align}
Therefore,
\begin{align}
&\inf_{x' \in \tilde{\cal X}: d(0,x')\ge r} 
\frac{-1}{t}\log 
\mathbb{E}_{x'}[2^{t(\xi_{0}(Y)-\xi_{x'}(Y))}]  
\nonumber \\
=&
\min \Big(
\inf_{ x' \in \tilde{\cal X}: r_0 \ge d(0,x')\ge r} 
\frac{-1}{t}\log 
\mathbb{E}_{x'}[2^{t(\xi_{0}(Y)-\xi_{x'}(Y))}]  ,\nonumber \\
&\inf_{ x' \in \tilde{\cal X}: d(0,x')> r_0} 
\frac{-1}{t}\log 
\mathbb{E}_{x'}[2^{t(\xi_{0}(Y)-\xi_{x'}(Y))}]  
\Big) \nonumber \\
\ge &
\min \Big(
\min_{ x' \in \tilde{\cal X}: r_0 \ge d(0,x')\ge r} 
D_{1-t}(W_{x'}\|W_{0}) ,\nonumber \\
&- \frac{1}{t} \log (\epsilon^{t}+\epsilon^{1-t})
\Big). \Label{NNT}
\end{align}
Since $D_{1-t}(W_{x'}\|W_{0}) >0$ for $x'\neq 0$,
the set $\{ x' \in \tilde{\cal X}| r_0 \ge d(0,x')\ge r\}$ is compact, and
the map $x' \mapsto D_{1-t}(W_{x'}\|W_{0})$ continuous, we find that
$\min_{ x' \in \tilde{\cal X}: r_0 \ge d(0,x')\ge r} 
D_{1-t}(W_{x'}\|W_{0}) >0$.
Hence, the quantity \eqref{NNT} is strictly positive.

Since
\begin{align}
\bar{\zeta}_{1,t}(r)
=&
\inf_{x \in {\cal X}, x' \in \tilde{\cal X}: d(x,x')\ge r} 
\frac{-1}{t}\log 
\mathbb{E}_{x'}[2^{t(\xi_{x}(Y)-\xi_{x'}(Y))}]\nonumber   \\
=&
\inf_{ x' \in \tilde{\cal X}: d(0,x')\ge r} 
\frac{-1}{t}\log 
\mathbb{E}_{x'}[2^{t(\xi_{0}(Y)-\xi_{x'}(Y))}]  ,
\end{align}
the condition \eqref{CS2C} holds.
\end{proof}

\section{Application to bit-string commitment}\Label{SecBS}
\subsection{Bit-string commitment based on secure list decoding}
Now, we construct a code for bit-string commitment
by using our code $(\phi,D)$ for secure list decoding.
(i) The previous studies \cite[Theorem 2]{BC1}, \cite{BC2} considered only the case with a discrete input alphabet ${\cal X}$ and discrete output alphabet ${\cal Y}$
while a continuous generalization of their result was mentioned as an open problem in 
\cite{BC2}. 
We allow a continuous output alphabet ${\cal Y}$ with a discrete input alphabet ${\cal X}$.
(ii) As another setting, we consider the continuous input alphabet ${\cal X}$.
In this case, it is possible to make the capacity infinite, 
as pointed by the paper \cite{NBSI} in the case of the Gaussian channel.
However, it is difficult to manage an input alphabet with infinitely many cardinality.
Hence, we consider a restricted finite subset $\tilde{\cal X}$ of the continuous input alphabet ${\cal X}$ so that
honest Alice accesses only a restricted finite subset $\tilde{\cal X}$ of the continuous input alphabet ${\cal X}$
and
dishonest Alice accesses the continuous input alphabet ${\cal X}$.

Since the binding condition (BIN) is satisfied by Condition (D) or (D'),
it is sufficient to strengthen Condition (B) to Concealing condition (CON).
For this aim, we combine a hash function and a code $(\phi,D)$ for secure list decoding.
A function $f$ from ${\cal M}$ to ${\cal K}$ is called 
{\it a regular hash function}
when $f$ is surjective and the cardinality $|f^{-1}(k)|$ does not depend on $k \in {\cal K}$.
When a code $(\phi,D)$ and a regular hash function $f$ are given, 
as explained in Fig. \ref{FF2},
we can naturally consider 
the following protocol for bit-string commitment with message set $ {\cal K}$.
Before starting the protocol, Alice and Bob share a code $(\phi,D)$ and a regular hash function $f$.
\begin{description}
\item[(I)] (Commit Phase) 
When $k \in {\cal K}$ is a message to be sent by Alice,
she randomly chooses an element $M \in {\cal M}$ subject to uniform distribution on  $f^{-1}(k)$. 
Then, Alice sends $\phi(M)$ to Bob via a noisy channel. 

\item[(II)] (Reveal Phase) 
From Bob's receiving information in Commit Phase,
Bob outputs $\sL$ elements of ${\cal M}$ as the list. 
Alice sends $M$ to Bob via a noiseless channel.
The list is required to contain the message $M$. 
If the transmitted information via the noiseless channel is contained in Bob's decoded list, Bob accepts it, 
%i.e., considers that $M=M'$,
and recovers the message $k=f(M)$.
Otherwise, Bob rejects it.
%Here, the remaining $L-1$ products will be resolved into their component elements and used for the resources for the next products.
\end{description}

\begin{figure}[t]
%\centering
%\includegraphics[scale=0.4]{MHRepeater.png}
\begin{center}
  \includegraphics[width=0.98\linewidth]{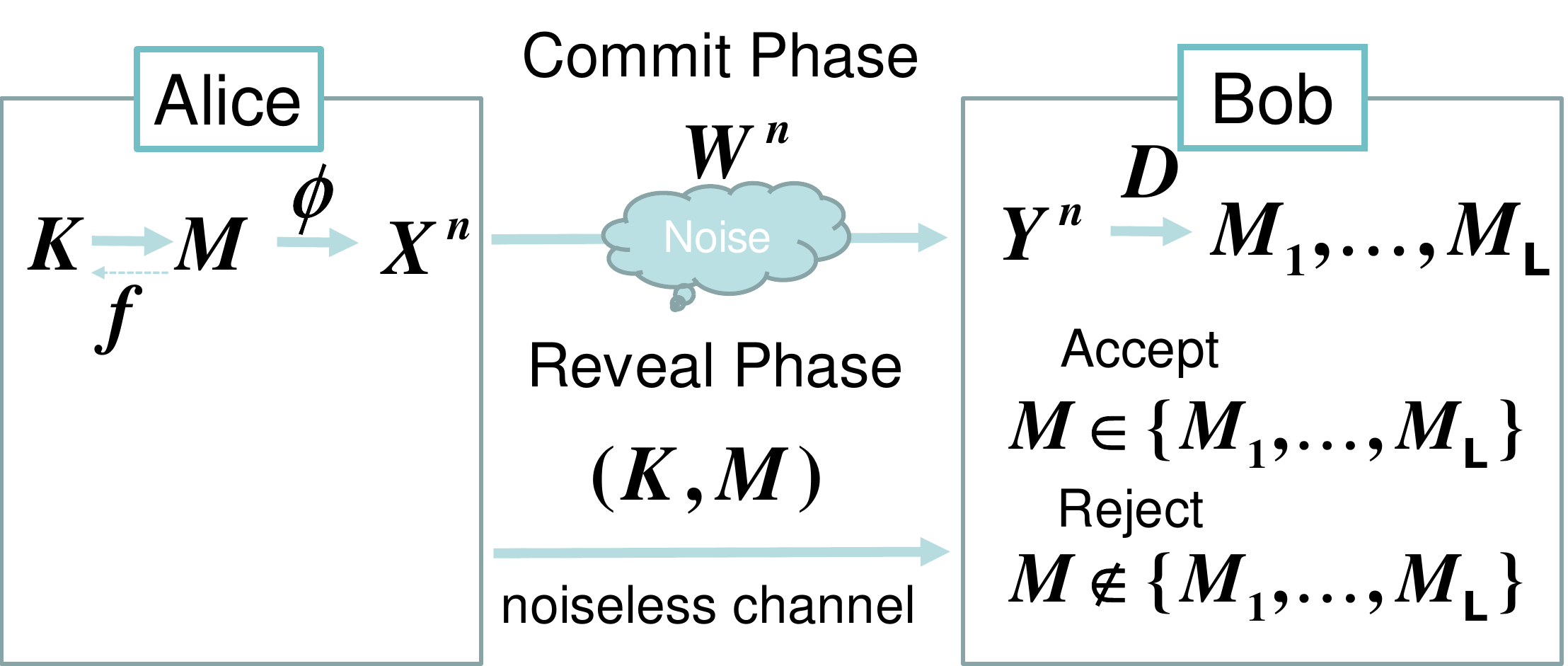}
  \end{center}
\caption{Our protocol for bit-string commitment with message set $ {\cal K}$.}
\Label{FF2}
\end{figure}   
%The paper \cite{Haya} showed the condition $R_1 - R_2 \le C(\bW)$ is needed for the verifiable condition (A).

The binding condition (BIN) is evaluated by the parameter 
$\delta_{C}(\phi,D)$, $\delta_D(D)$, or $\delta_{D'}(D)$.
To discuss the concealing condition (CON), for a deterministic encoder $\phi$ for secure list decoding,
we define the conditional distribution $P^{\phi,f}_{Y|K=k}$ 
and the distribution $P^{\phi,f}_{Y}$ on ${\cal Y}$ as 
\begin{align}
P^{\phi,f}_{Y|K=k}:=&
\sum_{m \in f^{-1}(k)} \frac{1}{|f^{-1}(k)|}W_{\phi(m)} \\
P^{\phi,f}_{Y}:=&
\sum_{m \in {\cal M}} \frac{1}{|{\cal M}|}W_{\phi(m)} .
\end{align}
When $\phi$ is given as a stochastic encoder by distributions $\{P_m\}_{m \in {\cal M}}$ on ${\cal X}$, 
these are defined as
\begin{align}
P^{\phi,f}_{Y|K=k}:=&
\sum_{m \in f^{-1}(k)} \frac{1}{|f^{-1}(k)|}\sum_{x\in {\cal X}} P_m(x) W_{x} \\
P^{\phi,f}_{Y}:=&
\sum_{m \in {\cal M}} \frac{1}{|{\cal M}|}\sum_{x\in {\cal X}} P_m(x) W_{x} .
\end{align}
The concealing condition (CON) is evaluated by the following quantity;
%when the message is restricted to a subset $\bar{\cal K} \subset {\cal K} $;
\begin{align}
\delta_{E}(f,\phi)
:=\max_{k,k' \in {\cal K}}
\frac{1}{2 } \| P^{\phi,f}_{Y|K=k} - P^{\phi,f}_{Y|K=k'} \|_1.\Label{AOY}
\end{align}
Therefore, we say that the tuple $(\phi,D,f)$
is a code for bit-string commitment based on secure list decoding.
Then, we have the following theorem, which is shown in Section \ref{S8B}.
\begin{thm}\Label{TH7}
For a code $(\phi,D)$ of secure list code with message set ${\cal M}$,
we assume that the size $\sM=|{\cal M}|=|(\phi,D)|_1$ is a power of a prime $p$, i.e., $\sM=p^{\sm}$.
Then, for an integer $\sk$ and 
a set ${\cal K}$ with $|{\cal K}|=p^{\sk}$,
there exist 
a subset $\bar{\cal K}\subset  {\cal K}$ 
with $|\bar{\cal K}|=p^{\sk-1}$,
a subset $\bar{\cal M}\subset {\cal M}$
with $|\bar{\cal K}|=p^{\sm-1}$,
and a regular hash function $f$
from ${\cal M}$ to ${\cal K}$ 
such that
$f(\bar{\cal M})=\bar{\cal K}$ 
and
\begin{align}
\delta_{E}(f,\phi|_{\bar{\cal M}})
\le \frac{3p}{p-1} 
p^{\frac{t\sk}{1+t}}2^{-\frac{t}{1+t} H_{1+t}(M|Y)  }. \Label{SO4}
\end{align}
\hfill $\square$\end{thm}

For a code $(\phi,D,f)$ for bit-string commitment based on secure list coding,
we define three parameters 
$|(\phi,D,f)|_1:=| (\phi,D ) |_1$,
$|(\phi,D,f)|_2:=|(\phi,D )|_2$,
and $|(\phi,D,f)|_3:=| \im f  |=|{\cal K}|$.
To discuss this type of code in the asymptotic setting, we make the following definitions.
A sequence of codes $\{(\phi_n,D_n,f_n)\}$ for bit-string commitment based on secure list coding
is called strongly (weakly, ultimately) secure
when $\epsilon_A(\phi_n,D_n)$, $\delta_{E}(f_n,\phi_n)$,
and $\delta_D(D_n)$ ($\delta_{C}(\phi_n,D_n)$, $\delta_{D'}(D_n)$) approach to zero.
A rate triple $(R_1,R_2,R_3)$ is strongly (weakly, ultimately) deterministically achievable 
for bit-string commitment based on secure list coding
when there exists a strongly (weakly, ultimately) secure deterministically
sequence of codes $\{(\phi_n,D_n,f_n)\}$
such that
$ \lim_{n \to \infty} \frac{1}{n}\log |(\phi_n,D_n,f_n)|_i=R_i$ for $i=1,2,3$.
We denote the set of strongly (weakly, ultimately) deterministically achievable rate triple 
$(R_1,R_2,R_3)$ for bit-string commitment based on secure list coding
by ${\cal R}_{(s,d)}^B$ (${\cal R}_{(w,d)}^B$, ${\cal R}_{(u,d)}^B$).
We define strongly (weakly, ultimately) stochastically achievable rate triple
for bit-string commitment based on secure list coding in the same way.
Then, we denote the set of strongly (weakly, ultimately) stochastically achievable rate triple 
$(R_1,R_2,R_3)$ for bit-string commitment based on secure list coding
by ${\cal R}_{(s,s)}^B$ (${\cal R}_{(w,s)}^B$, ${\cal R}_{(u,s)}^B$).
Then, we have
\begin{align}
{\cal R}_{(g,d)}^B \subset {\cal R}_{(g,s)}^B.\Label{COX}
\end{align}
for $g=s,w,u$. We obtain the following theorem under the above two settings.

\begin{thm}\Label{TH8}
(i)
Assume that the input alphabet ${\cal X}$ is discrete.
When Condition (W2) holds,
we have the following relations for $G\in \{(w,d),(s,d)\}$.
\begin{align}
\overline{{\cal R}_{G}^B}=\overline{{\cal C}}, \quad
\overline{{\cal R}_{(s,s)}^B}\subset \overline{{\cal C}^s}
\Label{TD42}.
\end{align}
(ii) 
Assume that the input alphabet ${\cal X}$ is continuous.
We choose a restricted finite subset $\tilde{\cal X}$ of the continuous input alphabet ${\cal X}$.
When the channel $\bW$ with $\tilde{\cal X}\subset {\cal X}$ satisfies
Conditions (W2) and (W3),
we have the following relations % for $K\in \{(u,s),(u,d)\}$.
\begin{align}
\overline{{\cal R}_{(u,d)}^B}=\overline{{\cal C}}, \quad
\overline{{\cal R}_{(u,s)}^B}\subset \overline{{\cal C}^s}
\Label{TD43}.
\end{align}
\hfill $\square$\end{thm}

Also, we define the optimal transmission rate in the above method as
\begin{align}
C^B_{G}:= \sup_{(R_1,R_2,R_3)\in {\cal R}_{G}^B} R_3
\end{align}
for $G \in \{ (s,d),(w,d),(u,d),(s,s),(w,s),(u,s)\}$.
Then, Lemma \ref{LL1}, Theorem \ref{TH8}, and \eqref{COX} imply the relation
\begin{align}
C^B_{G}=\sup_{P \in {\cal P}({\cal X})} H(X|Y)_P 
\end{align}
for $G\in \{ (s,d),(w,d),(u,d),(s,s),(u,s)\}$ under the same assumption 
as Theorem \ref{TH8}.
Here, we cannot determine only $C^B_{(w,s)}$
because the restriction for Alice is too weak in the setting $(w,s)$, i.e.,
Alice is allowed to use a stochastic encoder and 
Alice's cheating is not possible only when
Alice uses the correct encoder.
Fig. \ref{F-capacity} shows the numerical plot for %additive white Gaussian noise (
AWGN  channel with binary phase-shift keying (BPSK) modulation.

%The region $\overline{{\cal C}}$ contains a point with $R_3= \sup_{P \in {\cal P}({\cal X})} H(X|Y)_P$.
Since our setting allows the case with the continuous input and output systems,
Theorem \ref{TH8} can be considered as a generalization of the results by Winter et al \cite[Theorem 2]{BC1}, \cite{BC2} while a continuous generalization of their result was mentioned as an open problem in 
\cite{BC2}. 
Although the paper \cite{NBSI} addressed the Gaussian channel,
it considers only the special case when 
the cardinality of the input alphabet is infinitely many.
It did not derive a general capacity formula with a finite input alphabet and a continuous output alphabet.
At least, 
the paper \cite{NBSI} did not consider the case when
honest Alice accesses only a restricted finite subset $\tilde{\cal X}$ of the continuous input alphabet ${\cal X}$
and
dishonest Alice accesses the continuous input alphabet ${\cal X}$.

In addition to Theorem \ref{TH7},
to show Theorem \ref{TH8}, we prepare  the following lemma, which is shown in Section \ref{S8C}.
\begin{lem}\Label{KLR}
When a sequence of codes $\{(\phi_n,D_n,f_n)\}$ for bit-string commitment based on secure list coding
satisfies the condition $\delta_{E}(f_n,\phi_n)\to 0$,
we have
\begin{align}
\lim_{n \to \infty} \frac{1}{n}\log |(\phi_n,D_n,f_n)|_3 \le 
\lim_{n \to \infty} \frac{1}{n} E(\phi_n).\Label{KLRE}
\end{align}
\hfill $\square$
\end{lem}

\begin{proofof}{Theorem \ref{TH8}}
The converse part of Theorem \ref{TH8} follows from 
the combination of Theorem \ref{Converse} and Lemma \ref{KLR}, which is shown in Section \ref{S8C}.

The direct part of Theorem \ref{TH8} can be shown as follows.
For a given $\alpha>1$, 
the combination of Theorem \ref{TH7} and Corollary \ref{Cor46} implies 
${{\cal C}_\alpha} \subset {\cal R}_{G}^B$.
Taking the limit $\alpha \to 1$, we have
$\overline{\cal C} \subset \overline{{\cal R}_{G}^B}$.
In the same way, using Theorem \ref{TH7} and Corollary \ref{CorT}, we can show 
$\overline{\cal C} \subset \overline{{\cal R}_{(u,d)}^B}$.
\end{proofof}

\begin{figure}[t]%[htbp]
\begin{center}
%\scalebox{1}{
\includegraphics[scale=0.55]{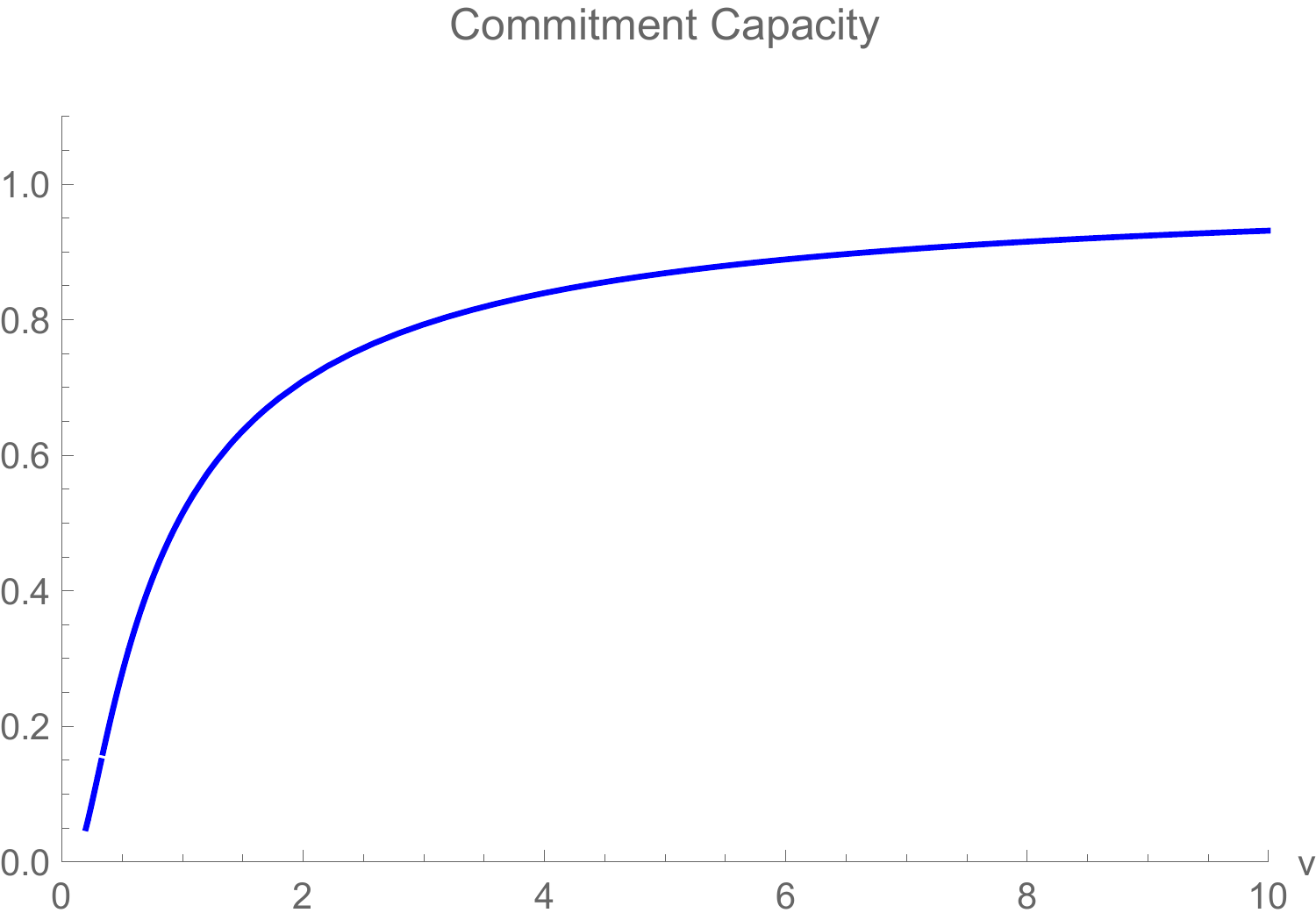}
\end{center}
\caption{
Numerical plot of the commitment capacity for 
AWGN channel with BPSK modulation. 
The vertical axis shows the commitment capacity,
and the horizontal axis shows the noise power of the 
AWGN channel.
$x \in \bF_2 \mapsto Y=(-1)^x+ N$, where 
$N$ subject to the Gaussian distribution with average $0$ and variance $v$. 
}
\Label{F-capacity}
\end{figure}%

%Combining (ii) of Theorem \ref{TH3} and Theorem \ref{TH7}, we obtain the following 
 \subsection{Randomized construction (Proof of Theorem \ref{TH7})}\Label{S8B}
To show Theorem \ref{TH7},
we treat the set of messages ${\cal M}$ as a vector space $\bF_p^{\sm}$ over the finite field $\bF_p$.
%Then, based on a fixed code $(\phi,D)$,
%we make a construction of a protocol for bit-string commitment
%whose message set is $\bF_p^{\sk}$, where $\sk \le \sm$.
%For this construction, we consider 
For a linear regular hash function $f$ from $\bF_p^{\sm}$ to 
${\cal K}:=\bF_p^{\sk}$ and a code $\phi $,
%To show Theorem \ref{TH7},
%For this proof, 
we define the following value; % for a linear hash function $f$ and a code $\phi $
\begin{align}
\bar\delta_E( f,\phi )
%:=&\frac{1}{2}\| P_{K, Y}^\phi - P_{K,\uni}\times P_Y^\phi \|_1 \\
:=\sum_{k \in {\cal K}}\frac{1}{2 |{\cal K}|} \| P^{\phi,f}_{Y|K=k} - P^{\phi,f}_{Y} \|_1 
\ge \frac{1}{2}\delta_E( f,\phi ),\Label{MOH}
\end{align}
where the inequality follows from the triangle inequality.
We denote the joint distribution of $K$ and $Y$ by $P_{K, Y}^{\phi,f}$ 
when $K$ is assumed to be subject to the uniform distribution on ${\cal K}$.
Then, the definition of $\bar\delta_E( f,\phi )$ is rewritten as
\begin{align}
\bar\delta_E( f,\phi )
=\frac{1}{2}\| P_{K, Y}^{\phi,f} - P_{K,\uni}\times P_Y^{\phi,f} \|_1.
\end{align}

In the following, we employ a randomized construction.
That is, we randomly choose a linear regular hash function $f_S$ from 
$\bF_p^{\sm}$ to $\bF_p^{\sk}$, where 
$S$ is a random seed to identify the function $f_S$.
A randomized function $f_S$ is called a universal2 hash function
when the collision probability satisfies the inequality 
\begin{align}
\Pr \{ f_S(m)=f_S(m')\} \le p^{-\sk}
\Label{HTVD}
\end{align}
for any distinct elements $m\neq m' \in \bF_p^{\sm}$ \cite{Carter,WC81}.

When $K$ is subject to the uniform distribution on ${\cal K}$, 
the stochastic behavior of $K$ can be simulated as follows.
First, $M$ is generated according to the uniform distribution on ${\cal M}$.
Then, %we apply a linear hash function $f_s$ to $M$.
the obtained outcome $K=f_s(M)$ of $f_s$ with a fixed $s$ is subject to the uniform distribution on ${\cal K}$.
When $f_S$ is a universal2 hash function with a variable $S$,
the R\'{e}nyi conditional entropy version of universal hashing lemma 
\cite[(67)]{Tight}\cite[Lemma 27]{Ran-Mar} \cite[Proposition 21]{epsilon}
implies that
\begin{align}
\rE_S \delta_E( f_S,\phi )
\le \frac{3}{2} |{\cal K}|^{\frac{t}{1+t}}2^{-\frac{t}{1+t} H_{1+t}(M|Y)  }.
\end{align}
Hence, there exists an element $s$ such that
\begin{align}
\delta_E( f_s,\phi )
\le \frac{3}{2} |{\cal K}|^{\frac{t}{1+t}}2^{-\frac{t}{1+t} H_{1+t}(M|Y)  }. \Label{SO1}
\end{align}
Due to Markov inequality, there exists a subset $\bar{\cal K}\subset {\cal K}$ with cardinality
$|{\cal K}|/p $ such that
any element $k \in \bar{\cal K}$ satisfies that
\begin{align}
\frac{1}{2 } \| P^{\phi,f_s}_{Y|K=k} - P^{\phi,f_s}_{Y} \|_1
\le \frac{p}{p-1}\delta_E( f_s,\phi ).\Label{APL}
\end{align}
This is because the number of elements that does not satisfy \eqref{APL} is upper bounded by 
$\frac{p-1}{p}|{\cal K}| $.
Hence, any elements $k,k' \in \bar{\cal K}$ satisfy that
\begin{align}
\frac{1}{2 } \| P^{\phi,f_s}_{Y|K=k} - P^{\phi,f_s}_{Y|K=k'} \|_1
\le \frac{2p}{p-1}\delta_E( f_s,\phi ).\Label{SO2}
\end{align}
The combination of \eqref{SO1} and \eqref{SO2} imply that
any elements $k,k' \in \bar{\cal K}$ satisfy that
\begin{align}
\frac{1}{2 } \| P^{\phi,f_s}_{Y|K=k} - P^{\phi,f_s}_{Y|K=k'} \|_1
\le \frac{3p}{p-1} |{\cal K}|^{\frac{t}{1+t}}2^{-\frac{t}{1+t} H_{1+t}(M|Y)  }.
\Label{SO3}
\end{align}
Choosing $\bar{\cal M}$ to be $f_s^{-1}(\bar{\cal K})$, we find that 
\eqref{SO3} is the same as \eqref{SO4} due to the definition \eqref{AOY}.
\hspace*{\fill}~\QED\par\endtrivlist\unskip

\subsection{Proof of Lemma \ref{KLR}}\Label{S8C}
To show Lemma \ref{KLR}, we prepare the following proposition.
\begin{proposition}[\protect{\cite[Lemma 30]{Ran-Mar}}]\Label{PO1} 
Any function $f$ defined on ${\cal M}$ and a joint distribution on ${\cal M}\times {\cal Y}$
satisfy the following inequality
\begin{align}
&\frac{1}{2}\|P_{f(M)Y}- P_{f(M)}\times P_{Y}\|_1\nonumber \\
 \ge & \sup_{\gamma \ge 0}\left[ P_{MY}\left\{ \log \frac{1}{P_{M|Y}(m|y)} < \gamma \right\} - \frac{2^{\gamma}}{|\im f  |} \right].
\end{align}
%for any $\gamma>0$.
\hfill $\square$
\end{proposition}

We focus on the joint distribution $P_{MY}$ when Alice generates $M$ according to the uniform distribution on ${\cal M}$ and chooses $X^n$ as $\phi(M)$.
Let $p$ be the probability 
$P_{MY}\left\{ \log \frac{1}{P_{M|Y}(m|y)} < \gamma \right\} $.
Then, the conditional entropy $H(M|Y)$ is lower bounded as
\begin{align}
H(M|Y) \ge \gamma (1-p).\Label{AK0}
\end{align}
The quantity $\delta_{E}(f,\phi)$ is evaluated as
\begin{align}
&\delta_{E}(f,\phi)
=\max_{k,k' \in {\cal K}}
\frac{1}{2 } \| P^{\phi,f}_{Y|K=k} - P^{\phi,f}_{Y|K=k'} \|_1
\nonumber \\
\ge &
\sum_{k, k' \in {\cal K}}
\frac{1}{2|{\cal K}|^2 } \| P^{\phi,f}_{Y|K=k} - P^{\phi,f}_{Y|K=k'} \|_1 \nonumber \\
\ge &
\sum_{k\in {\cal K}}
\frac{1}{2|{\cal K}| } \| P^{\phi,f}_{Y|K=k} - P_{Y} \|_1 \nonumber \\
=&\frac{1}{2}\|P_{f(M)Y}- P_{f(M)}\times P_{Y}\|_1 
\stackrel{(a)}{\ge} p- \frac{2^{\gamma}}{|\im f  |},
\end{align}
where $(a)$ follows from Proposition \ref{PO1}.
Hence, we have $\delta_{E}(f,\phi)+\frac{2^{\gamma}}{|\im f  |}\ge p$.
Applying this relation to \eqref{AK0}, we have
\begin{align}
H(M|Y) \ge \gamma \big(1- \delta_{E}(f,\phi)-\frac{2^{\gamma}}{|\im f  |}\big).
\end{align}
Therefore, 
\begin{align}
\gamma \big(1- \delta_{E}(f,\phi)-\frac{2^{\gamma}}{|(\phi,D,f)|_3}\big) \le E(\phi).
\end{align}
Choosing $\gamma= \log |(\phi_n,D_n,f_n)|_3-\sqrt{n}$, we have
\begin{align}
&(\log |(\phi_n,D_n,f_n)|_3-\sqrt{n}) (1- \delta_{E}(f_n,\phi_n)-2^{-\sqrt{n}})\nonumber \\
 \le & E(\phi_n).
\end{align}
Dividing the above by $n$ and taking the limit, we have \eqref{KLRE}.
\hspace*{\fill}~\QED\par\endtrivlist\unskip

\section{Proof of Converse Theorem}\Label{S-C}
In order to show Theorem \ref{Converse}, we prepare the following lemma.
\begin{lem}\Label{L4}
For $X^n=(X_{1}, \ldots, X_{n})$, 
we choose the joint distribution $P_{X^n}$.
Let $Y^n=(Y_{1}, \ldots, Y_{n})$ be the channel output variables of the inputs $X^n$ via the channel $\bW$.
Then, using the chain rule, we have
\begin{align}
I(X^n;Y^n) & 
= \sum_{j=1}^n I( X_{j};Y_j |Y^{j-1})
%\le \sum_{j=1}^n I( X_{j};Y_j |X^{j-1})
,\Label{LLP3}\\
H(X^n) & \le \sum_{j=1}^n H( X_{j}|Y^{j-1}).\Label{LLP2}
\end{align}
\hfill $\square$
\end{lem}
The proof of Lemma \ref{L4} is given in Appendix \ref{A-L4}.

\noindent{\it Proof of Theorem \ref{Converse}:}\quad
The proof of Theorem \ref{Converse} is composed of two parts.
The first part is the evaluation of $R_1$.
The second part is the evaluation of $R_1-R_2$.
The key point of the first part is the use of \eqref{LLP2} in Lemma \ref{L4}.
The key point of the second part is the meta converse for list decoding \cite[Section III-A]{Haya}. 

\noindent{\bf Step 1:} Preparation.

\noindent We show Theorem \ref{Converse} by showing the following relations;
\begin{align}
{\cal R}_{(w,d)}^L  &\subset \overline{\cal C}
%\cup_{P\in {\cal P}({\cal U}\times {\cal X})} 
%\{ (R_1,R_2,R_3) | 0 \le R_1-R_2\le I(X;Y|U)_P , ~R_1 \le H(X|U)_P ,~0\le R_1,~0\le R_2\}  
 ,\Label{Con1} \\
{\cal R}_{(s,s)}^L &\subset  \overline{{\cal C}^s}
%\cup_{P\in {\cal P}({\cal U}\times {\cal X})} 
%\{ (R_1,R_2,R_3) | 0 \le R_1-R_2\le I(X;Y|U)_P , ~R_1 \le H(X|U)_P ,~0\le R_1,~0\le R_2\} 
.\Label{Con2B} 
\end{align}
because ${\cal R}_{(s,d)} \subset  \overline{\cal C}$ follows from \eqref{Con1}.
Assume that a sequence of deterministic codes $\{(\phi_n,D_n)\}$ is 
weakly secure.
We assume that 
$R_i:=\lim_{n\to \infty}\frac{1}{n}\log |(\phi_n,D_n)|_i$ converges for $i=1,2$.
For the definition of $|(\phi_n,D_n)|_i$,
see the end of Section \ref{S2-1B}.
Also, we assume that $R_3 \le \lim_{n\to \infty} \frac{1}{n} E(\phi_n)$.

Letting $M$ be the random variable of the message, 
we define the variables $X^n=(X_1, \ldots, X_n):= \phi_n(M)$. 
The random variables $Y^n=(Y_1, \ldots, Y_n)$ are defined as the output of the channel $\bW^n$, 
which is the $n$ times use of the channel $\bW$.
Choosing the set ${\cal U}:=\cup_{i=1}^n \{i\} \times {\cal Y}^{i-1}$,
we define the joint distribution $P_n \in {\cal P}({\cal U}\times \cX)$ as follows; 
$p_n(x,u):= \frac{1}{n} p_{Y^{i-1},X}(y^{i-1},x)$ for  $u=(i, y^{i-1})$.

%Hence, the variable $U$ is subject to the uniform distribution on $\{1, \ldots, n\}$.
Under the distribution $P_{n}$, we denote the channel output by $Y$.
%Then, we have $I(X;Y|U)$ and $H(X|U)$.
%When we need to describe the dependence of $n$, we denote them by $I_n(X;Y|U)$ and $H_n(X|U)$.
In this proof, 
we use the notations $\sM_n: =|(\phi_n,D_n)|_1$ and $\sL_n:=|(\phi_n,D_n)|_2 $.
Also, instead of $\epsilon_{A}(\phi_n,D_n)$, we employ
$\epsilon_{A}'(\phi_n,D_n):= \sum_{m=1}^{\sM_n} \frac{1}{\sM_n}
\epsilon_{A,m}(\phi_n(m),D_n) $, which goes to zero. 

\noindent{\bf Step 2:} Evaluation of $R_1$.

\noindent %Let $M$ be the random variable to describe the message.
When a code $(\phi_n,D_n)$ satisfies 
$\delta_C(\phi_n,D_n)\le 1-\epsilon_{A}(\phi_n,D_n)$,
we have
\begin{align}
\log |(\phi_n,D_n)|_1
\stackrel{(a)}{\le}&
  H(X^n)
%+\frac{\epsilon_{A}'(\phi_n,D_n)}{1-\delta_C(\phi_n,D_n)} \log |(\phi_n,D_n)|_1
+\log 2 \nonumber \\
\stackrel{(b)}{\le}&
n H(X|U)_{P_n}
%+\frac{\epsilon_{A}'(\phi_n,D_n)}{1-\delta_C(\phi_n,D_n)}\log |(\phi_n,D_n)|_1
+\log 2 \Label{MGP},
\end{align}
where $(b)$ follows from \eqref{LLP2} in Lemma \ref{L4} and the variable $U$ is defined in Step 1.
Dividing the above by $n$ and taking the limit, we have
\begin{align}
\limsup_{n\to \infty}R_1 - H(X|U)_{P_n} \le 0 \Label{MGP2}.
\end{align}

To show $(a)$ in \eqref{MGP}, we consider the following protocol.
After converting the message $M$ to $X^n$ by the encoder $\phi_n(M)$,
Alice sends the $X^n$ to Bob $\sK$ times.
Here, we choose $\sK$ to be an arbitrary large integer.
Applying the decoder $D_n$, Bob obtains $\sK$ lists that contain up to $\sK \sL_n $ messages.
Among these messages, Bob chooses  
$\hat{M}$ as the element that most frequently appears in the $\sK$ lists. 
When $ \delta_C(\phi_n,D_n)<1- \epsilon_{A,M}(\phi_n(M),D_n)$, 
the element $M$ has the highest probability to be contained in the list.
In this case, when $\sK$ is sufficiently large,
Bob can correctly decode $M$ by this method
because $1- \epsilon_{A,M}(\phi_n(M),D_n)$ is the probability that the list contains $M$
and $ \delta_C(\phi_n,D_n)$ is the maximum of the probability that the list contains $m'\neq M$.
Therefore, 
when $\delta_C(\phi_n,D_n)\le 1-\epsilon_{A}(\phi_n,D_n)$,
the probability $\epsilon_{\sK}$ of the failure of decoding goes to zero as $\sK\to \infty$.
%is limited to the case when 
%$ 1-\delta_C(\phi_n,D_n) \le \epsilon_{A,M}(\phi_n(M),D_n)$.
%Since the average of $\epsilon_{A,M}(\phi_n(M),D_n) $ is $\epsilon_{A}'(\phi_n,D_n)$,
%Markov inequality guarantees that the error probability of this protocol is bounded by 
%$\frac{\epsilon_{A}'(\phi_n,D_n)}{1-\delta_C(\phi_n,D_n)}$.
Fano inequality shows that 
$H(M|\hat{M}) \le \epsilon_{\sK} \log |(\phi_n,D_n)|_1+\log 2$.
Then, we have
\begin{align}
&\log |(\phi_n,D_n)|_1- \epsilon_{\sK} \log |(\phi_n,D_n)|_1-\log 2
 \nonumber \\
 \le&  \log |(\phi_n,D_n)|_1- H(M|\hat{M}) \nonumber \\
 =& 
 I(M;\hat{M}) \le I(M;X^n) \Label{AL1} \\
 \le & H(X^n),
\end{align}
which implies $(a)$ in \eqref{MGP} with the limit $\sK\to \infty$.

\noindent{\bf Step 3:} Evaluation of $R_1-R_2$.

\noindent Now, we consider the hypothesis testing with two distributions 
$P(m,y^n):=\frac{1}{\sM_n} W^n(y^n| \phi_n(m))$ 
and $Q(m,y^n):=\frac{1}{\sM_n^2} \sum_{m'=1}^{\sM_n} W^n(y^n| \phi_n(m'))$
on ${\cal M}_n\times {\cal Y}^n$, where 
${\cal M}_n:=\{1, \ldots,  \sM_n\}$.
Then, we define the region ${\cal D}_n^*\subset {\cal M}_n\times {\cal Y}^n $
as $\cup_{m_1, \ldots, m_{\sL_n}} \{ m_1, \ldots, m_{\sL_n}\} \times {\cal D}_{m_1, \ldots, m_{\sL_n}}$.
Using the region ${\cal D}_n^*$ as our test, we define 
$\epsilon_Q$ as the error probability to incorrectly support $P$ while the true is $Q$.
Also, we define $\epsilon_P$ as the error probability to incorrectly support $Q$ while the true is $P$.
When we apply the monotonicity for the KL divergence between $P$ and $Q$, 
 dropping the term $\epsilon_{P}\log (1- \epsilon_Q ) $,
we have
\begin{align}
-\log \epsilon_Q 
\le \frac{D(P\|Q)+h(1-\epsilon_{P})}
{1-\epsilon_{P}},
\Label{MK5}
\end{align}
where $h$ is the binary entropy, i.e., 
$h(p):= -p \log (p) -(1-p)\log (1-p)$.
The meta converse for list decoding \cite[Section III-A]{Haya} 
shows that $\epsilon_Q \le \frac{|(\phi_n,D_n)|_2}{|(\phi_n,D_n)|_1} $
and $\epsilon_P\le \epsilon_{A}(\phi_n,D_n)$.
Since \eqref{LLP2} in Lemma \ref{L4} guarantees that
$D(P\|Q)= I(X^n;Y^n)= n I(X;Y|U)_{P_n}$, 
the relation \eqref{MK5} is converted to 
\begin{align}
& \log \frac{|(\phi_n,D_n)|_1}{|(\phi_n,D_n)|_2}
\le \frac{I(X^n;Y^n)+h(1-\epsilon_P)}
{1-\epsilon_P} \nonumber \\
\le &\frac{n I(X;Y|U)_{P_n}+h(1-\epsilon_{A}(\phi_n,D_n))}
{1-\epsilon_{A}(\phi_n,D_n)} 
\Label{MK1}
\end{align}
under the condition that $\epsilon_{A}(\phi_n,D_n) \le \frac{1}{2}$.
Dividing the above by $n$ and taking the limit, we have
\begin{align}
\limsup_{n\to \infty}R_1-R_2 - I(X;Y|U)_{P_n} \le 0 \Label{MK12}.
\end{align}

\noindent{\bf Step 4:} Evaluation of $R_3$.

\noindent 
Since the code $\phi_n$ is deterministic, 
remembering the definition of the variable $U$ given in Step 1,
we have
\begin{align}
&\log |(\phi_n,D_n)|_1- E(\phi_n)
=  H(M)-H(M|Y^n)\nonumber \\
=& I(M;Y^n) =I(X^n;Y^n)
=n I(X;Y| U)_{P_n} .
\end{align}
Dividing the above by $n$ and taking the limit, we have
\begin{align}
R_1 - R_3 \ge \limsup_{n\to \infty} 
I(X;Y| U)_{P_n}  \Label{MK13}.
\end{align}
Therefore, combining Eqs. \eqref{MGP2}, \eqref{MK12}, and \eqref{MK13},
we obtain Eq. \eqref{Con1}.

\noindent{\bf Step 5:} Proof of Eq. \eqref{Con2B}.

\noindent %Let $M$ be the random variable to describe the message.
Assume that a sequence of stochastic codes $\{(\phi_n,D_n)\}$ is 
strongly secure.
Then, there exists a sequence of deterministic encoders $\{\phi_n'\}$
such that
$\epsilon_{A}(\phi_n',D_n) \le \epsilon_{A}(\phi_n,D_n)$
and
$\delta_C(\phi_n',D_n) \le \delta_{D}(D_n)$.
Since $\epsilon_{A}(\phi_n',D_n)$ and 
$\delta_C(\phi_n',D_n)$  go to zero, 
we have Eqs. \eqref{MGP2} and \eqref{MK12}.
However, the derivation of \eqref{MK13} does not hold in this case.
Since the code is stochastic, 
the equality $I(M;Y^n) =I(X^n;Y^n)$ does not hold in general.

Instead of \eqref{MK13}, we have the following derivation.
Taking the limit $\sK\to \infty$ in \eqref{AL1}, we have
\begin{align}
\log |(\phi_n,D_n)|_1 -\log 2
 \le I(M;X^n).
 \end{align}
Hence, 
\begin{align}
&I(X^n;Y^n)=
I(X^n M ;Y^n)\nonumber \\
=&I(M;Y^n)+ I(X^n;Y^n|M) \nonumber \\
\le & I(M;Y^n)+ H(X^n|M)\nonumber \\
=& I(M;Y^n)+ H(X^n)-I(X^n;M) \nonumber \\
\le &  I(M;Y^n)+ H(X^n)-\log |(\phi_n,D_n)|_1 +\log 2 \nonumber \\
= &  H(M)\!-\! H(M|Y^n)+ H(X^n)\! -\! \log |(\phi_n,D_n)|_1 +\log 2\nonumber  \\
= &  \log |(\phi_n,D_n)|_1-\log |(\phi_n,D_n)|_3\nonumber \\
&+ H(X^n)-\log |(\phi_n,D_n)|_1 +\log 2\nonumber  \\
= & - \log |(\phi_n,D_n)|_3
+ H(X^n) +\log 2 .
\end{align}
Hence,
we have
\begin{align}
 &\log |(\phi_n,D_n)|_3
\le H(X^n) +\log 2 -I(X^n;Y^n)\nonumber \\
= &H(X^n|Y^n) +\log 2
=n H(X|Y U)_{P_n} +\log 2
\end{align}
Dividing the above by $n$ and taking the limit, we have
\begin{align}
R_3 \le \liminf_{n\to \infty} 
H(X|Y U)_{P_n}  \Label{MK14}.
\end{align}
Therefore, combining Eqs. \eqref{MGP2}, \eqref{MK12}, and \eqref{MK14},
we obtain Eq. \eqref{Con2B}.
\endproof

%\begin{align}
%\frac{1}{s}\log (1-\epsilon_{A}(\phi_n,D_n))
%\le I_{1+s} (X^n;Y^n) +\frac{1}{1-s}\log \frac{|(\phi_n,D_n)|_1}{|(\phi_n,D_n)|_2}.
%\end{align}

%\begin{proofof}{Theorem \ref{TH}}
%\end{proofof}

\section{Proof of direct theorem}\Label{S-K}
As explained in Section \ref{OUT-K}, we show only the second part (ii)
based on the random coding.
First, we show Lemma \ref{LS3}.
Then, using Lemma \ref{LS3}, we show the second part (ii) by preparing 
various lemmas, 
Lemmas \ref{LL12}, \ref{LL10}, \ref{LL11} and \ref{LL19}.
Using Lemmas \ref{LL10}, and \ref{LL11}, we 
extract an encoder $\phi_n$ and messages $m$ 
that have a small decoding error probability and satisfy two conditions, which will be stated as 
the conditions \eqref{CC2} and \eqref{ZSO}.
Then, using these two conditions, 
we show that the code satisfies 
the binding condition for dishonest Alice (D)
and the equivocation version of concealing condition (B).
In particular, Lemma \ref{LL12} is used
to derive 
the binding condition for
dishonest Alice (D).

\subsection{Proof of Lemma \ref{LS3}}\Label{PfL1}
\noindent{\bf Step 1:} %Preparation. %We prepare a lemma for our proof of Lemma \ref{LS3}.
%\noindent %Let $M$ be the random variable to describe the message.
For our proof of Lemma \ref{LS3}, we prepare the following lemma.
\begin{lem}\Label{VGT}
Let ${\cal S}$ be a closed convex subset of ${\cal P}({\cal Y})$.
Assume that a distribution $P \in {\cal P}({\cal Y})\setminus {\cal S}$
has the full support ${\cal Y}$.
We choose $P'$ as
\begin{align}
P':= \argmin_{Q \in {\cal S}} D(Q\|P).
\end{align}
(i) We have $\Supp(Q)\subset \Supp(P') $ for $Q \in {\cal S}$.
(ii) For $Q \in {\cal S}$,
we have
\begin{align}
D(P'\|P)\le \mathbb{E}_{Q} [  \log p'(Y) -\log p(Y) ].\Label{Eqa}
\end{align}
\hfill $\square$
\end{lem}

\begin{proof}
Now, we show (i) by contradiction.
We choose $Q \in {\cal S}$ such that $\Supp(Q)\not\subset \Supp(P') $.
We define the distribution 
$\bar{P}_t:= t Q+ (1-t)P'$.
Then, we have
\begin{align}
D(\bar{P}_t\|P)=
\sum_{y \in {\cal Y}} (\eta(\bar{p}_t(y) )-  \bar{p}_t(y)\log p(y)),
\end{align}
where $\eta(x):=x \log x$.
The derivative of $\sum_{y \in {\cal Y}}  \bar{p}_t(y)\log p(y)$ for $t$ at $t=0$ is
a finite value.
For $ y \in \Supp(P')$,
the derivative of $\eta(\bar{p}_t(y) )$ for $t$ at $t=0$ is a finite value.
For $ y \in \Supp(Q)\setminus\Supp(P')$,
the derivative of $\eta(\bar{p}_t(y) )$ for $t$ at $t=0$ is $-\infty$.
Hence, 
the derivative of $D(\bar{P}_t\|P)$ for $t$ at $t=0$ is $-\infty$.
It means that there exist a small real number $t_0>0$ such that
$D(\bar{P}_t\|P)\le D(\bar{P}_0\|P)=D(P'\|P)$.
Hence, we obtain a contradiction.

Next, we show (ii).
Theorem 11.6.1 of \cite{COVER} shows the following.
\begin{align}
D(Q\|P')+ D(P'\|P) \le D(Q\|P), 
\end{align}
which implies
\begin{align}
&D(P'\|P) \le D(Q\|P) -D(Q\|P')\nonumber \\
=& \mathbb{E}_{Q} [  \log p'(Y) -\log p(Y) ].
\end{align}
Hence, we obtain \eqref{Eqa}.
\end{proof}

\noindent{\bf Step 2:} We prove Lemma \ref{LS3} when ${\cal Y}$ is a finite set and the support of $W_x$ does not depend on 
$x\in {\cal X}$.

For $x\in {\cal X}$, we define the distribution $P_x \in {\cal P}({\cal X}\setminus \{x\})$ as
\begin{align}
P_x:= \argmin_{P \in {\cal P}({\cal X}\setminus \{x\})}
D\bigg( \sum_{x' \in {\cal X}\setminus \{x\}}P(x')W_{x'} \bigg\|W_x\bigg) \Label{IMC}
\end{align}
%\noindent %Let $M$ be the random variable to describe the message.
%To prove Lemma \ref{LS3} to the above case, we prepare the following lemma;
We choose $\xi_x$ as $\xi_x(y):=\log w_x(y)-\log w_{P_x}(y)-D(W_x\|W_{P_x})$, which satisfies \eqref{CS1}.
Applying (ii) of Lemma \ref{VGT} to the case when ${\cal S}$ is 
$\{ \sum_{x'' \in {\cal X}\setminus \{x\}}P(x'') W_{x''}\}_{P \in {\cal P}({\cal X}\setminus \{x\})}$,
we have
\begin{align}
\zeta_1&= \min_{x\neq x' \in {\cal X}} 
\mathbb{E}_{x'}[\log w_{P_x}(y)-\log w_x(y)]+D(W_x\|W_{P_x})\nonumber  \\
& \ge \min_{x\neq x' \in {\cal X}} D(W_{P_x}\|W_x)
+D(W_x\|W_{P_x})
\nonumber \\
=& \min_{x \in {\cal X}} D(W_{P_x}\|W_x)+D(W_x\|W_{P_x})>0.
\end{align}
Hence, it satisfies \eqref{CS2}.
Since the support of $W_x$ does not depend on $x\in {\cal X}$,
the function $\xi_x$ takes a finite value.
Since ${\cal Y}$ is a finite set, $\max_{x,y}\xi_x(y)$ exists.
Thus, it satisfies \eqref{CS3}.

\noindent{\bf Step 3:} We prove Lemma \ref{LS3}  when ${\cal Y}$ is a finite set  and the support of $W_x$ depends on 
$x\in {\cal X}$.

\noindent %Let $M$ be the random variable to describe the message.
For an element $x \in {\cal X}$ and a small real number $\delta>0$, we define $W_{x,\delta}$ as
\begin{align}
w_{x,\delta}(y):=
\left\{
\begin{array}{ll}
(1-\delta)w_x(y) & \hbox{ for } y \in \Supp(W_x) \\
\frac{\delta}{|\Supp(W_x)|^c} & \hbox{ for } y \in \Supp(W_x)^c,
\end{array}
\right.
\end{align}
where $\Supp(P)$ is the support of the distribution $P$.
We define
\begin{align}
P_{x,\delta}:=\argmin_{P \in {\cal P}({\cal X}\setminus \{x\})}D(W_{P} \|W_{x,\delta} ).
\end{align}
We choose $\delta>0$ to be sufficiently small such that
\begin{align}
%\min_{P \in {\cal P}({\cal X}\setminus \{x\})}
D(W_{P_{x,\delta}} \|W_{x,\delta} )&>0 \\
\log (1-\delta)+ 
\min_{P \in {\cal P}({\cal X}\setminus \{x\})}D(W_x\|W_{P}) &>0
\end{align}
for any $x \in {\cal X}$.

When 
$\Supp(W_x) \subset \cup_{x' \in {\cal X}\setminus \{x\}} \Supp(W_{x'})$,
we have 
$\Supp(W_x) \subset \Supp(P_{x,\delta})$ due to (i) of Lemma \ref{VGT}.
Then, 
\begin{align}
&\mathbb{E}_x[\log w_{x,\delta}(Y)-\log w_{P_{x,\delta}}(Y)]\nonumber \\
=&
D(W_x\|W_{P_{x,\delta}})+\log (1-\delta) \nonumber \\
\ge & \log (1-\delta)+ 
\min_{P \in {\cal P}({\cal X}\setminus \{x\})}D(W_x\|W_{P})>0.
\end{align}
Then, we define $\xi_x$ as
\begin{align}
\xi_{x}(y):=&
\log w_{x,\delta}(y)-\log w_{P_{x,\delta}}(y) \nonumber \\
&-
\mathbb{E}_x[\log w_{x,\delta}(Y)-\log w_{P_{x,\delta}}(Y)].
\end{align}
Then, we have
\begin{align}
&\mathbb{E}_x[\xi_x(Y)] =0,\Label{CS1H}\\
&\min_{x' \in {\cal X}\setminus \{x\}} \mathbb{E}_{x'}[-\xi_x(Y)] >0, \Label{CS2H}\\
&\max_{x' \in {\cal X}\setminus \{x\}} \mathbb{V}_{x'}[\xi_x(Y)] < \infty.\Label{CS3H}
\end{align}

When 
$\Supp(W_x) \not\subset \cup_{x' \in {\cal X}\setminus \{x\}} \Supp(W_{x'})$,
we have 
$\Supp(P_{x,\delta})= \cup_{x' \in {\cal X}\setminus \{x\}} \Supp(W_{x'})$ due to (i) of Lemma \ref{VGT}
because $W_{x,\delta}$ has the full support ${\cal Y}$.
Then, we define $\xi_x$ as
\begin{align}
\xi_{x}(y):=
\log w_{x,\delta}(y)-\log w_{P_{x,\delta}}(y) 
\end{align}
 for $ y \in \Supp(P_{x,\delta})$, and 
\begin{align}
&\xi_{x}(y)\nonumber\\
:=&
-\frac{
{\displaystyle\sum_{y \in \Supp(P_{x,\delta})}} w_x(y) (\log w_{x,\delta}(y)-\log w_{P_{x,\delta}}(y))
}{W_x(\Supp(P_{x,\delta})^c)}
\end{align}
for $ y \in \Supp(P_{x,\delta})^c$.
Then, we have \eqref{CS1H}, \eqref{CS2H}, and \eqref{CS3H}.
Therefore, our functions $\{ \xi_x\}_{x \in {\cal X}}$ satisfy the conditions 
\eqref{CS1}, \eqref{CS2}, and \eqref{CS3}.

\noindent{\bf Step 4:} We prove Lemma \ref{LS3}  when ${\cal Y}$ is not a finite set.
Since the channel $\bW$ satisfies Condition (W2), 
there exists a map $f$ from ${\cal Y}$ to a finite set ${\cal Y}_0$ such that
the channel $\bW\circ f^{-1}= \{ W_x\circ f^{-1} \}_{x \in {\cal X}}$ satisfies Condition (W2), 
where
$W_x\circ f^{-1}( \{y_0\}):= W_x (f^{-1} \{y_0\})$ for $y_0 \in {\cal Y}_0$.
Applying the result of Step 3 to the channel $\bW\circ f^{-1}$,
we obtain functions $\{\xi_{x,0}\}_{x \in {\cal X}}$ defined on ${\cal Y}_0$.
Then, for $x\in{\cal X}$,  we choose a function $\xi_x$ on ${\cal Y}$
as $ \xi_x(y):= \xi_{x,0}(f(y))$.
The functions $\{\xi_{x}\}_{x \in {\cal X}}$ satisfy the conditions \eqref{CS1}, \eqref{CS2}, and \eqref{CS3}.

\subsection{Preparation}
To show Theorem \ref{TH3}, we prepare notations and information quantities.
For $P \in {\cal P}({\cal X})$ and $t>0$, we define 
\begin{align}
G_{P|x}(t)&:= \log (2^t P(x)+1-P(x)) \\
G_{P,P'}(t)&:=\sum_{x\in{\cal X}}P'(x) \log (2^t P(x)+1-P(x)) .
%G_{P}^\delta(s)&:= \max_{P' \in U_{\epsilon,P}}G_{P,P'}(s).
\end{align}
Then, we have the Legendre transformation
\begin{align}
L[G_{P,P'}](r):=\min_{t>0} G_{P,P'}(t)-t r.
\end{align}
Using the $\epsilon$-neighborhood $U_{\epsilon,P}$ of $P$ with respect to the variational distance,
we define 
\begin{align}
L^\epsilon_{P}(r):= \max_{P' \in U_{\epsilon,P}} L[G_{P,P'}](r).
\end{align}
Then, we have the following lemma, which is shown in Appendix \ref{AP66}.
\begin{lem}\Label{L3d}
\begin{align}
\lim_{\delta \to +0}L[G_{P,P}](1-\delta)=-H(P).
\Label{CY2}
\end{align}
\begin{align}
 \lim_{\epsilon \to +0} L^{\epsilon}_{P}(r)=L[G_{P,P}](r).
\Label{CY1}
\end{align}
\hfill $\square$
\end{lem}

For $\alpha>1$, we choose $R_1$, $R_2$, and $R_3$ to satisfy the conditions \eqref{NHO2A}, \eqref{NHO2B}, and \eqref{NHO2}.
For our decoder construction, we choose three real numbers $\epsilon_1,\epsilon_2>0$ and $R_4$.
The real number $R_4$ is chosen as
\begin{align}
I(X;Y)_P >R_4> R_1-R_2.\Label{C4}
\end{align}
Using Lemma \ref{L3d}, we choose $\epsilon_2$ such that
\begin{align}
-L[G_{P,P}](1-\epsilon_2)>R_1.
\end{align}
Then, we choose $\epsilon_1$ to satisfy
\begin{align}
\zeta_1\frac{\epsilon_2}{2}- \epsilon_1>0.
\end{align}
Next, we fix 
the size of message $\sM_n:=2^{n R_1}$,
the list size $\sL_n:=2^{n R_2}$,
and a number $\sM_n':=2^{n R_4}$, which is smaller than 
the message size $\sM_n$.
For $x^n =(x^n_1, \ldots, x_n^n)\in {\cal X}^n$, we define 
$w_{x^n}(y^n):= w_{x^n_1}(y_1^n)\cdots w_{x^n_n}(y_n^n)$
for $y^n=(y^n_1,\ldots, y_n^n)$.
%We define $w_{x^n}^c(y^n):= w_{x_1^n}^c(y_1^n)\cdots w_{x^n_n}^c(y_n^n)$,
%and $J[x^n]:= \sum_{i=1}^n D(W_{x_i^n}^c\|W_{x_i^n})$.
We prepare the decoder used in this proof as follows.
\begin{define}[Decoder $D_{\phi_n}$]\Label{Def1}
Given a distribution $P$ on ${\cal X}$, 
we define the decoder $D_{\phi_n}$ for a given encoder $\phi_n$ (a map from
$\{1, \ldots, \sM_n\}$ to ${\cal X}^n$) in the following way.
Using the condition \eqref{CS4},
we define the subset 
${\cal D}_{x^n}:= \{ y^n| w_{x^n} (y^n) \ge \sM_n' w_{P^n}^{n}(y^n) ,
\xi_{x^n} (y^n)\ge -n \epsilon_1
%w_{x^n} (y^n) \ge 2^{J[x^n] -n \epsilon_1} w_{x^n}^c (y^n)
\}$.
Then, for $y^n \in {\cal Y}^n$, 
we choose up to $\sL_n$ elements $i_1, \ldots, i_{\sL_n'}$ $(\sL_n'\le \sL_n)$
as the decoded messages
such that $ y^n \in {\cal D}_{\phi_n(i_j)}$ for $j=1, \ldots, \sL_n'$.
\hfill $\square$
\end{define}

Remember that, 
for $x^n=(x^n_1, \ldots, x^n_n),{x^n}'=({x^n_1}', \ldots, {x^n_n}')\in {\cal X}^n$, 
Hamming distance $d_H(x^n,{x^n}')$ is defined to be the number of $k$ such that $x_k^n \neq {x_k^n}'$
in Subsection \ref{OUT-K}.
In the proof of Theorem \ref{TH3}, 
we need to extract an encoder $\phi_n$ and elements $m \in {\cal M}_n$ that satisfies the following condition; \begin{align}
d_H(\phi_n(m),\phi_n(j)) >  n \epsilon_2 
\hbox{ for }\forall j\neq m. \Label{CC2} 
\end{align}
For this aim, given a code $\phi_n$ and a real number $\epsilon_2 >0$, we define the function
$\eta_{\phi_n,\epsilon_2}^C$
from ${\cal M}_n $ to $\{0,1\}$ as
\begin{align}
\eta_{\phi_n,\epsilon_2}^C(m) &:=
\left\{
\begin{array}{ll}
0 & \hbox{ when \eqref{CC2} holds} \\
1 & \hbox{ otherwise. }
\end{array}
\right. \Label{De2}
\end{align}

As shown in Section \ref{S7-C}, we have the following lemma.
\begin{lem}\Label{LL12}
When a code $\tilde{\phi}_n$ defined in a subset $\tilde{{\cal M}}_n \subset {\cal M}_n$
satisfies
\begin{align}
d_H(\tilde{\phi}_n(m),\tilde{\phi}_n(m'))> n \epsilon_2 \Label{E41}
\end{align}
for two distinct elements $ m \neq m'\in \tilde{{\cal M}}_n$,
the decoder $D_{\tilde{\phi}_n}$ defined in Definition \ref{Def1} satisfies 
\begin{align}
&\delta_{D}(D_{\tilde{\phi}_n}) 
\le
\frac{ \zeta_2}{{n} [\zeta_1\frac{\epsilon_2}{2}- \epsilon_1 ]_+^2  } .
\end{align}
\hfill $\square$
\end{lem}

\subsection{Proof of Theorem \ref{TH3}}
\noindent {\bf Step 1}: Lemmas related to random coding.

\noindent To show Theorem \ref{TH3},
%we set $\sM_n=2^{n R_1}$.
we assume that the variable $\Phi_n(m)$ for $m \in {\cal M}_n$
is subject to the distribution $P^n$ independently.
Then, we have the following four lemmas, which are shown later.
In this proof, we treat the code $\Phi_n$ as a random variable.
Hence, the expectation and the probability for this variable
are denoted by $\rE_{\Phi_n} $ and ${\rm Pr}_{\Phi_n}$, respectively.

\begin{lem}\Label{LL10}
When 
\begin{align}
I(X;Y)_P >R_4> R_1-R_2,\Label{C4LL}
\end{align}
we have the average version of Verifiable condition (A), i.e., 
\begin{align}
%\lim_{n \to \infty}\rE_{\Phi_n} \epsilon_A'(\Phi_n,D_{\Phi_n}) =
\lim_{n \to \infty}
\rE_{\Phi_n} 
\sum_{m=1}^{\sM_n}
\frac{1}{\sM_n} 
\epsilon_{A,m}(\Phi_n,D_{\Phi_n}) 
=0
 \Label{ER1}.
\end{align}
\hfill $\square$
\end{lem}

\begin{lem}\Label{LL11}
For $\epsilon_2>0$,
we have
\begin{align}
\lim_{n \to \infty}
\rE_{\Phi_n} 
\sum_{m=1}^{\sM_n}
\frac{1}{\sM_n} 
\eta_{\Phi_n,\epsilon_2}^C(m) =0
\Label{ER3}.
\end{align}
\hfill $\square$
\end{lem}

\begin{lem}\Label{LL19}
We choose $Q_{P,\alpha} \in {\cal P}({\cal Y})$ as
\begin{align}
Q_{P,\alpha}:=  \argmin_{Q \in {\cal P}({\cal Y})} D_\alpha( \bW\times P\|  Q \times P).
\Label{LL19E}
\end{align}
We have
\begin{align}
\rE_{\Phi_n} 
\sum_{i=1}^{\sM_n} \frac{1}{\sM_n} 2^{(\alpha-1)D_\alpha(W_{\Phi_n(i)}\|Q_{P,\alpha}^n)}
=2^{n (\alpha-1) I_\alpha(X;Y)_P}
\Label{HR7}.
\end{align}
\hfill $\square$
\end{lem}

\noindent {\bf Step 2}: 
Extraction of an encoder $\phi_n$ and messages $m$ 
with a small decoding error probability 
that satisfies the condition \eqref{CC2}.

\noindent 
We define $\epsilon_{3,n}$ as
\begin{align}
 \epsilon_{3,n}:= 
9 \rE_{\Phi_n}\sum_{m=1}^{\sM_n}
\frac{1}{\sM_n} 
\Big(\epsilon_{A,m}(\phi_n,D_{\Phi_n}) 
+\eta_{\Phi_n,\epsilon_2}^C(m) 
\Big).
\end{align}
Lemmas \ref{LL10} and \ref{LL11} guarantees that $\epsilon_{3,n} \to 0$.
Then, there exists a sequence of codes $\phi_n$ such that
\begin{align}
&\sum_{m=1}^{\sM_n}
\frac{1}{\sM_n} 
\Big(\epsilon_{A,m}(\phi_n,D_{\phi_n}) 
+\eta_{\phi_n,\epsilon_2}^C(m) 
\Big)
 \le \frac{\epsilon_{3,n}}{3} \Label{BGC}, \\
&\sum_{m=1}^{\sM_n} \frac{1}{\sM_n} 2^{(\alpha-1)D_\alpha(W_{\phi_n(m)}\|Q_{P,\alpha}^n)}
 \le 3\cdot 2^{n (\alpha-1) I_\alpha(X;Y|P)}
\Label{HR7F}.
\end{align}
Due to Eq. \eqref{BGC}, Markov inequality guarantees that
there exist $2\sM_n/3$ elements 
$\tilde{{\cal M}}_n:=
\{m_1, \ldots, m_{2 \sM_n/3}\}$
such that every element $m \in \tilde{{\cal M}}_n$ satisfies
\begin{align}
\epsilon_{A,m}(\phi_n,D_{\phi_n}) 
+\eta_{\phi_n,\epsilon_2}^C(m) 
\le \epsilon_{3,n},
\end{align}
which implies that
\begin{align}
\epsilon_{A,m}(\phi_n,D_{\phi_n}) &\le \epsilon_{3,n}  \Label{NAC1}\\
\eta_{\phi_n,\epsilon_2}^C(m)& =0 \Label{NAC}
\end{align}
because $\eta_{\phi_n,\epsilon_2}^C$ takes value 0 or 1.
%of the definitions of \eqref{De1} and \eqref{De2}.
Then, we define a code $\tilde{\phi}_n$
on $\tilde{{\cal M}}_n$
as $\tilde{\phi}_n(m):= {\phi}_n(m)$ for $m \in \tilde{{\cal M}}_n $.
Eq. \eqref{NAC1} guarantees Condition (A).
Eq. \eqref{HR7F} guarantees that
\begin{align}
&\sum_{m\in \tilde{\cal M}_n} \frac{1}{|\tilde{\cal M}_n|} 
2^{(\alpha-1)D_\alpha(W_{\tilde{\phi}_n(m)}\|Q_{P,\alpha}^n)}
\nonumber\\
 =&
\sum_{m\in \tilde{\cal M}_n} \frac{3}{2\sM_n} 2^{(\alpha-1)D_\alpha(W_{\phi_n(m)}\|Q_{P,\alpha}^n)}\nonumber\\
\le& \frac{9}{2}\cdot 2^{n (\alpha-1) I_\alpha(X;Y|P)}\Label{ZSO}.
\end{align}

\noindent {\bf Step 3}: Proof of the binding condition for dishonest Alice (D).

\noindent 
The relation \eqref{NAC} guarantees the condition
\begin{align}
d_H(\tilde{\phi}_n(m),\tilde{\phi}_n(m')) > n \epsilon_2
\Label{MUFV}
\end{align}
for $ m \neq m'\in \tilde{{\cal M}}_n$.
Therefore, Lemma \ref{LL12} guarantees 
the binding condition for dishonest Alice (D), i.e., 
\begin{align}
& \delta_{D}(D_{\tilde{\phi}_n}) 
\le
\frac{ \zeta_2}{{n} [\zeta_1\frac{\epsilon_2}{2}- \epsilon_1 ]_+^2  } .
\Label{EF1}
\end{align}

\noindent {\bf Step 4}: Proof of the equivocation version of concealing condition (B).

\noindent 

%Let $P_{\tilde{{\cal M}}_n}$ be the uniform distribution on $\tilde{{\cal M}}_n$.
Eq. \eqref{ZSO} guarantees that
\begin{align}
&  \min_{Q_n \in {\cal P}({\cal Y}^n)} \sum_{m\in \tilde{\cal M}_n} 
\frac{1}{|\tilde{{\cal M}}_n|}
2^{(\alpha-1)D_\alpha(W_{\tilde{\phi}_n(m)}\|Q_n)}\nonumber \\
\le & \sum_{m\in \tilde{\cal M}_n} 
\frac{1}{|\tilde{{\cal M}}_n|}
2^{(\alpha-1)D_\alpha(W_{\tilde{\phi}_n(m)}\|Q_{P,\alpha}^n)}\nonumber \\
\le & \frac{9}{2}\cdot 2^{n (\alpha-1) I_\alpha(X;Y)_P}. %\cdot \frac{1}{|\tilde{{\cal M}}_n|^{\alpha-1}}.
\end{align}
Hence,
\begin{align}
 \lim_{n \to \infty} 
\frac{1}{n}E_\alpha( \tilde{\phi}_n) 
\ge  R_1- I_\alpha(X;Y)_P \ge R_3.
\end{align}

\endproof

\subsection{Proof of Lemma \ref{LL12}}\Label{S7-C}
\noindent {\bf Step 1}: 
Evaluation of $W^n_{x^n} ({\cal D}_{{x^n}'})$.

\noindent 
%As the preparation, we evaluate $W^n_{x^n} ({\cal D}_{{x^n}'})$.
The conditions \eqref{CS1} and \eqref{CS2} imply that 
\begin{align}
\mathbb{E}_{{x^n}'} 
[\xi_{x^n}] 
& \le -\zeta_1 d(x^n,{x^n}').
\end{align}
The condition \eqref{CS3} implies that 
\begin{align}
\mathbb{V}_{{x^n}'} [\xi_{x^n}] 
& \le n \zeta_2.
\end{align}
%Due to the definition of $V(\bW)$ (Eq. \eqref{OfIm}),
%$n V(\bW)$ is an upper bound of the variance of the variable
%$\log W_{x^n} (y^n) -\log W_{P^n}^{n}(y^n)  $ under the distribution $W^n_{x^n}$.
%For the definition of $V(\bW)$, see \eqref{OfIm}.
Hence, applying Chebyshev inequality to 
the variable $\xi_{x^n} (Y^n) $, we have
\begin{align}
W^n_{{x^n}'} ({\cal D}_{{x^n}})
\le &
W^n_{{x^n}'} (\{ y^n| 
\xi_{x^n} (y^n) \ge - n \epsilon_1 \})\nonumber \\
\le &
\frac{n \zeta_2}{[
\zeta_1 d(x^n,{x^n}')-n \epsilon_1 ]_+^2  }\Label{Che}.
\end{align}

\noindent {\bf Step 2}: Evaluation of smaller value of 
$W^n_{x^n} ({\cal D}_{\tilde{\phi}_n(m)}) $ and $W^n_{x^n} ({\cal D}_{\tilde{\phi}_n(m')}) $.
Since Eq. \eqref{E41} implies
\begin{align}
&n \epsilon_2< d(\tilde{\phi}_n(m),\tilde{\phi}_n(m')) \nonumber\\
\le & 
d_H(x^n,\tilde{\phi}_n(m)) + d_H(x^n,\tilde{\phi}_n(m')) ,
\end{align}
we have
\begin{align}
& \max ([\zeta_1 d_H(x^n,\tilde{\phi}_n(m)) - n \epsilon_1 ]_+,\nonumber\\
&[\zeta_1 d_H(x^n,\tilde{\phi}_n(m')) - n \epsilon_1 ]_+)
\nonumber\\
\ge &[n(\zeta_1\frac{\epsilon_2}{2}- \epsilon_1) ]_+^2 .
\end{align}
Hence, \eqref{Che} guarantees that
\begin{align}
&\min(W^n_{x^n} ({\cal D}_{\tilde{\phi}_n(m)}) ,W^n_{x^n} ({\cal D}_{\tilde{\phi}_n(m')}) ) \nonumber \\
\le &
\frac{n \zeta_2}{
\max ([\zeta_1 d(x^n\!,\tilde{\phi}_n(m)) \!-\! n \epsilon_1 ]_+^2,
[\zeta_1 d(x^n\!,\tilde{\phi}_n(m')) \!-\! n \epsilon_1 ]_+^2)} \nonumber \\
\le &
\frac{ n \zeta_2}{ [n(\zeta_1\frac{\epsilon_2}{2}- \epsilon_1) ]_+^2  } 
=\frac{ \zeta_2}{{n} [\zeta_1\frac{\epsilon_2}{2}- \epsilon_1 ]_+^2  } ,
\end{align}
which implies the desired statement.
\endproof

\subsection{Proof of Lemma \ref{LL10}}\Label{S7-D}
We show Lemma \ref{LL10} by employing an idea similar to \cite{Verdu,HN}. 
First, we show the following lemma.
\begin{lem}\Label{NMU}
We have the following inequality;
\begin{align}
& \epsilon_A(\Phi_n,D_{\Phi_n})\nonumber \\
\le & 
\sum_{i=1}^{\sM_n} \frac{1}{\sM_n}
\Big(W_{\Phi_n(i)}({\cal D}_{\Phi_n(i)}^c)
+
\sum_{j\neq i } \frac{1}{\sL_n}
W_{\Phi_n(i)}({\cal D}_{\Phi_n(j)}) \Big)\Label{NML}.
\end{align}
\hfill $\square$
\end{lem}
\begin{proof}
When $i$ is sent, 
there are two cases for incorrect decoding.
The first case is the case that the received element $y$ does not belong to ${\cal D}_{\Phi_n(i)}$.
The second case is the case that there are more than $\sL_n$ elements $i'$ 
to satisfy $y \in {\cal D}_{\Phi_n(i')}$.
In fact, the second case does not always realize incorrect decoding.
However, the sum of the probabilities of the first and second cases upper bounds
the decoding error probability $\epsilon_A(\Phi_n,D_{\Phi_n})$.
Hence, it is sufficient to evaluate these two probabilities.
The error probability of the first case is given in the first term of Eq. \eqref{NML}.
The error probability of the second case is given in the second term of Eq. \eqref{NML}.
\end{proof}

Taking the average in \eqref{NML} of Lemma \ref{NMU}
with respect to the variable $\Phi_n$, we obtain the following lemma.
The following discussion employs the notations $\rE_{\Phi_n}$ and $\rE_{X^n}$, which are defined in the middle of Section \ref{S51II}.
\begin{lem}\Label{FHU}
We have the following inequality;
\begin{align}
&\rE_{\Phi_n}\epsilon_A(\Phi_n,D_{\Phi_n})\nonumber \\
\le &
\sum_{x^n\in {\cal X}^n}P^n(x^n) 
\Big(W_{x^n}({\cal D}_{x_n}^c)
+
\frac{\sM_n-1}{\sL_n}
W_{P^n}({\cal D}_{x^n})\Big).
\Label{NML2}
\end{align}
\hfill $\square$
\end{lem}
Applying Lemma \ref{FHU}, % to the $n$-fold i.i.d. case,
we have
\begin{align}
& \rE_{\Phi_n} \epsilon_A(\Phi_n,D_{\Phi_n})\nonumber  \\
\le &
\rE_{X^n}W_{X^n} \big(\big\{y^n\big| 
2^{-n R_4} w_{X^n}(y^n)   < w^n_{P^n}(y^n) \big\} \big)\nonumber \\
&+\rE_{X^n}W_{X^n} \big(\big\{y^n\big| 
\xi_{x^n} (y^n) < - n \epsilon_1 
 \big\} \big)\nonumber \\
&+
\rE_{X^n}2^{n (R_1-R_2)}W_{P^n}
\big(\big\{y^n\big|\nonumber \\
&\qquad 2^{-n R_4} w_{X^n}(y^n)  \ge w_{P^n}(y^n)
\big\}\big)
\nonumber \\
\stackrel{(a)}{\le}&
\rE_{X^n}W_{X^n} \big(\big\{y^n\big| 
\log w_{X^n}(y^n)  -\log w^n_{P^n}(y^n) < n R_4\big\} \big)\nonumber \\
&+
\rE_{X^n}W_{X^n} \big(\big\{y^n\big| 
\xi_{x^n} (y^n) < - n \epsilon_1 \big\} \big)\nonumber \\
&+
\rE_{X^n}2^{n (R_1-R_2)}2^{-n R_4}w_{X^n}
\big(\big\{y^n\big| \nonumber \\
&\qquad
2^{-n R_4} w_{X^n}(y^n)  \ge w_{P^n}(y^n)\big\}\big)
\nonumber \\
\le &
\rE_{X^n}W_{X^n} \bigg(\bigg\{y^n\bigg| 
\frac{1}{n} (\log w_{X^n}(y^n)  \nonumber \\
&\qquad-\log w_{P^n}(y^n)) <  R_4\bigg\} \bigg)\nonumber \\
&+\rE_{X^n}W_{X^n} \bigg(\bigg\{y^n\bigg| 
\frac{1}{n} \xi_{x^n} (y^n) < - \epsilon_1 \bigg\} \bigg)
\nonumber \\
&+2^{n (R_1-R_2-R_4)}
 \Label{ER1B},
\end{align}
where $(a)$ follows from 
the relation 
\begin{align*}
&W_{P^n}
\big(\big\{y^n\big|
2^{-n R_4} w_{X^n}(y^n)  \ge w_{P^n}(y^n)
\big\}\big)\nonumber \\
\le &
2^{-n R_4}W_{X^n}
\big(\big\{y^n\big| 
2^{-n R_4} w_{X^n}(y^n)  \ge w_{P^n}(y^n)\big\}\big).
\end{align*}
The variable $\frac{1}{n} (\log w_{X^n}(y^n)  -\log w_{P^n}(y^n))$ 
is the mean of $n$ independent variables 
that are identical to the variable $\log w_X(Y)-\log w_P(Y)$ whose
average is $I(X;Y)_P> R_4$.
The variable $\frac{1}{n} \xi_{x^n} (y^n)$
is the mean of $n$ independent variables 
that are identical to the variable $\xi_{X} (Y)$ 
whose average is $0$.
Thus, the law of large number guarantees that the first and the second terms in 
\eqref{ER1B} approaches to zero as $n$ goes to infinity.
The third term in 
\eqref{ER1B} also approaches to zero due to 
the assumption \eqref{C4}.
Therefore, we obtain Eq. \eqref{ER1}. 
\endproof

\subsection{Proof of Lemma \ref{LL19}}\Label{S7-E9}
Eq. \eqref{HR7} can be shown as follows.
\begin{align}
& \rE_{\Phi} 
\sum_{i=1}^{\sM_n} \frac{1}{\sM_n} 2^{(\alpha-1)D_\alpha(W_{\Phi_n(i)}\|Q_{P,\alpha}^n)}\nonumber \\
=&
\rE_{\Phi} 
\sum_{i=1}^{\sM_n} \frac{1}{\sM_n} 
\prod_{j=1}^n
\mathbb{E}_{\Phi_n(i)_j}\Big[
\Big(\frac{w_{\Phi_n(i)_j}(Y)}{q_{P,\alpha}(Y)}\Big)^{\alpha-1}\Big] \nonumber \\
=&
\sum_{i=1}^{\sM_n} \frac{1}{\sM_n} 
\prod_{j=1}^n
\rE_{\Phi} 
\mathbb{E}_{\Phi_n(i)_j}\Big[
\Big(\frac{w_{\Phi_n(i)_j}(Y)}{q_{P,\alpha}(Y)}\Big)^{\alpha-1}\Big] \nonumber \\
=&
\sum_{i=1}^{\sM_n} \frac{1}{\sM_n} 
\prod_{j=1}^n
\sum_{x\in {\cal X}}P(x)
\mathbb{E}_{x}\Big[
\Big(\frac{w_{x}(Y)}{q_{P,\alpha}(Y)}\Big)^{\alpha-1}\Big] \nonumber \\
=&
\sum_{i=1}^{\sM_n} \frac{1}{\sM_n} 
\prod_{j=1}^n
2^{(\alpha-1) D_\alpha(\bW\times P\|  Q_{P,\alpha} \times P)}
\nonumber \\
\stackrel{(a)}{=}&
\sum_{i=1}^{\sM_n} \frac{1}{\sM_n} 
\prod_{j=1}^n
2^{(\alpha-1) I_\alpha(X;Y)_P}
\nonumber \\=&
\sum_{i=1}^{\sM_n} \frac{1}{\sM_n} 
2^{n (\alpha-1) I_\alpha(X;Y)_P}
=
2^{n (\alpha-1) I_\alpha(X;Y)_P}
\Label{HR7LL},
\end{align}
where $(a)$ follows from \eqref{LL192E} and \eqref{LL19E}.
\endproof

\subsection{Proof of Lemma \ref{LL11}}\Label{S7-E}
The outline of the proof of Lemma \ref{LL11} is the following.
To evaluate the value
$\rE_{\Phi_n} 
\sum_{m=1}^{\sM_n}
\frac{1}{\sM_n} 
\eta_{\Phi_n,\epsilon_2}^C(m)$, we convert it to the sum of certain probabilities.
We evaluate these probabilities by excluding a certain exceptional case.
That is, we show that the probability of the exceptional case is small and 
these probabilities under the condition to exclude the exceptional case is also small.
The latter will be shown by evaluating a certain conditional probability.
For this aim, 
we choose $\epsilon_4,\epsilon_5>0$ such that 
$\epsilon_4:=- L[G_{P,P}](1-\epsilon_2)-R_1 $ and
$-L_P^{\epsilon_5}(1-\epsilon_2)> R_1+\frac{\epsilon_4}{2}$.

\noindent {\bf Step 1}: Evaluation of a certain conditional probability.

\noindent 
We denote the empirical distribution of $x^n$ by $P[x^n]$.
That is, $nP[x^n](x)$ is the number of index $i=1, \ldots, n$ to satisfy $x^n_i=x$.
Hence, when $X^n=(X_1^n, \ldots, X_n^n)$ are independently subject to $P$, 
\begin{align}
&\rE_{X^n}
[2^{t (n-d(X^n,x^n)) }]
=2^{G_{P|x_1^n}(t) +\cdots+G_{P|x_n^n}(t) } \nonumber \\
=&2^{n G_{P,P[x^n]}(t)} .\Label{ZLO}
\end{align}

We define two conditions
$A_{n,i}$ and $B_{n,i}$ for the encoder $\Phi_n$ as
\begin{description}
\item[$A_{n,i}$]
$P[\Phi_n(i)] \in U_{\epsilon_5,P}$.
\item[$B_{n,i}$]
$\exists j\neq i, d(\Phi_n(i),\Phi_n(j)|P) \le n \epsilon_2$.
\end{description}
The aim of this step is the evaluation of the 
conditional probability $\Pr_{\Phi_n} (B_{n,i}|A_{n,i})$ that expresses
the probability that 
the condition $B_{n,i}$ holds under the condition $A_{n,i}$.

We choose $j \neq i$. Markov inequality implies that
\begin{align}
&\Pr_{\Phi_n(j)|\Phi_n(i)} \Big( d(\Phi_n(i),\Phi_n(j)) \le n \epsilon_2 \Big) 
\nonumber \\
=&\Pr_{\Phi_n(j)|\Phi_n(i)} \Big( n- d(\Phi_n(i),\Phi_n(j)) \ge n (1-\epsilon_2) \Big) \nonumber \\
\le &
\rE_{\Phi_n(j)|\Phi_n(i)}
[2^{t (n-d(\Phi_n(i),\Phi_n(j))) }]
2^{- tn (1-\epsilon_2)}\nonumber \\
=&2^{n G_{P,P[\Phi_n(i)]}(t) - t n (1-\epsilon_2)},
\end{align}
where $\Pr_{\Phi_n(j)|\Phi_n(i)}$ and $\rE_{\Phi_n(j)|\Phi_n(i)}$ 
are the conditional probability and the conditional expectation
for the random variable $\Phi_n(j)$ with
the fixed variable $\Phi_n(i)$.
The final equation follows from \eqref{ZLO}.
When the fixed variable $\Phi_n(i)$ satisfies the condition $A_{n,i}$,
taking the infimum with respect to $s$, we have
\begin{align}
& \Pr_{\Phi_n(j)|\Phi_n(i)} \Big( d(\Phi_n(i),\Phi_n(j)) \le n \epsilon_2 \Big) 
\nonumber \\ 
\le& 2^{n L[G_{P,P[\Phi_n(i)]}](1-\epsilon_2)}
\le 2^{n L_P^{\epsilon_5}(1-\epsilon_2)}.
\end{align}
Hence, 
\begin{align}
&\Pr_{\Phi_{n,i,c} |\Phi_n(i)} (B_{n,i})
\nonumber \\
\le &
\sum_{j (\neq i)\in {\cal M}_n}
\Pr_{\Phi_n(j)|\Phi_n(i)} \Big( d(\Phi_n(i),\Phi_n(j)) \le n \epsilon_2 \Big) 
 \nonumber \\
%\big(\exists j\neq i, F^n(\Phi_n(i),\Phi_n(j)|P) \ge n (R_3-\epsilon_2)
%\big| G^n(s, \Phi_n(i)|P) < n (- R_1+s(R_1-R_2)-(1+s)\epsilon) 
%\big) \\
\le &
\sum_{j (\neq i)\in {\cal M}_n}
%(2^{n R_1}-1) 
2^{n L_P^{\epsilon_5}(1-\epsilon_2)}
\le
2^{n (L_P^{\epsilon_5}(1-\epsilon_2)+R_1)} 
\le 2^{-n \epsilon_4/2},
\end{align}
where $ \Phi_{n,i,c}$ expresses the random variables $\{\Phi_n(j)\}_{j \neq i}$.
Then, we have
\begin{align}
\Pr_{\Phi_n} (B_{n,i}|A_{n,i}) \le 2^{-n\epsilon_4/2}.
\Label{N2}
\end{align}

\noindent {\bf Step 2}: Evaluation of 
$\rE_{\Phi_n} 
\sum_{m=1}^{\sM_n}
\frac{1}{\sM_n} 
\eta_{\Phi_n,\epsilon_2}^C(m)$.

\noindent 
The quantity $\rE_{\Phi_n} 
\sum_{m=1}^{\sM_n}
\frac{1}{\sM_n} 
\eta_{\Phi_n,\epsilon_2}^C(m)$ can be evaluated as
\begin{align}
&\rE_{\Phi_n} 
\sum_{m=1}^{\sM_n}
\frac{1}{\sM_n} 
\eta_{\Phi_n,\epsilon_2}^C(m) \nonumber \\
=& 
\frac{1}{\sM_n} 
\rE_{\Phi_n} 
|\{ i | B_{n,i} \hbox{ holds. }\}|
=
\sum_{i=1}^{\sM_n}
\frac{1}{\sM_n} \Pr_{\Phi_n}
(B_{n,i}
%\exists j\neq i, F^n(\Phi_n(i),\Phi_n(j)|P) \ge n (R_3-\epsilon_2)
%\big| G^n(s, \Phi_n(i)|P) < n (- R_1+s(R_1-R_2)-(1+s)\epsilon) 
)\nonumber  \\
\le &
\sum_{i=1}^{\sM_n}
\frac{1}{\sM_n} 
\Big(\Pr_{\Phi_n} (A_{n,i}) \Pr_{\Phi_n} (B_{n,i}|A_{n,i}) 
\nonumber \\
&\qquad+ (1-\Pr_{\Phi_n} (A_{n,i}))
\Big)\nonumber  \\
\stackrel{(a)}{\le}
 &
2^{-n\epsilon_4/2}
+
\sum_{i=1}^{\sM_n}
\frac{1}{\sM_n} 
(1-\Pr (A_{n,i})),\Label{N5}
\end{align}
where $(a)$ follows from Eq. \eqref{N2}.

Since $P[\Phi_n(i)]$ converges to $P $ in probability,
we have
\begin{align}
\Pr_{\Phi_n} (A_{n,i}) \to 1.
\Label{N3}
\end{align}
Hence, the combination of Eqs. \eqref{N5} and \eqref{N3}
implies the desired statement.
\endproof

\section{Proof of Theorem \ref{TH4}}
\subsection{Main part of proof of Theorem \ref{TH4}}
Hamming distance $d_H$ plays a central role in our proof of Theorem \ref{TH3}.
However, since elements of $\tilde{\cal X} \setminus {\cal X}$ 
can be sent by dishonest Alice,
Hamming distance $d_H$ does not work in our proof of Theorem \ref{TH4}.
Hence, we introduce an alternative distance on $\tilde{\cal X}^n$.
We modify the distance $d$ on $\tilde{\cal X}$ as
\begin{align}
\bar{d}(x,x'):= \frac{1}{\zeta_3}\min( d(x,x'), \zeta_3),
\end{align}
where
\begin{align}
\zeta_3:= \min_{x\neq x' \in {\cal X} } d(x,x').
\end{align}
Then, we define 
\begin{align}
\bar{d}_H(x^n,{x^n}'):= \sum_{i=1}^n\bar{d}(x_i^n,{x_i^n}'),
\end{align}
which is the same as Hamming distance $d_H$ on ${\cal X}^n$.
Instead of Lemma \ref{LL12}, we have the following lemma.
\begin{lem}\Label{LL12B}
When a code $\tilde{\phi}_n$ defined in a subset $\tilde{{\cal M}}_n \subset {\cal M}_n$
satisfies
\begin{align}
{d}_H(\tilde{\phi}_n(m),\tilde{\phi}_n(m'))> n \epsilon_2 \Label{E41B}
\end{align}
for two distinct elements $ m \neq m'\in \tilde{{\cal M}}_n$,
the decoder $D_{\tilde{\phi}_n}$ defined in Definition \ref{Def1} satisfies 
\begin{align}
&\delta_{D'}(D_{\tilde{\phi}_n}) 
\le
2^{t n(2 \epsilon_1 -  \frac{ \epsilon_2}{4}\bar{\zeta}_{1,t}(\zeta_3 \frac{\epsilon_2}{4}))
}
+\frac{n \bar{\zeta}_2}{[n \epsilon_1 ]_+^2  } .\Label{E41C}
\end{align}
\hfill $\square$
\end{lem}

In our proof of Theorem \ref{TH4}, we choose 
the real numbers $R_4, \epsilon_2, \epsilon_1$.
We fix $s \in (0,1/2)$.
While we choose $R_4, \epsilon_2>0$ in the same way as our proof of Theorem \ref{TH3}, 
we choose $\epsilon_1>0$ to satisfy 
\begin{align}
\frac{ \epsilon_2}{4}\bar{\zeta}_{1,t}(\zeta_3 \frac{\epsilon_2}{4})> 2 \epsilon_1 .
\end{align}

In this choice, the RHS of \eqref{E41C} goes to zero.
 Since the conditions \eqref{E41B} and \eqref{E41C} take the same role as the conditions of Lemma \ref{LL12},
the proof of  Theorem \ref{TH3} works by replacing Lemma \ref{LL12} by Lemma \ref{LL12B}
as a proof of  Theorem \ref{TH4}.

\subsection{Proof of Lemma \ref{LL12B}}
\noindent {\bf Step 1}: 
Evaluation of $W^n_{x^n} ({\cal D}_{{x^n}'})$.

\noindent 
As shown in {\bf Step 3}, 
when $\bar{d}_H(x^n,{x^n}') = k$, for $t \in (0,\frac{1}{2})$,
we have
\begin{align}
\frac{-1}{t}\log 
\mathbb{E}_{x'}[2^{t(\xi_{x}(Y)-\xi_{x'}(Y))}]
\ge \frac{k}{2}\bar{\zeta}_{1,t}(\zeta_3 \frac{k}{2n})
\Label{XOL}.
\end{align}
Applying Markov inequality to the variable $2^{t(\xi_{x}(Y)-\xi_{x'}(Y))} $, we have
\begin{align}
& W^n_{{x^n}'} (\{ y^n| 
\xi_{x^n} (y^n)-\xi_{{x^n}'} (y^n) \ge-  2 n \epsilon_1\})\nonumber \\
=&
 W^n_{{x^n}'} (\{ y^n| 
2^{t(\xi_{x^n} (y^n)-\xi_{{x^n}'} (y^n))} \ge 2^{-2t n \epsilon_1}\})\nonumber \\
\le & \mathbb{E}_{x'}[2^{t(\xi_{x}(Y)-\xi_{x'}(Y))}] 2^{2t n \epsilon_1}
\le 2^{t (2 n \epsilon_1 -  \frac{k}{2}\bar{\zeta}_{1,t}(\zeta_3 \frac{k}{2n}))}.
\Label{XOL2}
\end{align}

%As the preparation, we evaluate $W^n_{x^n} ({\cal D}_{{x^n}'})$.
The condition \eqref{CS1} implies that 
\begin{align}
\mathbb{E}_{{x^n}'} 
[\xi_{{x^n}'}]=0 .
\end{align}
The condition \eqref{CS3C} implies that 
\begin{align}
\mathbb{V}_{{x^n}'} [\xi_{{x^n}'}] 
& \le n \bar{\zeta}_2.
\end{align}
Hence, applying Chebyshev inequality to 
the variable $\xi_{{x^n}'} (Y^n)$, we have
\begin{align}
W^n_{{x^n}'} (\{ y^n| \xi_{{x^n}'} (y^n) \le -n \epsilon_1 \})
\le
\frac{n \bar{\zeta}_2}{[n \epsilon_1 ]_+^2  }\Label{CheC}.
\end{align}

%Due to the definition of $V(\bW)$ (Eq. \eqref{OfIm}),
%$n V(\bW)$ is an upper bound of the variance of the variable
%$\log W_{x^n} (y^n) -\log W_{P^n}^{n}(y^n)  $ under the distribution $W^n_{x^n}$.
%For the definition of $V(\bW)$, see \eqref{OfIm}.
Hence, we have
\begin{align}
&W^n_{{x^n}'} ({\cal D}_{{x^n}}) \nonumber \\
\le & W^n_{{x^n}'} (\{ y^n| 
- n \epsilon_1 \le \xi_{x^n} (y^n)  \})\nonumber \\
=& W^n_{{x^n}'} (\{ y^n| 
- n \epsilon_1 \le \xi_{x^n} (y^n)-\xi_{{x^n}'} (y^n)
+\xi_{{x^n}'} (y^n)  \})\nonumber \\
\le & W^n_{{x^n}'} (\{ y^n| - 2n \epsilon_1 \le
\xi_{x^n} (y^n)-\xi_{{x^n}'} (y^n) \}) \nonumber \\
&+ W^n_{{x^n}'} \left( \left\{  y^n \left| 
\begin{array}{ll}
- n \epsilon_1 \le \!\!\!\!&\xi_{x^n} (y^n)-\xi_{{x^n}'} (y^n) \\
&+\xi_{{x^n}'} (y^n) ,\\
- 2n \epsilon_1> \!\!\!\!& \xi_{x^n} (y^n)-\xi_{{x^n}'} (y^n) 
\end{array}
\right\} \right)\right.\nonumber \\
\stackrel{(a)}{\le} & 
W^n_{{x^n}'} (\{ y^n| - 2n \epsilon_1\le 
\xi_{x^n} (y^n)-\xi_{{x^n}'} (y^n)  \}) \nonumber \\
&+ W^n_{{x^n}'} (\{ y^n| 
\xi_{{x^n}'} (y^n) > n \epsilon_1 
%\xi_{x^n} (y^n)-\xi_{{x^n}'} (y^n) < - 2n \epsilon_1
\})\nonumber \\
%\le & W^n_{{x^n}'} (\{ y^n| 
%\xi_{x^n} (y^n)-\xi_{{x^n}'} (y^n) \le 2n \epsilon_1 \})
%+ W^n_{{x^n}'} (\{ y^n| \xi_{{x^n}'} (y^n) \le -n  \epsilon_1 \})\nonumber \\
\stackrel{(b)}{\le}&
2^{t (2 n \epsilon_1 -  \frac{k}{2}\bar{\zeta}_{1,t}(\zeta_3 \frac{k}{2n}))}+
\frac{n \bar{\zeta}_2}{[n \epsilon_1 ]_+^2  },\Label{CheM}
\end{align}
where
$(a)$ follows from the fact that
the conditions $- n \epsilon_1 \le
\xi_{x^n} (y^n)-\xi_{{x^n}'} (y^n)
+\xi_{{x^n}'} (y^n)$ and $
- 2n \epsilon_1> 
\xi_{x^n} (y^n)-\xi_{{x^n}'} (y^n) $
imply the condition $\xi_{{x^n}'} (y^n) > n \epsilon_1 $,
and
$(b)$ follows from \eqref{XOL2} and \eqref{CheC}.

\noindent {\bf Step 2}: Evaluation of smaller value of 
$W^n_{x^n} ({\cal D}_{\tilde{\phi}_n(m)}) $ and $W^n_{x^n} ({\cal D}_{\tilde{\phi}_n(m')}) $.
We simplify $d(x^n,\tilde{\phi}_n(m))$ and $d(x^n,\tilde{\phi}_n(m')) $ to $k_1$ and $k_2$.
Since Eq. \eqref{E41B} implies
\begin{align}
n \epsilon_2< d(\tilde{\phi}_n(m),\tilde{\phi}_n(m')) \le k_1+k_2,
%d(x^n,\tilde{\phi}_n(m)) + d(x^n,\tilde{\phi}_n(m')) ,
\end{align}
we have
\begin{align}
\frac{n \epsilon_2}{2} \le k_3:= \max(k_1,k_2).
\Label{NXT}
\end{align}
Since $\bar{\zeta}_{1,t}(r) $ is monotonically increasing for $r$,
\eqref{NXT} yields
\begin{align}
&\min \Big[t \Big(2 n \epsilon_1 -  \frac{k_1}{2}\bar{\zeta}_{1,t}\Big(\zeta_3 \frac{k_1}{2n}\Big)\Big),\nonumber \\
&\qquad t \Big(2 n \epsilon_1 -  \frac{k_2}{2}\bar{\zeta}_{1,t}\Big(\zeta_3 \frac{k_2}{2n}\Big)\Big)\Big] \nonumber \\
\le &
t \Big(2 n \epsilon_1 - 
\max \Big[
 \frac{k_1}{2}\bar{\zeta}_{1,t}\Big(\zeta_3 \frac{k_1}{2n}\Big),
 \frac{k_2}{2}\bar{\zeta}_{1,t}\Big(\zeta_3 \frac{k_2}{2n}  \Big)
 \Big]\Big)
\nonumber \\
= &
t \Big(2 n \epsilon_1 - 
\frac{k_3}{2}\bar{\zeta}_{1,t}\Big(\zeta_3 \frac{k_3}{2n}\Big)
\Big)
\nonumber \\
\le & t \Big(2 n \epsilon_1 -  \frac{n \epsilon_2}{4}\bar{\zeta}_{1,t}\Big(\zeta_3 \frac{n \epsilon_2}{4n}\Big)\Big) 
= t n \Big(2 \epsilon_1 -  \frac{\epsilon_2}{4}\bar{\zeta}_{1,t}\Big(\zeta_3 \frac{\epsilon_2}{4}\Big)\Big).
\Label{XLP1}
\end{align}
Thus, 
\begin{align}
&\min[W^n_{x^n} ({\cal D}_{\tilde{\phi}_n(m)}) ,W^n_{x^n} ({\cal D}_{\tilde{\phi}_n(m')}) ] \nonumber \\
\stackrel{(a)}{\le}&
\min \big[2^{t (2 n \epsilon_1 -  \frac{k_1}{2}\bar{\zeta}_{1,t}(\zeta_3 \frac{k_1}{2n}))}
2^{t (2 n \epsilon_1 -  \frac{k_2}{2}\bar{\zeta}_{1,t}(\zeta_3 \frac{k_2}{2n}))}\big]\nonumber \\
&+\frac{n \bar{\zeta}_2}{[n \epsilon_1 ]_+^2  } \nonumber \\
= &
2^{
\min \big[t (2 n \epsilon_1 -  \frac{k_1}{2}\bar{\zeta}_{1,t}(\zeta_3 \frac{k_1}{2n})),
t (2 n \epsilon_1 -  \frac{k_2}{2}\bar{\zeta}_{1,t}(\zeta_3 \frac{k_2}{2n}))\big]}
+\frac{n \bar{\zeta}_2}{[n \epsilon_1 ]_+^2  } \nonumber \\
\stackrel{(b)}{\le}&
2^{t n (2 \epsilon_1 -  \frac{ \epsilon_2}{4}\bar{\zeta}_{1,t}(\zeta_3 \frac{\epsilon_2}{4}))
}
+\frac{n \bar{\zeta}_2}{[n \epsilon_1 ]_+^2  },\Label{NLP}
\end{align}
where $(a)$ follows \eqref{CheM}, and $(b)$ follows from \eqref{XLP1}.
Eq. \eqref{NLP} implies \eqref{E41C}, i.e., the desired statement of Lemma \ref{LL12B}.

\noindent {\bf Step 3}: Proof of \eqref{XOL}.
To show \eqref{XOL}, we consider the random variable  $J$ subject to 
the uniform distribution $P_{\uni, n}$ on $\{1, \ldots, n\}$.
The quantity $1-\bar{d}(x_J^n,{x_J^n}')$ can be considered as a non-negative random variable
whose expectation is $1-\frac{k}{n}$.
We apply the Markov inequality to the variable $1-\bar{d}(x_A^n,{x_A^n}')$.
Then, we have 
\begin{align}
& \Big|\Big\{j \in \{1, \ldots, n\} \Big|
\bar{d}(x_j^n,{x_j^n}') < \frac{k}{2n} \Big\}\Big|\nonumber \\
=&\Big|\Big\{j \in \{1, \ldots, n\} \Big|
1-\bar{d}(x_j^n,{x_j^n}') 
> 1-\frac{k}{2n} \Big\}\Big| \nonumber\\ 
\le & n \cdot \frac{1-\frac{k}{n}}{1-\frac{k}{2n} }
\le n \cdot \Big(1-\frac{k}{2n}\Big)=n-\frac{k}{2},
\end{align}
where the final inequality follows from the relation between 
arithmetic and geometric means. 
Hence, we have
\begin{align}
\Big|\Big\{j \in \{1, \ldots, n\} \Big|
\bar{d}(x_j^n,{x_j^n}') \ge \frac{k}{2n} \Big\}\Big|
\ge \frac{k}{2}.\Label{VSP}
\end{align}
Since $\bar{d}(x_j^n,{x_j^n}') \ge\frac{k}{2n}$ implies $
d(x_j^n,{x_j^n}') \ge \zeta_3 \frac{k}{2n}$,
\eqref{VSP} implies \eqref{XOL}.

\endproof

\section{Conclusion}
We have  proposed a new concept, secure list decoding, 
%which has additional requirements for the conventional list decoding.
which imposes additional requirements on conventional list decoding
to work as a relaxation of bit-string commitment.
This scheme has three requirements.
Verifiable condition (A), Equivocation version of concealing condition (B), and 
Binding condition.
Verifiable condition (A) means that the message sent by Alice (sender) is contained in the list output by Bob (receiver).
Equivocation version of concealing condition (B)
is given as a relaxation of the concealing condition of bit-string commitment.
That is, it expresses Bob's uncertainty of Alice's message.
Binding condition has two versions.
One is the condition (C) for honest Alice.
The other is the condition (D) for dishonest Alice.
Since there is a possibility that dishonest Alice uses a different code, 
we need to guarantee the impossibility of cheating even for such a dishonest Alice.
In this paper, we have shown the existence of a code to satisfy these three conditions.
Also, we have defined the capacity region as the possible triplet of the rates of the message and the list,
and the equivocation rate,
and have derived the capacity region 
when the encoder is given as a deterministic map. 
Under this condition, we have shown that 
the conditions (C) and (D) have the same capacity region. 
However, we have not derived the capacity region when 
the stochastic encoder is allowed.
Therefore, the characterization of the capacity region of this case is an interesting open problem.

As the second contribution, we have formulated the secure list decoding with a general input system.
For this formulation, we have assumed that 
honest Alice accesses only a fixed subset of the general input system
and dishonest Alice can access any element of the general input system.
Then, we have shown that the capacity region of this setting is the same as 
the capacity region of the above setting when the encoder is limited to a deterministic map. 

As the third contribution, we have proposed a method to convert a code for secure list decoding to 
a protocol for bit-string commitment.
Then, we have shown that this protocol can achieve
the same rate of the message size as the equivocation rate of the original code for secure list decoding. 
This method works even when the input system is a general probability space and 
dishonest Alice can access any element of the input system.
Since many realistic noisy channels have continuous input and output systems,
this result extends the applicability of our method for bit-string commitment.
 
Since the constructed code in this paper is based on random coding,
it is needed to construct practical codes for secure list decoding.
Fortunately, the existing study \cite{Guruswami} systematically constructed several types of 
codes for list decoding with their algorithms.
While their code construction is practical,
in order to use their constructed code for secure list decoding and bit-string commitment,
we need to clarify their security parameters, i.e., 
the equivocation rate and the binding parameter $\delta_D$ for dishonest Alice 
in addition to the decoding error probability $\epsilon_A$.
It is a practical open problem to calculate these security parameters of their codes.

%* Amount of leaked information
%* Further relation with bit commitment (efficient construction)
%* Security parameter of algorithmic construction of list decoding

\section*{Acknowledgments}
The author is grateful to 
Dr. Vincent Tan, Dr. Anurag Anshu, Dr.  Seunghoan Song, and Dr. Naqueeb Warsi  
for helpful discussions and helpful comments.
In particular, Dr. Naqueeb Warsi informed me 
the application of Theorem 11.6.1 of \cite{COVER}.
%In particular, Dr. Vincent Tan suggested Lemma \ref{LA1} and Dr. Anurag Anshu did the relation with bit commitment.
The work reported here was supported in part by 
Guangdong Provincial Key Laboratory (Grant No. 2019B121203002),
Fund for the Promotion of Joint International Research
(Fostering Joint International Research) Grant No. 15KK0007,
the JSPS Grant-in-Aid for Scientific Research 
(A) No.17H01280, (B) No. 16KT0017, 
and Kayamori Foundation of Informational Science Advancement K27-XX-46.

\appendices

\section{Proof of Lemma \ref{LL2}}\Label{A-LL2}
\noindent{\bf Step 1:} Preparation.

\noindent We define the functions
\begin{align}
\bar{\gamma}_1(R_1):=& 
\min_{P \in {\cal P}({\cal U}\times {\cal X})}
\{  I(X;Y|U)_P
| H(X|U)_P=R_1 \} \\
\bar{\gamma}_{1,o}(R_1):=& 
\min_{P \in {\cal P}({\cal X})}
\{  I(X;Y)_P
| H(X)_P=R_1 \} \\
\bar{\gamma}_\alpha(R_1):=& 
\min_{P \in {\cal P}({\cal U}\times {\cal X})}
\{  I_\alpha(X;Y|U)_P
| H(X|U)_P=R_1 \} \\
%\bar{\gamma}_{\alpha,o}(R_1):=& 
%\min_{P \in {\cal P}( {\cal X})}
%\{  I_\alpha(X;Y)_P
%| H(X)_P=R_1 \}  \\
\kappa_1(R_1):= &\max \{ R_3| (R_1,R_3) \in \overline{{\cal C}^{1,3}}\}
\\
\kappa_1^s(R_1) := &\max \{ R_3| (R_1,R_3) \in \overline{{\cal C}^{s,1,3}}\}\\
\kappa_\alpha(R_1) := &\max \{ R_3| (R_1,R_3) \in \overline{{\cal C}_\alpha^{1,3}}\}.
\end{align}
Then, it is sufficient to show the following relations;
\begin{align}
\kappa_1(R_1)&=R_1-\bar{\gamma}_1(R_1) 
=\gamma_1(R_1) \Label{CO1}\\
\kappa_1^s(R_1) &= \max_{R \le R_1}\gamma_1(R) \Label{CO2}\\
\kappa_\alpha(R_1) &=R_1-\bar{\gamma}_\alpha(R_1)= \gamma_\alpha(R_1).\Label{CO3}
\end{align}
Since the second equations in \eqref{CO1} and \eqref{CO3}
follows from the definitions,
it is sufficient to show the first equations in \eqref{CO1} and \eqref{CO3}.
From the definitions, we have
\begin{align}
&\overline{{\cal C}^{1,3}} \nonumber \\
=& 
\bigcup_{P\in {\cal P}( {\cal U}\times {\cal X})} 
\left\{
 (R_1,R_3) \!\left| \!
\begin{array}{l}
0\le R_3 \le R_1- I(X;Y|U)_P, \!\!\!\\
0\le R_1 \le H(X|U)_P 
\end{array}
\right. \right\}  \Label{XZ1} \\
&\overline{{\cal C}^{s,1,3}} \nonumber \\
=&
\bigcup_{P\in {\cal P}({\cal U}\times {\cal X})} 
\left\{
 (R_1,R_3) \left| 
\begin{array}{l}
0 \le  R_3 \le H(X|YU)_P, \\
0\le R_1 \le H(X|U)_P 
\end{array}
\right. \right\} \Label{XZ2} \\
&\overline{{\cal C}_{\alpha}^{1,3}} \nonumber \\
=&
\bigcup_{P\in {\cal P}({\cal U}\times {\cal X})} \!
\left\{
 (R_1,R_3) \!\left| \!
\begin{array}{l}
0 \le R_3 \le R_1- I_\alpha(X;Y|U)_P, \!\!\!\\
0 \le R_1 \le H(X|U)_P 
\end{array}
\right. \right\} .\Label{XZ3}
\end{align}
Hence, \eqref{XZ2} implies \eqref{CO2}.
To show \eqref{CO1} and \eqref{CO3}, 
we derive the following relations from \eqref{XZ1} and \eqref{XZ3}.
\begin{align}
\kappa_1(R_1)&=\max_{R \le R_1}R_1-\bar{\gamma}_1(R) \Label{CD1}\\
%\kappa_1^s(R_1) &= \max_{R \le R_1}R-\bar{\gamma}_1(R)
%\Label{CD2}\\
\kappa_\alpha(R_1) &= \max_{R \le R_1}R_1-\bar{\gamma}_\alpha(R).\Label{CD3}
\end{align}

\noindent{\bf Step 2:} Proof of \eqref{CO1}.

\noindent 
Given $R>0$, we choose 
$P(R):= \argmin_{P \in {\cal P}({\cal X})}
\{  I(X;Y)_P
| H(X)_P=R \} $.
We have $  I(X;Y)_{P(R)}=
D(\bW \times P(R)\| W_{P(R)} \times P(R))
=\sum_{x \in {\cal X}}P(R)(x) D(W_x\| W_{P(R)}) $.
As shown later, when $P(R)(x')> P(R)(x)$, we have
\begin{align}
 D(W_{x'}\| W_{P(R)}) \le D(W_x\| W_{P(R)}).
\Label{ACO}
 \end{align}
We choose $x_1$ and $x_d$ such that
$ D(W_{x_1}\| W_{P(R)}) \le  D(W_{x}\| W_{P(R)}) \le 
 D(W_{x_d}\| W_{P(R)}) $ or $x \in {\cal X}$.
Given $\epsilon>0$, we define the distribution $P(R)_\epsilon$
as
\begin{align}
P(R)_\epsilon(x_1):=& P(R)(x_1)+\epsilon,\\
P(R)_\epsilon(x_d):= &P(R)(x_d)-\epsilon,~
P(R)_\epsilon(x):= P(R)(x)
\end{align}
for $x \neq (x_1,x_d) \in {\cal X}$.
We have $H(P(R)_\epsilon)< H(P(R))=R$.
In particular, 
when $R_o< R$ is sufficiently close to $R$, 
there exists $\epsilon>0$ such that $H(P(R)_\epsilon)=R_0$.
Then,
\begin{align}
&\bar{\gamma}_{1,o}(  R_o) 
=\bar{\gamma}_{1,o}(  H(P(R)_\epsilon)) 
\le  I(X;Y)_{P(R)_\epsilon} \nonumber \\
=&\min_Q
D(\bW \times P(R)_\epsilon\| Q \times P(R)) \nonumber \\
\le &
D(\bW \times P(R)_\epsilon\| W_{P(R)} \times P(R)) \nonumber \\
\le &
D(\bW \times P(R)\| W_{P(R)} \times P(R))\nonumber \\
=& I(X;Y)_{P(R)}
=\bar{\gamma}_{1,o}(  R) .\Label{CM1}
\end{align}
Then, we find that $\bar{\gamma}_{1,o}(  R)$ is monotonically increasing for $R$.

Also, we have
\begin{align}
&\bar{\gamma}_{1}(  R)\nonumber \\
=&\min_{\lambda\in [0,1], R_1,R_2 \in [0, \log d] }
\{\lambda \bar{\gamma}_{1,o}(  R_1)
+(1-\lambda) \bar{\gamma}_{1,o}(  R_2)| (*)\}.\Label{NCP}
\end{align}
where the condition $(*)$ is given as 
$\lambda R_1+ (1-\lambda) R_2=R$.
Since $\bar{\gamma}_{1,o}(  R)$ is monotonically increasing for $R$,
\eqref{NCP} guarantees that
$\bar{\gamma}_{1}(  R)$ is also monotonically increasing for $R$.
Hence, \eqref{CD1} yields \eqref{CO1}, respectively. 

\noindent{\bf Step 3:} Proof of \eqref{ACO}.

\noindent 
Assume that there exist $x\neq x' \in {\cal X}$ such that $P(R)(x')> P(R)(x)$ and the condition \eqref{ACO} does not hold.
We define the distribution $\bar{P}(R)$ as follows.
\begin{align}
\bar{P}(R)(x)&:={P}(R)(x'),~
\bar{P}(R)(x'):={P}(R)(x),\\
\bar{P}(R)(x_o)&:={P}(R)(x_o)
\end{align}
for $x_o(\neq x,x')\in {\cal X} $.
Then, 
\begin{align*}
& I(X;Y)_{\bar{P}(R)}
= 
\min_Q D(\bW \times \bar{P}(R)\|Q \times \bar{P}(R)) \\
\le &  
D(\bW \times \bar{P}(R)\|W_{P(R)} \times \bar{P}(R)) 
\nonumber \\
\le &
D(\bW \times {P}(R)\|W_{P(R)} \times {P}(R)) 
=
I(X;Y)_{{P}(R)},
 \end{align*}
which implies \eqref{ACO}.

\noindent{\bf Step 4:} Proof of \eqref{ACO}.

\noindent 
Instead of $\bar{\gamma}_\alpha(R_1)$
and $\bar{\gamma}_{\alpha,o}(R_1)$, we define
\begin{align}
\bar{\gamma}^p_\alpha(R_1):=& 
\min_{P \in {\cal P}({\cal U}\times {\cal X})}
\{  2^{(\alpha-1)I_\alpha(X;Y|U)_P}
| H(X|U)_P\!=\! R_1\! \} \\
\bar{\gamma}^p_{\alpha,o}(R_1):=& 
\min_{P \in {\cal P}( {\cal X})}
\{  2^{(\alpha-1)I_\alpha(X;Y)_P}
| H(X)_P=R_1 \}  .
\end{align}
Given $R>0$, we choose 
$P_\alpha(R):= \argmin_{P \in {\cal P}({\cal X})}
\{  I_\alpha(X;Y)_{P_\alpha(R)}
| H(X)_P=R \} $.
We choose 
$Q_\alpha(R):= \argmin_{Q \in {\cal P}({\cal Y})}
D_\alpha(\bW \times P_\alpha(R)\| Q \times P_\alpha(R))$.
We have
\begin{align}
 2^{(\alpha-1)I_\alpha(X;Y)_P}
=\sum_{x \in {\cal X}}P(R)(x)
 2^{(\alpha-1)D_\alpha( W_x \| Q_\alpha(R) )}.
\end{align}
In the same way as \eqref{ACO},
when $P_\alpha(R)(x')> P_\alpha(R)(x)$, we have
\begin{align}
 D_\alpha(W_{x'}\| W_{P(R)}) \le D_\alpha(W_x\| W_{P(R)}).
\Label{ACO2}
\end{align}
In the same way as the case with $\bar{\gamma}_{1,o}$, we can show that
$\bar{\gamma}^p_{\alpha,o}(  R)$ is monotonically increasing for $R$.
Hence, 
in the same way as the case with $\bar{\gamma}_{1}$, we can show that
$\bar{\gamma}^p_{\alpha}(  R)$ is monotonically increasing for $R$.
Therefore,
$\bar{\gamma}_{\alpha}(  R)$ is monotonically increasing for $R$.
Hence, \eqref{CD3} yields \eqref{CO3}. 

\section{Proof of Lemma \ref{LL3}}\Label{A-LL3}
The first statement follows from \eqref{BNI}.
The second statement can be shown as follows.
Assume that $\gamma_{\alpha,o}$ is a concave function.
We choose 
\begin{align}
P=
\argmax_{P \in {\cal P}({\cal U}\times {\cal X})}
\{ R_1- I_\alpha(X;Y|U)_P
| H(X|U)_P=R_1 \}.
\end{align}
Then, we have
\begin{align*}
&\gamma_{\alpha}(R_1)\\
=&R_1- I_\alpha(X;Y|U)_P
\stackrel{(a)}{\le} R_1-\sum_{u \in {\cal U}}P_U(u) I_\alpha(X;Y)_{P_{X|U=u}}\\
=& \sum_{u \in {\cal U}}P_U(u)
(H(X)_{P_{X|U=u}}- I_\alpha(X;Y)_{P_{X|U=u}})\\
\stackrel{(b)}{\le}& \sum_{u \in {\cal U}}P_U(u)
\gamma_{\alpha,o}(H(X)_{P_{X|U=u}})
\stackrel{(c)}{\le} \gamma_{\alpha,o}(R_1),
\end{align*}
where 
$(a)$ follows from the concavity of  $x \mapsto - \log x$ 
and the relation $$2^{(\alpha-1)I_\alpha(X;Y|U)_P}=
\sum_{u \in {\cal U}}P_U(u)
 2^{(\alpha-1)I_\alpha(X;Y)_{P_{X|U=u}}}, $$
$(b)$ follows from the definition of $\gamma_{\alpha,o}$, and
$(c)$ follows from the assumption that 
$\gamma_{\alpha,o}$ is a concave function.
Hence, we have $\gamma_{\alpha}(R_1)=\gamma_{\alpha,o}(R_1)$.

\section{Lemma \ref{L4}}\Label{A-L4}
Since we have the Markovian chain
$ Y^{j-1}-(X^{j-1},X_{j+1}, \ldots, X_n)- X_j -Y_j $,
the relation
\begin{align}
I( X^n;Y_j |Y^{j-1})
=( X_{j};Y_j |Y^{j-1})\Label{LLP3-A}
\end{align}
holds.
Hence,
\begin{align}
& I(X^n;Y^n) 
= \sum_{j=1}^n I( X^n;Y_j |Y^{j-1})\nonumber \\
= &\sum_{j=1}^n I( X_{j};Y_j |X^{j-1}),\Label{LLP3-B}
\end{align}
which implies \eqref{LLP3}.
Since we have the Markovian chain
$X_j-X^{j-1}-Y^{j-1} $, 
we have
\begin{align}
&H( X_{j}|Y^{j-1})-
H( X_{j}|X^{j-1})\nonumber \\
=&
H( X_{j}|Y^{j-1})-
H( X_{j}|X^{j-1} Y^{j-1})\nonumber \\
=&
I( X_{j}; X^{j-1}|Y^{j-1})\ge 0.
\end{align}
Thus,
\begin{align}
H(X^n) & = \sum_{j=1}^n H( X_{j}|X^{j-1})
\le \sum_{j=1}^n H( X_{j}|Y^{j-1}),
\Label{LLP2-A}
\end{align}
which implies \eqref{LLP2}.

\section{Proof of Lemma \ref{L3d}}\Label{AP66}
When $s$ is sufficiently large and $\delta>0$ is small, we have
\begin{align}
&G_{P,P}(s)-s (1-\delta)
\nonumber \\
=& \bigg(\sum_{x\in{\cal X}}P(x) \log (2^s P(x)+1-P(x))\bigg) -s (1-\delta)\nonumber \\
=& \bigg(\sum_{x\in{\cal X}}P(x) \Big(s+ \log P(x)+ \log (1+\frac{1-P(x)}{2^s P(x)}) \Big)
\bigg)\nonumber \\
& -s (1-\delta)\nonumber \\
=& \bigg(\sum_{x\in{\cal X}}P(x) \Big(s+ \log P(x)\nonumber \\
&+ \log_e (2)^{-1} \log_e (1+\frac{1-P(x)}{2^s P(x)}) \Big)\bigg) -s (1-\delta)\nonumber \\
\cong& \bigg(\sum_{x\in{\cal X}}P(x) \Big(s+ \log P(x)
+ \frac{1-P(x)}{\log_e (2) 2^s P(x)} \Big)
\bigg) \nonumber \\
&-s (1-\delta)\nonumber \\
= & s- H(P)+ \bigg(\sum_{x\in{\cal X}}\frac{1-P(x)}{2^s \log_e (2)} \bigg) -s (1-\delta)\nonumber \\
= & -H(P)+ \bigg(\frac{|{\cal X}|-1}{2^s \log_e (2)} \bigg) +s \delta.
\end{align}
Under the above approximation, 
the minimum with respect to $s$ is realized when 
$2^s= \frac{|{\cal X}|-1}{\delta}$.
Hence, the minimum is approximated to 
$-H(P)+\delta \log (\frac{e(|{\cal X}|-1) }{\delta})$.
This value goes to $-H(P)$ when $\delta $ goes to $+0$.
Hence, we have \eqref{CY2}. 

Also, we have
\begin{align}
&G_{P,P'}(s)-s r\nonumber \\
=& \bigg(\sum_{x\in{\cal X}}P'(x) \log (2^s P(x)+1-P(x))\bigg) -s r\nonumber \\
=& \bigg(\sum_{x\in{\cal X}}P'(x) \Big(s+ \log P(x)+ \log (1+\frac{1-P(x)}{2^s P(x)}) \Big)
\bigg)\nonumber \\
& -s r\nonumber \\
=& \bigg(\sum_{x\in{\cal X}}P'(x) \Big(s+ \log P(x)\nonumber \\
&+ \log_e (2)^{-1} \log_e (1+\frac{1-P(x)}{2^s P(x)}) \Big)\bigg) -s r.
\end{align}
For each $x$,
the $\Big(s+ \log P(x)+ \log_e (2)^{-1} \log_e (1+\frac{1-P(x)}{2^s P(x)}) \Big)$
is bounded even when $s$ goes to infinity.
Hence, 
we have 
\begin{align}
\lim_{P'\to P} \max_{s>0}G_{P,P'}(s)-s r
= \max_{s>0}G_{P,P}(s)-s r,
\end{align}
which implies \eqref{CY1}. 
\hspace*{\fill}~\QED\par\endtrivlist\unskip

\bibliographystyle{IEEE}

\begin{IEEEbiographynophoto}%[{\includegraphics[width=1in,height=1.25in,clip,keepaspectratio]{Hayashi.jpg}}]
{Masahito Hayashi}(Fellow, IEEE) was born in Japan in 1971.
He received the B.S.\ degree from the Faculty of Sciences in Kyoto
University, Japan, in 1994 and the M.S.\ and Ph.D.\ degrees in Mathematics from
Kyoto University, Japan, in 1996 and 1999, respectively. He worked in Kyoto University as a Research Fellow of the Japan Society of the
Promotion of Science (JSPS) from 1998 to 2000,
and worked in the Laboratory for Mathematical Neuroscience,
Brain Science Institute, RIKEN from 2000 to 2003,
and worked in ERATO Quantum Computation and Information Project,
Japan Science and Technology Agency (JST) as the Research Head from 2000 to 2006.
He worked in the Graduate School of Information Sciences, Tohoku University as Associate Professor from 2007 to 2012.
In 2012, he joined the Graduate School of Mathematics, Nagoya University as Professor.
In 2020, he joined Shenzhen Institute for Quantum Science and Engineering, Southern University of Science and Technology, Shenzhen, China
as Chief Research Scientist.

In 2011, he received Information Theory Society Paper Award (2011) for ``Information-Spectrum Approach to Second-Order Coding Rate in Channel Coding''.
In 2016, he received the Japan Academy Medal from the Japan Academy
and the JSPS Prize from Japan Society for the Promotion of Science.
In 2006, he published the book ``Quantum Information: An Introduction''  from Springer, whose revised version was published as ``Quantum Information Theory: Mathematical Foundation'' from Graduate Texts in Physics, Springer in 2016.
In 2016, he published other two books ``Group Representation for Quantum Theory'' and ``A Group Theoretic Approach to Quantum Information'' from Springer.
He is on the Editorial Board of {\it International Journal of Quantum Information}
and {\it International Journal On Advances in Security}.
His research interests include classical and quantum information theory and classical and quantum statistical inference.
\end{IEEEbiographynophoto}
\end{document}